\documentclass[a4paper,UKenglish,cleveref, autoref, thm-restate, numberwithinsect]{lipics-v2021}

\hideLIPIcs  

\title{Decomposing Finite Languages} 

\author{Daniel Alexander Spenner}{Technische Universität Dortmund, Germany}{daniel.spenner@tu-dortmund.de}{https://orcid.org/0009-0001-2784-5914}{}

\authorrunning{D. A. Spenner} 

\Copyright{Daniel Alexander Spenner} 



\category{} 

\relatedversion{} 

\acknowledgements{I want to thank Thomas Schwentick for his advice and encouragement.}

\nolinenumbers 

\EventEditors{John Q. Open and Joan R. Access}
\EventNoEds{2}
\EventLongTitle{42nd Conference on Very Important Topics (CVIT 2016)}
\EventShortTitle{CVIT 2016}
\EventAcronym{CVIT}
\EventYear{2016}
\EventDate{December 24--27, 2016}
\EventLocation{Little Whinging, United Kingdom}
\EventLogo{}
\SeriesVolume{42}
\ArticleNo{23}

\usepackage{tikz}
\usetikzlibrary{automata, positioning, arrows}
\tikzset{
	node distance=3cm,	
	->,	
	initial text=$ $, 
}

\usepackage[ruled]{algorithm}
\usepackage{algorithmic}
\newcommand{\setDel}{\hspace{2mm} | \hspace{2mm}} 

\newcommand{\natNum}{\mathbb{N}}
\newcommand{\natNumGeq}[1]{\mathbb{N}_{\geq #1}}

\newcommand{\lang}[1]{\mathcal{L}(#1)} 
\newcommand{\ind}[1]{\text{ind}(#1)} 
\newcommand{\size}[1]{|#1|} 
\newcommand{\maxOp}[1]{\max(#1)} 
\newcommand{\complementOp}[1]{\overline{#1}} 
\newcommand{\indexOp}[1]{\text{index}(#1)} 
\newcommand{\qdistOp}[1]{\text{q-dist}(#1)} 

\newcommand{\defHighlight}[1]{\textit{#1}} 

\newcommand{\problemFont}[1]{\textsc{#1}}
\newcommand{\complexityClassFont}[1]{\textsc{#1}}
\newcommand{\myUnderbar}[1]{\underline{#1}}

\newcommand{\primeDFA}[1]{\problemFont{#1Prime-DFA}}
\newcommand{\sPrimeDFA}[1]{\problemFont{#1S-Prime-DFA}}

\newcommand{\primeDFAfin}[1]{\problemFont{#1Prime-DFA}_\text{fin}}
\newcommand{\minimalDFA}[1]{\problemFont{#1Minimal-DFA}}
\newcommand{\emptyDFA}[1]{\problemFont{#1Empty-DFA}}
\newcommand{\emptyDFAfin}[1]{\problemFont{#1Empty-DFA}_\text{fin}}
\newcommand{\emptyDFAq}[1]{\problemFont{#1Empty-DFA}_\text{q+}}

\theoremstyle{definition}
\newtheorem{definition2}[theorem]{Definition}

\begin{document}

\maketitle

\begin{abstract}
The paper completely characterizes the primality of acyclic DFAs, where a DFA $\mathcal{A}$ is \defHighlight{prime} if there do not exist DFAs $\mathcal{A}_1,\dots,\mathcal{A}_t$ with $\lang{\mathcal{A}} = \bigcap_{i=1}^{t} \lang{\mathcal{A}_i}$ such that each $\mathcal{A}_i$ has strictly less states than the minimal DFA recognizing the same language as $\mathcal{A}$.
A regular language is prime if its minimal DFA is prime. Thus, this result also characterizes the primality of finite languages.

Further, the \complexityClassFont{NL}-completeness of the corresponding decision problem $\primeDFAfin{}$ is proven.
The paper also characterizes the primality of acyclic DFAs under two different notions of compositionality, union and union-intersection compositionality.

Additionally, the paper introduces the notion of \defHighlight{S-primality}, where a DFA $\mathcal{A}$ is S-prime if there do not exist DFAs $\mathcal{A}_1,\dots,\mathcal{A}_t$ with $\lang{\mathcal{A}} = \bigcap_{i=1}^{t} \lang{\mathcal{A}_i}$ such that each $\mathcal{A}_i$ has strictly less states than $\mathcal{A}$ itself.
It is proven that the problem of deciding S-primality for a given DFA is \complexityClassFont{NL}-hard. To do so, the \complexityClassFont{NL}-completeness of \minimalDFA{2}, the basic problem of deciding minimality for a DFA with at most two letters, is proven.
\end{abstract}

\section{Introduction}
\label{sec:introduction}
Under intersection compositionality a deterministic finite automaton (DFA) $\mathcal{A}$ is \defHighlight{composite} if there exist DFAs $\mathcal{A}_1,\dots,\mathcal{A}_t$ with $\lang{\mathcal{A}} = \bigcap_{i=1}^{t} \lang{\mathcal{A}_i}$ such that the size of each $\mathcal{A}_i$ is smaller than the index of $\mathcal{A}$. Otherwise, $\mathcal{A}$ is \defHighlight{prime} \cite{DBLP:journals/iandc/KupfermanM15}. 
The index of $\mathcal{A}$ is the size of the minimal DFA recognizing the same language as $\mathcal{A}$.
\primeDFA{} denotes the problem of deciding primality for a given DFA. $\primeDFAfin{}$ denotes the restriction of \primeDFA{} to DFAs recognizing a finite language.

Compositionality in general is a key concept in both practical and theoretical computer science \cite{DBLP:conf/compos/1997,DBLP:journals/pieee/Tripakis16}.
The intersection decomposition of finite automata can be motivated by LTL model checking as well as automaton identification. Both will be briefly discussed below.

The notion of intersection compositionality of finite automata was introduced in \cite{DBLP:journals/iandc/KupfermanM15},
while a limitation of this notion was already studied in \cite{DBLP:conf/sofsem/GaziR08}.
Surprisingly, \cite{DBLP:journals/iandc/KupfermanM15} found even the complexity of the basic problem \primeDFA{} to be open.
They proved that \primeDFA{} is in \complexityClassFont{ExpSpace} and is \complexityClassFont{NL}-hard. So far, this doubly exponential gap has not been closed.

Given the difficulties in tackling the general problem, it has proven fruitful to characterize the intersection compositionality of fragments of the regular languages \cite{DBLP:journals/iandc/KupfermanM15,DBLP:conf/mfcs/JeckerKM20,DBLP:conf/concur/JeckerM021}. Our study joins this line of research by completely characterizing the intersection compositionality of acyclic DFAs (ADFA) and thereby of finite languages. Further, we prove the \complexityClassFont{NL}-completeness of $\primeDFAfin{}$ and characterize the compositionality of finite languages under two different notions of compositionality suggested in \cite{DBLP:journals/iandc/KupfermanM15}, union and union-intersection compositionality.

Additionally, we present a proof of the \complexityClassFont{NL}-completeness of the basic problem \minimalDFA{2}, the problem of deciding minimality for a DFA with at most two letters. 
For arbitrary alphabets, the \complexityClassFont{NL}-hardness is a folklore result that seemingly has not been explicitly published but follows from the constructions in \cite{DBLP:journals/iandc/ChoH92}, while the \complexityClassFont{NL}-hardness of \minimalDFA{2} appears to be new \cite{FernauHolzerPersonalCommunication}. We use this result to establish complexity boundaries for \sPrimeDFA{}, a modification of \primeDFA{} using the size of the given DFA, not its index.

\subparagraph*{Related Work}
The notion of intersection compositionality was introduced in \cite{DBLP:journals/iandc/KupfermanM15}, where the aforementioned complexity boundaries were established. They already considered language fragments, analyzing safety DFAs and permutation DFAs. This line of research was followed up in \cite{DBLP:conf/mfcs/JeckerKM20,DBLP:conf/concur/JeckerM021}, which focused on unary DFAs and permutation DFAs, respectively.

The intersection decomposition of automata can be motivated by LTL model checking, where the validity of a specification, given as an LTL formula, is checked for a system. 
The automata-based approach entails translating the specification into a finite automaton \cite{DBLP:conf/lics/VardiW86}.
Since the LTL model checking problem is \complexityClassFont{PSpace}-complete in the size of the LTL formula \cite{DBLP:books/daglib/0020348}, it is desirable to decompose the formula into a conjunction of subformulas. This can also be understood as decomposing the finite automaton corresponding to the formula.

Another application of intersection decomposition arises in the field of automaton identification. The basic task here is, given a set of labeled words, to construct a finite automaton conforming to this set \cite{DBLP:journals/iandc/Gold78}. An interesting approach is to construct multiple automata instead of one, which can lead to smaller and more intuitive solutions \cite{DBLP:conf/fmcad/LaufferYVSS22}.

An alternative notion of compositionality uses concatenation. Here, a language $L$ is composite if there exist two non-trivial languages $L_1,L_2$ with $L = L_1 L_2$.
The concatenation primality problem for regular languages is \complexityClassFont{PSpace}-complete \cite{DBLP:conf/pods/MartensNS10}. 
The restriction to finite languages is known to be \complexityClassFont{NP}-hard \cite{DBLP:journals/corr/abs-1902-06253}, while the conjectured \complexityClassFont{NP}-completeness of this restriction remains open \cite{DBLP:conf/dlt/SalomaaY99,DBLP:journals/actaC/MateescuSY02,DBLP:journals/igpl/Wieczorek10}.

\subparagraph*{Contributions}
In \cref{sec:fl_characterization} we completely characterize the intersection compositionality of ADFAs and thereby of finite languages. 
We expand on this by proving the \complexityClassFont{NL}-completeness of $\primeDFAfin{}$ in \cref{sec:fl_complexity}, thus showing that finite languages are significantly easier to handle under intersection compositionality than under concatenation compositionality.
We characterize the union and union-intersection compositionality of finite languages in \cref{sec:fl_differentNotionsOfCompositionality}, where we also prove the existence of languages that are union-intersection composite but both union prime and intersection prime.

In \cref{sec:2DFAMinimalAndSPrimeDFA} we introduce the problem \sPrimeDFA{}, which is analogous to \primeDFA{} but uses the size for the definition of compositionality, not the index. We prove that \sPrimeDFA{} is in \complexityClassFont{ExpSpace} and is \complexityClassFont{NL}-hard. We also prove these boundaries for \primeDFA{2} and \sPrimeDFA{2}, the restrictions of the respective problems to DFAs with at most two letters.
To establish these boundaries we prove the \complexityClassFont{NL}-completeness of \minimalDFA{2}.

Detailed proofs of these results are provided in the appendix.
\section{Preliminaries}
\label{sec:preliminaries}
A \defHighlight{deterministic finite automaton} (DFA) is a $5$-tuple $\mathcal{A} = (Q,\Sigma,q_I,\delta,F)$, where $Q$ is a finite set of states, $\Sigma$ is a finite non-empty alphabet, $q_I \in Q$ is an initial state, $\delta: Q \times \Sigma \rightarrow Q$ is a transition function, and $F \subseteq Q$ is a set of accepting states. As usual, we extend $\delta$ to words: $\delta: Q \times \Sigma^* \rightarrow Q$ with $\delta(q,\varepsilon) = q$ and $\delta(q,\sigma_1\dots\sigma_n) = \delta(\delta(q,\sigma_1\dots\sigma_{n-1}),\sigma_n)$. For $q \in Q$, the DFA $\mathcal{A}^q$ is constructed out of $\mathcal{A}$ by setting $q$ as the initial state, thus $\mathcal{A}^q = (Q,\Sigma,q,\delta,F)$.

The \defHighlight{run} of $\mathcal{A}$ on a word $w = \sigma_1\dots\sigma_n$ starting in state $q$ is the sequence $q_0,\sigma_1,q_1,\dots,\sigma_n,q_n$ with $q_0 = q$ and $q_i = \delta(q_{i-1},\sigma_i)$ for each $i \in \{1,\dots,n\}$. The \defHighlight{initial run} of $\mathcal{A}$ on $w$ is the run of $\mathcal{A}$ on $w$ starting in $q_I$. 
The run of $\mathcal{A}$ on $w$ starting in $q$ is \defHighlight{accepting} if $q_n \in F$, otherwise it is \defHighlight{rejecting}.
The DFA $\mathcal{A}$ \defHighlight{accepts} $w$ if the initial run of $\mathcal{A}$ on $w$ is accepting. Otherwise, it \defHighlight{rejects} $w$.
The \defHighlight{language} $\lang{\mathcal{A}}$ of $\mathcal{A}$ is the set of words accepted by $\mathcal{A}$. We say that $\mathcal{A}$ \defHighlight{recognizes} $\lang{\mathcal{A}}$. A language is \defHighlight{regular} if there exists a DFA recognizing it. Since we only consider regular languages, we use the terms language and regular language interchangeably.

The \defHighlight{size} $\size{\mathcal{A}}$ of $\mathcal{A}$ is the number of states in $Q$. The DFA $\mathcal{A}$ is \defHighlight{minimal} if $\lang{\mathcal{A}} \neq \lang{\mathcal{B}}$ holds for every DFA $\mathcal{B}$ with $\size{\mathcal{B}} < \size{\mathcal{A}}$. It is well known that for every regular language $L$ there exists a canonical minimal DFA recognizing $L$. The \defHighlight{index} $\ind{L}$ of $L$ is the size of this canonical minimal DFA. The index of $\mathcal{A}$ is the index of the language recognized by $\mathcal{A}$, thus $\ind{\mathcal{A}} = \ind{\lang{\mathcal{A}}}$. Note that $\mathcal{A}$ is minimal iff $\size{\mathcal{A}} = \ind{\mathcal{A}}$.

We borrow a few terms from graph theory. Let $q_0,\sigma_1,q_1,\dots,\sigma_n,q_n$ be the run of $\mathcal{A}$ on $w = \sigma_1\dots\sigma_n$ starting in $q_0$. Then $q_0,\dots,q_n$ is a \defHighlight{path} in $\mathcal{A}$ from $q_0$ to $q_n$. The \defHighlight{length} of this path is $n$. Thus, for two states $q,q'$ there exists a path from $q$ to $q'$ in $\mathcal{A}$ of length $n$ iff there exists a $w \in \Sigma^n$ with $\delta(q,w) = q'$. The state $q'$ is \defHighlight{reachable from} $q$ if there exists a path from $q$ to $q'$. Otherwise, $q'$ is \defHighlight{unreachable from} $q$. 
Obviously, if $q'$ is reachable from $q$ then there exists a path from $q$ to $q'$ of a length strictly smaller than $\size{\mathcal{A}}$.
We say that $q'$ is \defHighlight{reachable} if it is reachable from $q_I$. Otherwise, it is \defHighlight{unreachable}. A \defHighlight{cycle} in $\mathcal{A}$ is a path $q_0,\dots,q_n$ in $\mathcal{A}$ where $q_0 = q_n$ and $n \in \natNumGeq{1}$. The DFA $\mathcal{A}$ is \defHighlight{acyclic} (ADFA) if every cycle in $\mathcal{A}$ begins in a rejecting sink. Clearly, a DFA recognizes a finite language iff its minimal DFA is acyclic.

We call a DFA $\mathcal{A} = (Q,\Sigma,q_I,\delta,F)$ \defHighlight{linear} if for every $q,q' \in Q$ with $q \neq q'$ either $q'$ is reachable from $q$ or $q$ is reachable from $q'$, but not both. Thus, in a linear DFA reachability induces a linear order over the states.
Obviously, every linear DFA has exactly one sink. Furthermore, a minimal ADFA $\mathcal{A}$ is linear iff $\size{\mathcal{A}} = n+2$, where $n$ is the length of the longest word in $\lang{\mathcal{A}}$.

Consider a word $w = \sigma_1 \dots \sigma_n \in \Sigma^n$. A word $wv$ with $v \in \Sigma^+$ is an \defHighlight{extension} of $w$. A word $\sigma_1\dots\sigma_i\sigma_{i+l}\dots\sigma_n$ with $i \in \{0,\dots,n-2\}, l \in \{2,\dots,n-i\}$ is a \defHighlight{compression} of $w$. An ADFA $\mathcal{A}$ has the \defHighlight{compression-extension-property} (CEP) if for every $w \in \lang{\mathcal{A}}$ with $|w|=n$, where $n$ is the length of the longest word in $\lang{\mathcal{A}}$, there exists a compression $w'$ of $w$ such that every extension of $w'$ is rejected by $\mathcal{A}$.

We introduce a type of DFA already inspected in \cite{DBLP:journals/iandc/KupfermanM15}.
A regular language $L \subseteq \Sigma^*$ is a \defHighlight{safety language} if $w \notin L$ implies $wy \notin L$ for every $y \in \Sigma^*$.
A DFA $\mathcal{A}$ is a \defHighlight{safety DFA} if $\lang{\mathcal{A}}$ is a safety language.
A regular language $L \subseteq \Sigma^*$ is a \defHighlight{co-safety language} if the complement language $\complementOp{L}$ of $L$ is a safety language. 
A DFA $\mathcal{A}$ is a \defHighlight{co-safety DFA} if $\lang{\mathcal{A}}$ is a co-safety language.
Clearly, every non-trivial minimal safety DFA has exactly one rejecting state, and this state is a sink. Conversely, every non-trivial minimal co-safety DFA has exactly one accepting state, and this state is a sink.

We introduce the notions intersection compositionality and primality of DFAs and languages, following the definitions in \cite{DBLP:journals/iandc/KupfermanM15}:
\begin{definition2}
	\label{def:compositionality}
	For $k \in \natNumGeq{1}$, a DFA $\mathcal{A}$ is \defHighlight{$k$-decomposable} if there exist DFAs $\mathcal{A}_1,\dots,\mathcal{A}_t$ with $\lang{\mathcal{A}} = \bigcap_{i=1}^{t} \lang{\mathcal{A}_i}$ and $\size{\mathcal{A}_i} \leq k$ for each $i \in \{1,\dots,t\}$, where $t \in \natNumGeq{1}$. We call such DFAs $\mathcal{A}_1,\dots,\mathcal{A}_t$ a \defHighlight{$k$-decomposition} of $\mathcal{A}$.
	We call $\mathcal{A}$ \defHighlight{composite} if $\mathcal{A}$ is $k$-decomposable for a $k < \ind{\mathcal{A}}$, that is, if it is $(\ind{\mathcal{A}}-1)$-decomposable. Otherwise, we call $\mathcal{A}$ \defHighlight{prime}.\lipicsEnd
\end{definition2}
We use compositionality or $\cap$-compositionality when referring to intersection compositionality.

When analyzing the compositionality of a given DFA $\mathcal{A}$, it is sufficient to consider minimal DFAs $\mathcal{B}$ strictly smaller than the minimal DFA of $\mathcal{A}$ with $\lang{\mathcal{A}} \subseteq \lang{\mathcal{B}}$. Thus, we define $\alpha(\mathcal{A}) = \{\mathcal{B} \setDel \mathcal{B} \text{ is a minimal DFA with } \ind{\mathcal{B}}<\ind{\mathcal{A}} \text{ and } \lang{\mathcal{A}} \subseteq \lang{\mathcal{B}}\}$. Obviously, the DFA $\mathcal{A}$ is composite iff $\lang{\mathcal{A}} = \bigcap_{\mathcal{B} \in \alpha(\mathcal{A})} \lang{\mathcal{B}}$. We call a word $w \in (\bigcap_{\mathcal{B} \in \alpha(\mathcal{A})} \lang{\mathcal{B}}) \setminus \lang{\mathcal{A}}$ a \defHighlight{primality witness} of $\mathcal{A}$. Clearly, the DFA $\mathcal{A}$ is composite iff $\mathcal{A}$ has no primality witness.

We extend the notions of $k$-decompositions, compositionality, primality and primality witnesses to regular languages by identifying a regular language with its minimal DFA.

We denote the problem of deciding primality for a given DFA with \primeDFA{}.
We denote the restriction of \primeDFA{} to DFAs recognizing a finite languages with $\primeDFAfin{}$.
\primeDFA{} is in \complexityClassFont{ExpSpace} and is \complexityClassFont{NL}-hard \cite{DBLP:journals/iandc/KupfermanM15}.

We denote the connectivity problem in directed graphs, which is \complexityClassFont{NL}-complete \cite{DBLP:books/daglib/0095988}, with \problemFont{STCON}. We denote the restriction of \problemFont{STCON} to graphs with a maximum outdegree of two with \problemFont{2STCON}. Clearly, \problemFont{2STCON} is \complexityClassFont{NL}-complete as well.
We denote the problem of deciding minimality for a given DFA with \minimalDFA{}.
For $k \in \natNumGeq{2}$, the problem \minimalDFA{k} is the restriction of \minimalDFA{} to DFAs with at most $k$ letters.
As mentioned in \cref{sec:introduction}, the \complexityClassFont{NL}-completeness of \minimalDFA{k} for $k \in \natNumGeq{3}$ is folklore, while the \complexityClassFont{NL}-hardness of \minimalDFA{2} appears to be open.
\section{Compositionality of Finite Languages}
\label{sec:fl_characterization}
We characterize the compositionality of ADFAs and thereby of finite languages by proving:
\begin{restatable}{theorem}{theFlCharacterization}
	\label{the:fl_characterization}
	Consider a minimal ADFA $\mathcal{A} = (Q,\Sigma,q_I,\delta,F)$ recognizing a non-empty language. Then $\mathcal{A}$ is prime iff $\mathcal{A}$ is linear and:
	\begin{romanenumerate}
		\item $\sigma^n \in \lang{\mathcal{A}}$ for some $\sigma \in \Sigma$, where $n \in \natNum$ is the length of the longest word in $\lang{\mathcal{A}}$, or
		\item $\mathcal{A}$ is a safety DFA and $\mathcal{A}$ does not have the CEP.\lipicsEnd
	\end{romanenumerate}
	\vspace{1.9ex}
\end{restatable}

To prove \cref{the:fl_characterization} we will consider five cases in turn.

First, if the ADFA $\mathcal{A}$ is not linear we essentially have a surplus of states, allowing us to construct one DFA rejecting overlong words and one specific DFA for each of the remaining words also rejected by $\mathcal{A}$.
This approach fails with linear ADFAs.
Nevertheless, we will come back to the idea of excluding words longer than a threshold value and tailoring a DFA for each word shorter than the threshold value which has to be rejected as well.

Second, if $\mathcal{A}$ is linear and $\sigma^n \in \lang{\mathcal{A}}$ holds the DFAs in $\alpha(\mathcal{A})$ do not possess enough states to differentiate the words $\sigma^0,\dots,\sigma^n$ but have to accept $\sigma^n$, which implies cyclic behavior on the words in $\{\sigma\}^*$ from which primality follows.

Third, if there is no $\sigma \in \Sigma$ with $\sigma^n \in \lang{\mathcal{A}}$ and $\mathcal{A}$ is not a safety DFA we can return to the idea of excluding words longer than a threshold value. For each of the words left to reject, it is possible to construct a DFA similar to $\mathcal{A}$ but without the rejecting sink, which circles back to the rejecting non-sink.

Fourth, if there is no $\sigma \in \Sigma$ with $\sigma^n \in \lang{\mathcal{A}}$ and $\mathcal{A}$ has the CEP we can utilize DFAs similar to $\mathcal{A}$ possessing a rejecting sink, since the CEP allows us to skip over one state.

Fifth and finally, if $\mathcal{A}$ is linear and $\mathcal{A}$ is a safety DFA and does not have the CEP both of the above approaches fail. There is no state to circle back to, and for the word breaching the CEP skipping over states is not possible either, which implies primality.

Formalizing these five cases, we get:
\begin{restatable}{claim}{claFlCharacterization}
	\label{cla:fl_characterization}
	Consider a minimal ADFA $\mathcal{A} = (Q,\Sigma,q_I,\delta,F)$ recognizing a non-empty language. Let $n \in \natNum$ be the length of the longest word in $\lang{\mathcal{A}}$. The following assertions hold:
	\begin{alphaenumerate}
		\item $\mathcal{A}$ is composite if $\mathcal{A}$ is not linear.\label{cla_ass:fl_characterization_non-linear}
		\item $\mathcal{A}$ is prime if $\mathcal{A}$ is linear and $\sigma^n \in \lang{\mathcal{A}}$ holds for some $\sigma \in \Sigma$.\label{cla_ass:fl_characterization_linear+sigmaN}
		\item $\mathcal{A}$ is composite if there is no $\sigma \in \Sigma$ with $\sigma^n \in \lang{\mathcal{A}}$ and $\mathcal{A}$ is not a safety DFA.\label{cla_ass:fl_characterization_non-sigmaN+non-safety}
		\item $\mathcal{A}$ is composite if there is no $\sigma \in \Sigma$ with $\sigma^n \in \lang{\mathcal{A}}$ and $\mathcal{A}$ has the CEP.\label{cla_ass:fl_characterization_non-sigmaN+CEP}
		\item $\mathcal{A}$ is prime if $\mathcal{A}$ is linear and $\mathcal{A}$ is a safety DFA and $\mathcal{A}$ does not have the CEP.\label{cla_ass:fl_characterization_linear+safety+non-CEP}\lipicsEnd
	\end{alphaenumerate}
	\vspace{1.9ex}
\end{restatable}

Formalizing the intuition given above for (\ref{cla_ass:fl_characterization_non-linear}) and (\ref{cla_ass:fl_characterization_linear+sigmaN}) is not too complex. Assertions (\ref{cla_ass:fl_characterization_non-sigmaN+non-safety})-(\ref{cla_ass:fl_characterization_linear+safety+non-CEP}) prove to be much harder. Thus, we commence by discussing (\ref{cla_ass:fl_characterization_non-sigmaN+non-safety}) in \cref{subsec:fl_linearNonSafetyDFAs} and (\ref{cla_ass:fl_characterization_non-sigmaN+CEP}) and (\ref{cla_ass:fl_characterization_linear+safety+non-CEP}) in \cref{subsec:fl_linearSafetyDFAs}. Henceforth, we consider a minimal ADFA $\mathcal{A} = (Q,\Sigma,q_I,\delta,F)$ recognizing the non-empty language $L$ with $\sigma^n \notin L$ for each $\sigma \in \Sigma$, where $n \in \natNum$ is the length of the longest word in $L$. W.l.o.g. we assume $Q = \{q_0,\dots,q_{n+1}\}$ with $q_j$ being reachable from $q_i$ for all $i < j$, which implies $q_I = q_0$ and $q_n \in F$ with $q_{n+1}$ being the rejecting sink. Finally, we define $\Sigma_{i,j} = \{\sigma \in \Sigma \setDel \delta(q_i,\sigma) = q_j\}$.

\subsection{Linear non-safety ADFAs}
\label{subsec:fl_linearNonSafetyDFAs}
We consider \cref{cla:fl_characterization} (\ref{cla_ass:fl_characterization_non-sigmaN+non-safety}). Therefore, we assume that $\mathcal{A}$ is not a safety DFA, which implies $\{q_n\} \subseteq F \subset Q \setminus \{q_{n+1}\}$. Let $d \in \{0,\dots,n-1\}$ with $q_d \notin F$.

We show the compositionality of $\mathcal{A}$ by specifying an $(n+1)$-decomposition of $\mathcal{A}$. First, we construct DFAs rejecting words not in $L$ that are not extensions of words $u \in L, |u| = n$. Afterwards, we turn to such extensions, whose handling poses the main difficulty. Here, we first construct DFAs rejecting such extensions that are longer than a certain threshold value. For the remaining extensions we employ the idea of circling back to $q_d$.

We begin by considering words not in $L$ which are not extensions of words $u \in L, |u|=n$. We introduce three DFA types handling these words.

First, let $\mathcal{A}_0$ be the DFA constructed out of $\mathcal{A}$ by removing $q_n$, redirecting every transition $q \rightarrow q_n$ to $q_0$, and including $q_0$ into the acceptance set. Clearly, $\mathcal{A}_0 \in \alpha(\mathcal{A})$ and $\mathcal{A}_0$ rejects every $w \notin L$ on which $\mathcal{A}$ enters the rejecting sink prematurely, that is, without entering $q_n$.

Second, let $\hat{\mathcal{A}}_d$ be the DFA constructed out of $\mathcal{A}$ by removing $q_{n+1}$, redirecting every transition $q_i \rightarrow q_{n+1}$ with $i < n$ to $q_n$ and every transition $q_n \rightarrow q_{n+1}$ to $q_d$. Clearly, $\hat{\mathcal{A}}_d \in \alpha(\mathcal{A})$ and $\hat{\mathcal{A}}_d$ rejects every $w \notin L$ on which $\mathcal{A}$ does not enter the rejecting sink.

Third, we construct DFAs rejecting extensions of words $w \in L, |w| < n$ with $\delta(q_0,w)=q_n$. Let $I = \{0,\dots,n\}$. For each $m \in \{1,\dots,n-1\}$ let $I_m = \{(i_0,\dots,i_m) \in I^{m+1} \setDel 0 = i_0 < \dots < i_m = n\}$. For each $\myUnderbar{i} \in I_m$ define $\mathcal{A}_{\myUnderbar{i}}$ as in \cref{subfig:fl_A_myUnderbariA_m<n-1,subfig:fl_A_myUnderbariA_m=n-1}. It is easy to confirm that each $\mathcal{A}_{\myUnderbar{i}}$ is in $\alpha(\mathcal{A})$ and rejects extensions of words on which $\mathcal{A}$ visits the states $q_{i_0},\dots,q_{i_m}$.
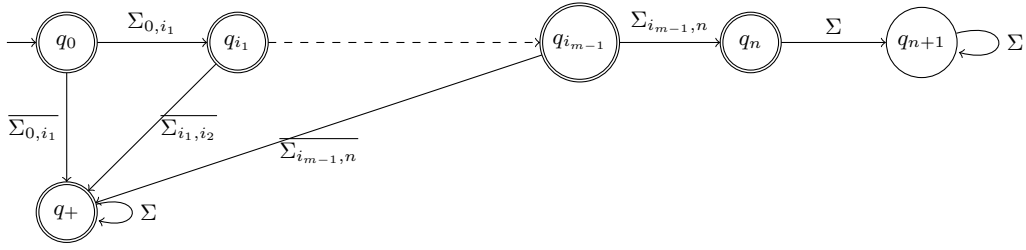
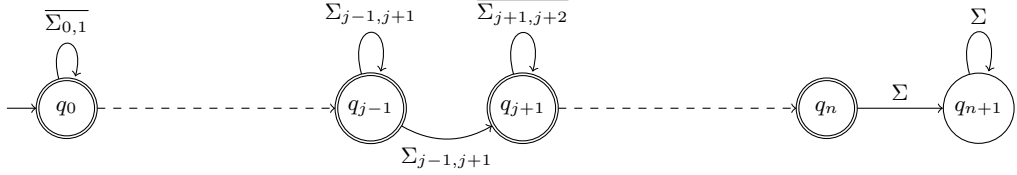
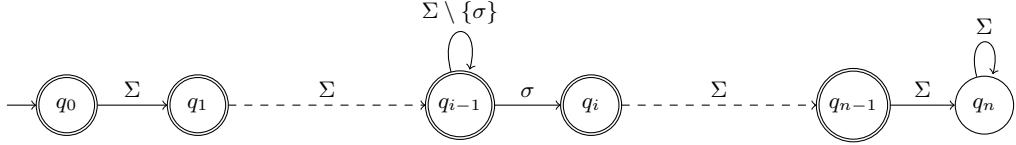
\begin{figure}[t]
	\begin{subfigure}[t]{0.75\textwidth}
		\centering
		\begin{tikzpicture}[node distance=2.25cm]
		\footnotesize
		\node[state, initial, accepting] 					(q0) 	{$q_0$};
		\node[state, right of=q0, accepting] 				(qi1) 	{$q_{i_1}$};
		\node[state, right of=qi1, draw=none]				(empty-node) 	{};
		\node[state, right of=empty-node, accepting] 		(qim-1) {$q_{i_{m-1}}$};
		\node[state, right of=qim-1, accepting] 			(qn) {$q_n$};
		\node[state, right of=qn] 							(qn+1) {$q_{n+1}$};
		\node[state, below of=q0, accepting] 				(q+) 	{$q_+$};
		
		\draw	(q0)	edge[above]						node{$\Sigma_{0,i_1}$}								(qi1);
		\draw	(q0)	edge[left]						node{$\complementOp{\Sigma_{0,i_1}}$}								(q+);
		\draw	(qi1)	edge[dashed]					node{}								(qim-1);
		\draw	(qi1)	edge[right]						node{$\complementOp{\Sigma_{i_1,i_2}}$}								(q+);
		\draw	(qim-1)	edge[above]						node{$\Sigma_{i_{m-1},n}$}								(qn);
		\draw	(qim-1)	edge[below]						node{$\complementOp{\Sigma_{i_{m-1},n}}$}								(q+);
		\draw	(qn)	edge[above]						node{$\Sigma$}								(qn+1);
		\draw	(qn+1)	edge[loop right]				node{$\Sigma$}						(qn+1);
		\draw	(q+)	edge[loop right]				node{$\Sigma$}						(q+);
		\end{tikzpicture}
		\caption{$\mathcal{A}_{\myUnderbar{i}}$ if $m < n-1$.}
		\label{subfig:fl_A_myUnderbariA_m<n-1}
	\end{subfigure}\hfill
	\begin{subfigure}[t]{0.75\textwidth}
		\centering
		\begin{tikzpicture}[node distance=2cm]
		\footnotesize
		\node[state, initial, accepting] 					(q0) 	{$q_0$};
		\node[state, right of=q0, draw=none]				(empty-node1) 	{};
		\node[state, right of=empty-node1, accepting] 		(qj-1) {$q_{j-1}$};
		\node[state, right of=qj-1, accepting] 				(qj+1) {$q_{j+1}$};
		\node[state, right of=qj+1, draw=none]				(empty-node2) 	{};
		\node[state, right of=empty-node2, accepting] 		(qn) {$q_n$};
		\node[state, right of=qn] 							(qn+1) {$q_{n+1}$};
		
		\draw	(q0)	edge[dashed]						node{}								(qj-1);
		\draw	(q0)	edge[loop above]					node{$\complementOp{\Sigma_{0,1}}$}	(q0);
		\draw	(qj-1)	edge[below, bend right]							node{$\Sigma_{j-1,j+1}$}			(qj+1);
		\draw	(qj-1)	edge[loop above]					node{$\complementOp{\Sigma_{j-1,j+1}}$}	(qj-1);
		\draw	(qj+1)	edge[dashed]						node{}								(qn);
		\draw	(qj+1)	edge[loop above]					node{$\complementOp{\Sigma_{j+1,j+2}}$}	(qj-1);
		\draw	(qn)	edge[above]							node{$\Sigma$}						(qn+1);
		\draw	(qn+1)	edge[loop above]					node{$\Sigma$}						(qn+1);
		\end{tikzpicture}
		\caption{$\mathcal{A}_{\myUnderbar{i}}$ if $m = n-1$, where $\myUnderbar{i} = (0,\dots,j-1,j+1,\dots,n)$.}
		\label{subfig:fl_A_myUnderbariA_m=n-1}
	\end{subfigure}\hfill
	\begin{subfigure}[t]{0.75\textwidth}
		\centering
		\begin{tikzpicture}[node distance=1.725cm]
		\footnotesize
		\node[state, initial, accepting] 					(q0) 	{$q_0$};
		\node[state, right of=q0, accepting] 				(q1) 	{$q_1$};
		\node[state, right of=q1, draw=none]				(empty-node1) 	{};
		\node[state, right of=empty-node1, accepting] 		(qi-1) 	{$q_{i-1}$};
		\node[state, right of=qi-1, accepting] 				(qi) 	{$q_{i}$};
		\node[state, right of=qi, draw=none]				(empty-node2) 	{};
		\node[state, right of=empty-node2, accepting]		(qn-1) 	{$q_{n-1}$};
		\node[state, right of=qn-1] 						(qn) 	{$q_n$};
		
		\draw	(q0)	edge[above]							node{$\Sigma$}					(q1);
		\draw	(q1)	edge[above,dashed]					node{$\Sigma$}					(qi-1);
		\draw	(qi-1)	edge[loop above]					node{$\Sigma\setminus\{\sigma\}$}(qi-1);
		\draw	(qi-1)	edge[above]							node{$\sigma$}					(qi);
		\draw	(qi)	edge[above,dashed]					node{$\Sigma$}					(qn-1);
		\draw	(qn-1)	edge[above]							node{$\Sigma$}					(qn);
		\draw	(qn)	edge[loop above]					node{$\Sigma$}					(qn);
		\end{tikzpicture}
		\caption{$\mathcal{A}_{\sigma,i}$.}
		\label{subfig:fl_A_sigmai}
	\end{subfigure}
	\caption{DFA $\mathcal{A}_{\myUnderbar{i}}$ for $\myUnderbar{i} \in I_m$ with $m \in \{1,\dots,n-1\}$ and DFA $\mathcal{A}_{\sigma,i}$ for $\sigma \in \Sigma, i \in \{1,\dots,n\}$.}
	\label{fig:fl_A_myUnderbariA_sigmai}
\end{figure}

\cref{lem:fl_A_0A_dA_myUnderbari} formalizes the results concerning $\mathcal{A}_0$, $\hat{\mathcal{A}}_d$ and $\mathcal{A}_{\myUnderbar{i}}$:
\begin{restatable}{lemma}{lemFlANULLAdAMyUnderbari}
	\label{lem:fl_A_0A_dA_myUnderbari}
	The following assertions hold:
	\begin{romanenumerate}
		\item $\mathcal{A}_0,\hat{\mathcal{A}}_d,\mathcal{A}_{\myUnderbar{i}} \in \alpha(\mathcal{A})$, where $\myUnderbar{i} \in \bigcup_{m =1}^{n-1} I_m$.
		\item Consider a word $w \notin L$, where $w$ is not an extension of a word $u \in L, |u|=n$. Then $w \notin \lang{\mathcal{A}_0} \cap \lang{\hat{\mathcal{A}}_d} \cap \bigcap_{m=1}^{n-1}\bigcap_{\myUnderbar{i} \in I_m}\lang{\mathcal{A}_{\myUnderbar{i}}}$ holds.\lipicsEnd
	\end{romanenumerate}
	\vspace{1.9ex}
\end{restatable}

Next, we turn to the extensions of words $u \in L, |u| = n$. We begin by constructing DFAs that taken together reject every word strictly longer than $n + (n-2)$. Then we turn to the remaining extensions one by one, of which only a finite number are left to reject.

Let $\sigma \in \Sigma$. Since $\sigma^n \notin L$, there exists a value $i \in \{1,\dots,n\}$ with $\sigma \notin \Sigma_{i-1,i}$. Define $\mathcal{A}_{\sigma,i}$ as in \cref{subfig:fl_A_sigmai}. First, note that $\mathcal{A}_{\sigma,i} \in \alpha(\mathcal{A})$ because a word rejected by $\mathcal{A}_{\sigma,i}$ is strictly longer $n$ or is of length $n$ with letter $\sigma$ at position $i$.
Next, consider a word $w = \sigma_1\dots\sigma_m \in \Sigma^m$ such that $\sigma_j = \sigma$ for a $j \in \{1,\dots,m\}$ with $j \geq i$ and $m \geq j+(n-i)$.
After reading the prefix $\sigma_1\dots\sigma_{j-1}$ the DFA $\mathcal{A}_{\sigma,i}$ is at least in state $q_{i-1}$. Thus, after reading $\sigma_1\dots\sigma_{j}$ it is at least in state $q_i$ and will reject after reading $n-i$ more letters. Since $m \geq j + (n-i)$, we have $w \notin \lang{\mathcal{A}_{\sigma,i}}$.
\cref{lem:fl_A_sigmai} formalizes this result:\begin{restatable}{lemma}{lemFlASigmai}
	\label{lem:fl_A_sigmai}
	Let $\sigma \in \Sigma$ and $i \in \{1,\dots,n\}$ with $\sigma \notin \Sigma_{i-1,i}$. The following assertions hold:
	\begin{romanenumerate}
		\item $\mathcal{A}_{\sigma,i} \in \alpha(\mathcal{A})$.
		\item Let $m \in \natNum$. Let $w \in \sigma_1 \dots \sigma_m \in \Sigma^m$ such that $\sigma_j = \sigma$ for a $j \in \{1,\dots,m\}$ with $j \geq i$ and $m \geq j+(n-i)$. Then $w$ is rejected by $\mathcal{A}_{\sigma,i}$.\lipicsEnd
	\end{romanenumerate}
	\vspace{1.9ex}
\end{restatable}

Now consider a word $w = \sigma_1\dots\sigma_m \in \Sigma^m$ with $m \geq n + (n-1)$ and $\sigma_1\dots\sigma_n \in L$. Note that \cref{lem:fl_A_sigmai} implies $w \notin \lang{\mathcal{A}_{\sigma_n,i}}$ where $i \in \{1,\dots,n\}$ with $\sigma_n \notin \Sigma_{i-1,i}$.
With this limitation of length, we only need DFAs to reject the extensions of words $u \in L, |u|=n$ with a maximum length of $n+(n-2)$. Consider such an extension $w = \sigma_1 \dots \sigma_m \in \Sigma^m$. That is, $n+1 \leq m \leq n+(n-2)$ and $\sigma_1\dots\sigma_n \in L$. This implies $\sigma_i \in \Sigma_{i-1,i}$ for each $i \in \{1,\dots,n\}$ but provides no information about the $\sigma_i$ with $i \in \{n+1,\dots,m\}$. Therefore, we construct DFAs rejecting every such extension not confirming to a certain structure. This structure will be key to the further DFA constructions.

For a word $w \in \Sigma^*$, let $\mathcal{A}_w^!$ be the DFA rejecting exactly the words containing $w$ as a subsequence. Clearly, the following holds:
\begin{restatable}{lemma}{lemFlAWEXCL}
	\label{lem:fl_A_w^!}
	Let $w \notin L, |w|=n$. Then $\mathcal{A}_w^! \in \alpha(\mathcal{A})$ holds.\lipicsEnd
	\vspace{1.9ex}
\end{restatable}

With the DFAs $\mathcal{A}_w^!$ for every $w \notin L, |w|=n$ in hand, we only have to consider extensions of words $u \in L, |u|=n$ with a maximum length of $n+(n-2)$ for which every subsequence of length $n$ is in $L$. 

Let $w = \sigma_1\dots\sigma_m$ be an extension satisfying these conditions. We construct a DFA $\tilde{\mathcal{A}}_w \in \alpha(\mathcal{A})$ rejecting $w$. We utilize the rejecting state $q_d$ and define $\tilde{\mathcal{A}}_w = (\tilde{Q}_w,\Sigma,q_0,\tilde{\delta}_w,\tilde{F}_w)$ with $\tilde{Q} = \{q_0,\dots,q_n\}$, $\tilde{F}_w = \tilde{Q}_w \setminus \{q_d\}$ and $\tilde{\delta}_w(q_0,w) = q_d$.
Further, we have $\tilde{\delta}_w(q_0,v) = q_d$ for a $v \in \Sigma^*$ only if $\delta(q_0,v) \in \{q_d,q_{n+1}\}$, ensuring $\tilde{\mathcal{A}}_w \in \alpha(\mathcal{A})$.
In order to utilize $q_d$ in this manner, the DFA $\tilde{\mathcal{A}}_w$ simulates the behavior of $\mathcal{A}$ for the states $q_0,\dots,q_{d-1}$. 
The task then is to select the transitions of states $q_{d},\dots,q_n$.

If $|\sigma_{d+1}\dots\sigma_m|_{\sigma_m} \leq n-d$ the DFA $\tilde{\mathcal{A}}_w$ can simply advance for occurrences of $\sigma_m$ and the first $n-d-|\sigma_{d+1}\dots\sigma_{m-1}|_{\sigma_m}$ occurrences of letters unequal to $\sigma_m$.
Thus, we only have to consider the case $|\sigma_{d+1}\dots\sigma_m|_{\sigma_m} > n-d$.

If $\sigma_{n+1} \neq \sigma_m$ the DFA $\tilde{\mathcal{A}}_w$ can advance for each letter in $\Sigma$, ensuring $\tilde{\delta}_w(q_d,\sigma_{d+1}\dots\sigma_n) = q_n$. Further, we can define $\tilde{\delta}_w(q_n,\sigma_{n+1}) = q_{n-[(m-1)-(n+2)+1]}$ and $\tilde{\delta}_w(q_n,\sigma_m) = q_d$. Note that $|\sigma_{n+2}\dots\sigma_{m-1}| = (m-1)-(n+2)+1$. Since every subsequence of $w$ of length $n$ is in $L$, we have $\tilde{\delta}_w(q_{n-[(m-1)-(n+2)+1]},\sigma_{n+2}\dots\sigma_{m-1}) = q_n$.

The case $\sigma_{n+1} = \sigma_m$ is more complex and needs a further case distinction, but the idea used above of circling back after reading an appropriate prefix can be employed again.

\cref{lem:fl_tildeA_w} summarizes these ideas:
\begin{restatable}{lemma}{lemFlTildeAW}
	\label{lem:fl_tildeA_w}
	Let $w \in \Sigma^*$ with $|w| > n$ such that $w \in \lang{\mathcal{A}_v^!}$ for each $v \notin L, |v|=n$ and $w \in \bigcap_{\sigma \in \Sigma} \lang{\mathcal{A}_{\sigma,i_\sigma}}$, where for each $\sigma \in \Sigma$ it is $i_\sigma = \maxOp{\{i \in \{1,\dots,n\} \setDel \sigma \notin \Sigma_{i-1,i}\}}$.
	Then there exists a DFA $\tilde{\mathcal{A}}_w \in \alpha(\mathcal{A})$ rejecting $w$.\lipicsEnd
	\vspace{1.9ex}
\end{restatable}

\cref{lem:fl_A_0A_dA_myUnderbari,lem:fl_A_sigmai,lem:fl_A_w^!,lem:fl_tildeA_w} imply \cref{cla:fl_characterization} (\ref{cla_ass:fl_characterization_non-sigmaN+non-safety}). To be more precise, we have $\lang{\mathcal{A}} = \lang{\mathcal{A}_0} \cap \lang{\hat{\mathcal{A}}_d} \cap \bigcap_{m=1}^{n-1}\bigcap_{\myUnderbar{i} \in I_m} \lang{\mathcal{A}_{\myUnderbar{i}}} \cap \bigcap_{\sigma \in \Sigma} \lang{\mathcal{A}_{\sigma,i_{\sigma}}} \cap \bigcap_{w \in X^!} \lang{\mathcal{A}_w^!} \cap \bigcap_{w \in \tilde{X}} \lang{\tilde{\mathcal{A}}_w}$, where $X^! = \{w \in \Sigma^n \setDel w \notin L\}$ and $\tilde{X}$ is the set of all extensions $w$ of words $u \in L,|u|=n$ with $|w| \leq n+(n-2)$ for which every subsequence of length $n$ is in $L$. This proves the compositionality of $\mathcal{A}$ and thereby \cref{cla:fl_characterization} (\ref{cla_ass:fl_characterization_non-sigmaN+non-safety}).

\subsection{Linear safety ADFAs}
\label{subsec:fl_linearSafetyDFAs}
Next, we consider \cref{cla:fl_characterization} (\ref{cla_ass:fl_characterization_non-sigmaN+CEP}) and (\ref{cla_ass:fl_characterization_linear+safety+non-CEP}). For (\ref{cla_ass:fl_characterization_non-sigmaN+CEP}) we argue that $\mathcal{A}$ is composite if it has the CEP, even if $\mathcal{A}$ is a safety DFA, which makes circling back impossible. For (\ref{cla_ass:fl_characterization_linear+safety+non-CEP}) we argue that $\mathcal{A}$ is prime if it is a safety DFA and it does not have the CEP.

First, we consider (\ref{cla_ass:fl_characterization_non-sigmaN+CEP}). We assume that $\mathcal{A}$ has the CEP and argue that this implies compositionality.
Note that we can reuse the DFAs $\mathcal{A}_0$ and $\mathcal{A}_{\myUnderbar{i}}$, while $\hat{\mathcal{A}}_d$ is not needed. This again leaves the task of rejecting the extensions of words $w \in L, |w|=n$. But, since for every such word $w = \sigma_1\dots\sigma_n$ there now exist $i \in \{0,\dots,n-2\},l \in \{2,\dots,n-i\}$ such that $\delta(q_0,\sigma_1\dots\sigma_i\sigma_{i+l}\dots\sigma_n) \in \{q_n,q_{n+1}\}$, we can construct a DFA $\mathcal{A}_{i,l} \in \alpha(\mathcal{A})$ rejecting every extension of $w$.

The DFA $\mathcal{A}_{i,l}$ possesses states $q_0,\dots,q_{i+l-2},q_{i+l},\dots,q_{n+1}$. It simulates the behavior of $\mathcal{A}$ for states $q_0,\dots,q_{i-1}$, redirecting transitions $q_j \rightarrow q_{i+l-1}$ to $q_i$. From $q_i$ it directly advances to $q_{i+l}$ if a letter in $\bigcup_{j=i+l}^{n+1} \Sigma_{i,j}$ is read, otherwise it advances to $q_{i+1}$. The states $q_i,\dots,q_{i+l-2}$ form a loop. For states $q_{i+l},\dots,q_n$, every transition leads to the direct successor state. The state $q_{n+1}$ is a rejecting sink.

It is shown in the appendix that every extension of $w$ is rejected by $\mathcal{A}_{i,l}$, where $i$ is the largest possible value belonging to $w$, and that $\mathcal{A}_{i,l} \in \alpha(\mathcal{A})$. Thus, $\lang{\mathcal{A}} = \lang{\mathcal{A}_0} \cap \bigcap_{m=1}^{n-1}\bigcap_{\myUnderbar{i} \in I_m} \lang{\mathcal{A}_{\myUnderbar{i}}} \cap \bigcap_{i=0}^{n-2}\bigcap_{l=2}^{n-i} \mathcal{A}_{i,l}$ holds, proving the compositionality of $\mathcal{A}$ and thus (\ref{cla_ass:fl_characterization_non-sigmaN+CEP}).

Next, we consider (\ref{cla_ass:fl_characterization_linear+safety+non-CEP}) and assume that $\mathcal{A}$ is a safety DFA and does not have the CEP.
Thus, there is a $w = \sigma_1\dots\sigma_n$ such that $\delta(q_0,\sigma_1\dots\sigma_i\sigma_{i+l}\dots\sigma_n) \notin \{q_n,q_{n+1}\}$ holds for every $i \in \{0,\dots,n-2\},l \in \{2,\dots,n-i\}$.
This implies the existence of a letter $\sigma \in \Sigma_{n-1,n}$ with $\sigma \notin \Sigma_{j,n+1}$ for every $j \in \{0,\dots,n-1\}$.
We show in the appendix that $w\sigma$ is a primality witness of $\mathcal{A}$, thus proving the primality of $\mathcal{A}$ and thereby (\ref{cla_ass:fl_characterization_linear+safety+non-CEP}).

This completes our discussion of \cref{cla:fl_characterization} (\ref{cla_ass:fl_characterization_non-linear})-(\ref{cla_ass:fl_characterization_linear+safety+non-CEP}). Since they imply \cref{the:fl_characterization}, we have characterized the compositionality of ADFAs and thereby of finite languages.
\section{Complexity of $\primeDFAfin{}$}
\label{sec:fl_complexity}
After characterizing the compositionality of ADFAs and thereby of finite languages in \cref{sec:fl_characterization}, we now analyze the complexity of $\primeDFAfin{}$. We argue:
\begin{restatable}{theorem}{theFlPrimeDFAFinComplexity}
	\label{the:fl_primeDFAFinComplexity}
	The problem $\primeDFAfin{}$ is \complexityClassFont{NL}-complete. The \complexityClassFont{NL}-completeness holds true even when restricting $\primeDFAfin{}$ to DFAs with at most two letters.\lipicsEnd
	\vspace{1.9ex}
\end{restatable}

We begin by arguing that $\primeDFAfin{}$ is in \complexityClassFont{NL}, providing an \complexityClassFont{NL}-algorithm for $\primeDFAfin{}$ with \cref{alg:fl_primeDFAFinNLalgorithm}.
The algorithm accepts in line 1 if the given DFA $\mathcal{A}$ recognizes the empty language.
Then lines 2-18 ensure that the minimal DFA belonging to $\mathcal{A}$ is linear.
Lines 19-22 ensure that $\mathcal{A}$ is accepted if a letter $\sigma \in \Sigma$ with $\sigma^n \in L$ exists or else that $\mathcal{A}$ is rejected if it is not a safety DFA. 
Finally, in lines 23-29 the CEP is checked for $\mathcal{A}$.
\begin{algorithm}[t]
	\caption{\complexityClassFont{NL}-algorithm for $\primeDFAfin{}$.}
	\label{alg:fl_primeDFAFinNLalgorithm}
	\begin{algorithmic}[1]
		\REQUIRE DFA $\mathcal{A} = (Q,\Sigma,q_0,\delta,F)$ with $Q = \{q_0,\dots,q_m\}$ recognizing a finite language $L$.
		\ENSURE The DFA $\mathcal{A}$ is prime.
		\STATE Accept if $L = \emptyset$.
		\STATE $c \leftarrow 0$
		\FORALL {$i \in \{0,\dots,m\}$}
			\IF {$q_i$ is unreachable}
				\STATE $c \leftarrow c + 1$
			\ELSE
				\STATE $j \leftarrow 0$, $b \leftarrow \textit{true}$
				\WHILE {$j \leq i-1$ \AND $b$}
					\IF {$q_j$ is reachable and $\lang{\mathcal{A}^{q_i}} = \lang{\mathcal{A}^{q_j}}$}
						\STATE $c \leftarrow c + 1$
						\STATE $b \leftarrow \textit{false}$
					\ENDIF
					\STATE $j \leftarrow j + 1$
				\ENDWHILE
			\ENDIF
		\ENDFOR
		\STATE $n \leftarrow (m+1)-c-2$
		\STATE Choose nondeterministically a word $w \in \Sigma^n$. Reject if $w \notin L$.
		\STATE Choose nondeterministically a letter $\sigma \in \Sigma$. Accept if $\sigma^n \in L$.
		\FORALL {$i \in \{0,\dots,m\}$ where $q_i$ is not unreachable}
			\STATE Reject if $q_i \notin F$ and $\lang{\mathcal{A}^{q_i}} \neq \emptyset$.
		\ENDFOR
		\FORALL {$x \in \{1,\dots,n\}$}
			\STATE Choose nondeterministically a word $w = \sigma_1\dots\sigma_n \in \Sigma^n$. Reject if $w \notin L$.
			\FORALL {$i \in \{0,\dots,n-2\}, l \in \{2,\dots,n-i\}$ with $i+l=x$}
				\STATE Choose nondeterministically a word $w' = \sigma_1'\dots\sigma_n' \in \Sigma^n$ with $\sigma_{i+l}' = \sigma_x$ and a word $v \in \Sigma^+$. Reject if $w' \notin L$ or if $\sigma_1'\dots\sigma_i'\sigma_{i+l}'\dots\sigma_n'v \notin L$.
			\ENDFOR
		\ENDFOR
		\STATE Accept.
	\end{algorithmic}
\end{algorithm}

The \complexityClassFont{NL}-hardness of $\primeDFAfin{}$ can be proven by \complexityClassFont{L}-reducing \problemFont{STCONDAG} to $\primeDFAfin{}$, where \problemFont{STCONDAG} is the restriction of \problemFont{STCON} to acyclic graphs. The \complexityClassFont{L}-reduction is similar to the \complexityClassFont{L}-reduction of \problemFont{STCON} to the emptiness problem for DFAs.
\section{Finite Languages under Different Notions of Compositionality}
\label{sec:fl_differentNotionsOfCompositionality}
So far, we have only considered $\cap$-compositionality. Now we will define two further notions of compositionality and characterize the compositionality of finite languages for these notions.
\begin{definition2}
	For $k \in \natNumGeq{1}$, a DFA $\mathcal{A}$ is \defHighlight{$k$-$\cup$-decomposable} (\defHighlight{$k$-DNF-decomposable}) if there exist DFAs $\mathcal{A}_1,\dots,\mathcal{A}_t$ ($\mathcal{A}_{1,1},\dots,\mathcal{A}_{1,t_1},\dots,\mathcal{A}_{s,1},\dots,\mathcal{A}_{s,t_s}$) with $\lang{\mathcal{A}} = \bigcup_{i=1}^{t} \lang{\mathcal{A}_i}$ ($\lang{\mathcal{A}} = \bigcup_{i=1}^{s}\bigcap_{j=1}^{t_i} \lang{\mathcal{A}_{i,j}}$) and $\size{\mathcal{A}_i} < k$ for every $i$ ($\size{\mathcal{A}_{i,j}} < k$ for every pair $i,j$). The further concepts introduced in \cref{def:compositionality} are defined analogously.\lipicsEnd
\end{definition2}

In \cite{DBLP:journals/iandc/KupfermanM15}, it is correctly remarked that many results for $\cap$-compositionality can be trivially transferred to $\cup$-compositionality. For example, the complexity boundaries for \problemFont{Prime-DFA} established in \cite{DBLP:journals/iandc/KupfermanM15} also hold for $\cup$-compositionality. This does not hold true for results concerning language fragments that are not closed under complement. In particular, the complement language of a finite language is not finite, but co-finite. Thus, characterizing the $\cup$-compositionality of finite languages is equivalent to characterizing $\cap$-compositionality of co-finite languages.

Also in \cite{DBLP:journals/iandc/KupfermanM15}, the notion of compositionality allowing both union and intersection is suggested.
Note that DNF-compositionality enforces a structure similar to a disjunctive normal from, but is as strong as unrestricted union-intersection compositionality.
It is correctly remarked in \cite{DBLP:journals/iandc/KupfermanM15} that union-intersection compositionality - and thus, DNF-compositionality - is strictly stronger than $\cap$-compositionality. Obviously, it is also strictly stronger than $\cup$-compositionality. It is less obvious whether languages exist that are DNF-composite, but are neither $\cap$- nor $\cup$-composite. 
We will see that there are finite languages witnessing this.

The following result characterizes the $\cup$- and DNF-compositionality of finite languages:
\begin{restatable}{theorem}{theFlCupDNFCharacterization}
	\label{the:fl_cupDNFCharacterization}
	Consider a minimal ADFA $\mathcal{A} = (Q,\Sigma,q_I,\delta,F)$ recognizing a non-empty language. Let $n \in \natNum$ be the length of the longest word in $\lang{\mathcal{A}}$. The following assertions hold:
	\begin{romanenumerate}
		\item $\mathcal{A}$ is $\cup$-prime iff $\mathcal{A}$ is linear. \label{lem_ass:fl_cupDNFCharacterization_cup}
		\item $\mathcal{A}$ is DNF-prime iff $\mathcal{A}$ is linear and there exists a $\sigma \in \Sigma$ with $\sigma^n \in \lang{\mathcal{A}}$.\label{lem_ass:fl_cupDNFCharacterization_DNF}\lipicsEnd
	\end{romanenumerate}
	\vspace{1.9ex}
\end{restatable}

These conditions are similar to the conditions in \cref{the:fl_characterization}, but much simpler. Let $\mathcal{A}$ and $n$ be as required. It is easy to show $\cup$- and DNF-compositionality if $\mathcal{A}$ is not linear.

The proof of $\cup$-primality if $\mathcal{A}$ is linear relies on the observation that every minimal DFA $\mathcal{B}$ with $\lang{\mathcal{B}} \subseteq \lang{\mathcal{A}}$ and $\ind{\mathcal{B}} < \ind{\mathcal{A}}$ has to have a rejecting sink. From this follows that no such DFA $\mathcal{B}$ can accept a word $w \in \lang{\mathcal{A}},|w|=n$. Thus, $\mathcal{A}$ is $\cup$-prime.

If $\mathcal{A}$ is linear and there exists no $\sigma \in \Sigma$ with $\sigma^n \in \lang{\mathcal{A}}$ the DNF-compositionality of $\mathcal{A}$ follows from \cite[Example 3.2]{DBLP:journals/iandc/KupfermanM15}. 
On the other hand, if $\mathcal{A}$ is linear and there exists a $\sigma \in \Sigma$ with $\sigma^n \in \lang{\mathcal{A}}$ DNF-primality can be shown by adapting the proof of \cref{cla:fl_characterization} (\ref{cla_ass:fl_characterization_linear+sigmaN}).

As mentioned, \cref{the:fl_characterization,the:fl_cupDNFCharacterization} immediately imply:
\begin{restatable}{theorem}{theFlCapCupVsDNF}
	\label{the:fl_capCupVsDNF}
	There exists a finite language that is DNF-composite but $\cap$- and $\cup$-prime.\lipicsEnd
	\vspace{1.9ex}
\end{restatable}

To summarize, \cref{the:fl_characterization,the:fl_cupDNFCharacterization} characterize the $\cap$-, $\cup$- and DNF-compositionality of ADFAs and thus of finite languages. Obviously, this characterizes the $\cap$-, $\cup$- and DNF-compositionality of co-finite languages as well. The results further imply the existence of languages that are DNF-composite but $\cap$- and $\cup$-prime.
\section{\minimalDFA{2} and \sPrimeDFA{}}
\label{sec:2DFAMinimalAndSPrimeDFA}
We defined compositionality using the index of the given DFA. Thus, the compositionality of a DFA $\mathcal{A}$ is a characteristic of $\lang{\mathcal{A}}$. Slightly changing the definition, using the size instead of the index, turns compositionality of $\mathcal{A}$ into a characteristic of $\mathcal{A}$ itself. It is interesting to analyze the effects of this change, which results in the notion of S-compositionality.

Many results known for compositionality hold for S-compositionality as well. The characterization of finite languages in \cref{sec:fl_characterization} and other results concerning language fragments \cite{DBLP:journals/iandc/KupfermanM15,DBLP:conf/mfcs/JeckerKM20,DBLP:conf/concur/JeckerM021} are valid with only minor technical modifications. In fact, \cite{DBLP:conf/mfcs/JeckerKM20,DBLP:conf/concur/JeckerM021} already implicitly used S-compositionality instead of compositionality without discussing the differences.
The upper complexity boundary of \primeDFA{} holds for \sPrimeDFA{} as well. But the known lower boundary, the \complexityClassFont{NL}-hardness of \primeDFA{}, cannot simply be adapted for \sPrimeDFA{}. The lower boundary for \sPrimeDFA{} is connected to \minimalDFA{}, since non-minimal DFAs are trivially S-composite. Note that \primeDFA{} is connected to the emptiness problem for DFAs in a similar manner \cite{DBLP:journals/iandc/KupfermanM15}.

We begin by discussing \minimalDFA{}, proving the \complexityClassFont{NL}-hardness of \minimalDFA{2}. Then we formally introduce S-compositionality and prove the \complexityClassFont{NL}-hardness of the restriction \sPrimeDFA{2} and thereby of \sPrimeDFA{} as well. We also prove the \complexityClassFont{NL}-hardness of the restriction \primeDFA{2}, so far only known for the unrestricted problem \primeDFA{}.

\subsection{\complexityClassFont{NL}-hardness of \minimalDFA{2}}
\label{subsec:2DFAMinimal}
As mentioned, the \complexityClassFont{NL}-hardness and thus \complexityClassFont{NL}-completeness of \minimalDFA{k} for $k \in \natNumGeq{3}$ is folklore, while the \complexityClassFont{NL}-hardness of \minimalDFA{2} appears to be open. We prove:
\begin{restatable}{theorem}{theTwoDFAMinimalNLComplete}
	\label{the:2DFAMinimalNLComplete}
	The problem \minimalDFA{2} is \complexityClassFont{NL}-hard and thus \complexityClassFont{NL}-complete.\lipicsEnd
	\vspace{1.9ex}
\end{restatable}

The \complexityClassFont{NL}-hardness of \minimalDFA{3} can be proven by \complexityClassFont{L}-reducing \problemFont{2STCON} to \minimalDFA{3}. This known reduction uses an additional letter and cannot be used to prove the \complexityClassFont{NL}-hardness of \minimalDFA{2}. We give an \complexityClassFont{L}-reduction of \problemFont{2STCON} not using an additional letter, proving the \complexityClassFont{NL}-hardness and thus the \complexityClassFont{NL}-completeness of \minimalDFA{2}.

Let $(G,s,t)$ be an input for \problemFont{2STCON}. That is, $G = (V,E)$ is a graph with a maximum outdegree of two and $s,t \in V$ are nodes of $G$. We construct a DFA $\mathcal{A} = (Q,\Sigma,q_I,\delta,F)$ with $\Sigma = \{0,1\}$, which is minimal iff there exists a path in $G$ from $s$ to $t$. If $s=t$ such a path exists trivially and we can construct the minimal DFA for the empty language. Thus, we only have to consider the case $s \neq t$. W.l.o.g. we assume $V = \{0,\dots,n-1\}$ and $s=0,t=n-1$.

Let $\mathcal{A}' = (Q',\Sigma,0,\delta',F')$ be the DFA constructed out of $G$ in the usual manner, that is, by turning nodes into states, edges into transitions, setting the state $0$ as the initial state and $n-1$ as the only accepting state. For $\mathcal{A}$, we introduce the new states $p_0,\dots,p_{n-1}$, called $p$-states, the new states $q_0,\dots,q_{n-1}$ and $q_0'$, called $q$-states, and for each $i \in Q'$ the states $i_0',i_1',i_0,i_1$. We call the states $i,i_0',i_1',i_0,i_1$ for $i \in Q'$ $v$-states. We say that states $p_i,q_i,i,i_0',i_1',i_0,i_1$ for an $i \in Q'$ are located on the same layer. \cref{fig:2STCONto2DFAMinimalReduction} specifies the DFA $\mathcal{A}$ constructed for the \complexityClassFont{L}-reduction. We now discuss the key ideas of this construction.
\begin{figure}[t]
	\centering
	\begin{tikzpicture}[node distance=1.5cm]
	\tiny
	\node[state, initial right]				(p0) 	{$p_0$};
	\node[state, above left of=p0]			(01')	{$0_1'$};
	\node[state, below left of=01'] 		(0) 	{$0$};
	\node[state, above left of=0]			(00')	{$0_0'$};
	\node[state, left of=00'] 		(q0) 	{$q_0$};
	\node[state, left of=q0] 				(q0') 	{$q_0'$};
	\node[state, above of=00']				(00)	{$0_0$};
	\node[state, above of=01']				(01)	{$0_1$};
	
	\node[state, above right of=01]			(p1) 	{$p_1$};
	\node[state, above left of=p1]			(11')	{$1_1'$};
	\node[state, below left of=11'] 		(1) 	{$1$};
	\node[state, above left of=1]			(10')	{$1_0'$};
	\node[state, left of=10'] 		(q1) 	{$q_1$};
	\node[state, above of=10']				(10)	{$1_0$};
	\node[state, above of=11']				(11)	{$1_1$};
	
	\node[state, draw=none, above right of=11]	(empty-node) 	{};
	
	\node[state, above of=empty-node]		(pn-1) 	{$p_{n-1}$};
	\node[state, above left of=pn-1]		(n-11')	{${n-1}_1'$};
	\node[state, below left of=n-11', accepting](n-1) 	{$n-1$};
	\node[state, above left of=n-1]			(n-10')	{${n-1}_0'$};
	\node[state, left of=n-10'] 		(qn-1) 	{$q_{n-1}$};
	\node[state, above of=n-10']			(n-10)	{${n-1}_0$};
	\node[state, above of=n-11']			(n-11)	{${n-1}_1$};
	
	\draw	(p0)	edge[right]				node{$0$}	(p1);
	\draw	(p1)	edge[right,dashed]		node{$0$}	(pn-1);
	\draw	(pn-1)	edge[loop right]		node{$0$}	(pn-1);
	
	\draw	(p0)	edge[above]				node{$1$}	(0);
	\draw	(p1)	edge[above]				node{$1$}	(1);
	\draw	(pn-1)	edge[above]				node{$1$}	(n-1);
	
	\draw	(q0)	edge[above]				node{$0$}	(q0');
	\draw	(q0')	edge[left]				node{$0$}	(q1);
	\draw	(q1)	edge[left,dashed]		node{$0$}	(qn-1);
	\draw	(qn-1)	edge[left,bend right=45]node{$0$}	(q0);
	
	\draw	(q0)	edge[below]				node{$1$}	(0);
	\draw	(q0')	edge[loop left]			node{$1$}	(q0');
	\draw	(q1)	edge[loop left]			node{$1$}	(q1);
	\draw	(qn-1)	edge[loop left]			node{$1$}	(qn-1);
	
	\draw	(0)		edge[right]					node{$0$}	(00');
	\draw	(0)		edge[left]					node{$1$}	(01');
	\draw	(00')	edge[loop right]			node{$0$}	(00');
	\draw	(01')	edge[loop right]			node{$0$}	(01');
	\draw	(00')	edge[right,pos=0.25]		node{$1$}	(00);
	\draw	(01')	edge[left,pos=0.25]			node{$1$}	(01);
	\draw	(00)	edge[above,pos=0.25]		node{$1$}	(q0);
	\draw	(01)	edge[above,pos=0.25]		node{$0$}	(q0);
	
	\draw	(1)		edge[right]					node{$0$}	(10');
	\draw	(1)		edge[left]					node{$1$}	(11');
	\draw	(10')	edge[loop right]			node{$0$}	(10');
	\draw	(11')	edge[loop right]			node{$0$}	(11');
	\draw	(10')	edge[right,pos=0.25]		node{$1$}	(10);
	\draw	(11')	edge[left,pos=0.25]			node{$1$}	(11);
	\draw	(10)	edge[above,pos=0.25]		node{$1$}	(q1);
	\draw	(11)	edge[above,pos=0.25]		node{$0$}	(q1);
	
	\draw	(n-1)	edge[right]					node{$0$}	(n-10');
	\draw	(n-1)	edge[left]					node{$1$}	(n-11');
	\draw	(n-10')	edge[loop right]			node{$0$}	(n-10');
	\draw	(n-11')	edge[loop right]			node{$0$}	(n-11');
	\draw	(n-10')	edge[right,pos=0.25]		node{$1$}	(n-10);
	\draw	(n-11')	edge[left,pos=0.25]			node{$1$}	(n-11);
	\draw	(n-10)	edge[below,pos=0.125]		node{$1$}	(qn-1);
	\draw	(n-11)	edge[above,pos=0.25]		node{$0$}	(qn-1);
	
	\draw	(00)	edge[left,pos=0.125,dashed,shorten >= 1.125cm]	node{$0$}		(10');
	\draw	(01)	edge[right,pos=0.125,dashed,shorten >= 1.125cm]	node{$1$}		(11');
	\draw	(10)	edge[left,pos=0.075,dashed,shorten >= 2.425cm]	node{$0$}		(n-10');
	\draw	(11)	edge[right,pos=0.075,dashed,shorten >= 2.425cm]	node{$1$}		(n-11');
	\draw	(n-10)	edge[bend right=60,left,pos=0.25,dashed,shorten >= 0.28cm]	node{$0$}		(n-10');
	\draw	(n-11)	edge[bend left=60,right,pos=0.25,dashed,shorten >= 0.28cm]	node{$1$}		(n-11');
	\end{tikzpicture}
	\caption{DFA $\mathcal{A}$ constructed for the \complexityClassFont{L}-reduction of \problemFont{2STCON} to \minimalDFA{2}. The $j$-transitions exiting states of the form $i_j$ are only indicated.}
	\label{fig:2STCONto2DFAMinimalReduction}
\end{figure}
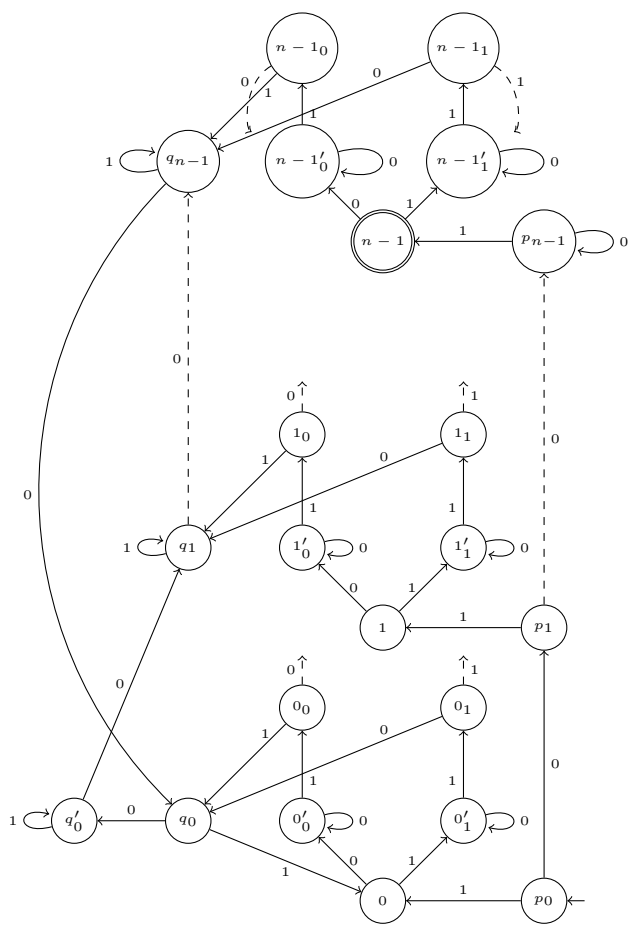

First, note that the idea of the $p$- and $q$-states is similar to the known \complexityClassFont{L}-reduction of \problemFont{2STCON} to \minimalDFA{3}. The $p$-states are used to access every state in $Q$, thus avoiding unreachable states. The $q$-states are used to allow the return to $0$ from every state.

Second, we cannot use an additional letter to switch from $p_i$ to $i$ to $q_i$. 
Thus, letter $1$ is used to leave the $p$-states and to exit $q_0$ to state $0$.
Letter $0$ is used to advance to the next layer in both the $p$- and $q$-states. To allow switching from the $v$-states to the $q$-states, we introduce for each $i \in Q'$ a component consisting of $i$ and the two branches $i_0',i_0$ and $i_1',i_1$.
The states $i_0',i_1'$ are waiting states used to prove the non-equivalence of $q$- and $v$-states.
The states $i_0,i_1$ implement on the one hand the original transitions in $\mathcal{A}'$, that is, $\delta(i_j,j) = \delta'(i,j)$, and on the other hand the transitions into the $q$-states, that is, $\delta(i_j,1-j) = q_i$.

Third, an extra $q$-state $q_0'$ is introduced, which is only directly accessible from $q_0$. Without $q_0'$ the situation $\delta(1_1,1) = 0 = \delta(q_0,1)$ and $\delta(1_1,0) = q_1 = \delta(q_0,0)$ would be possible, immediately implying the non-minimality of $\mathcal{A}$. The introduction of $q_0'$ solves this problem.

Note that there is a path from $0$ to $n-1$ in $\mathcal{A}$ iff there is such a path in $G$. Using this it follows that $\mathcal{A}$ is minimal iff there exists a path from $0$ to $n-1$ in $G$. Since $\mathcal{A}$ can obviously be constructed in logarithmic space, the given construction is indeed an \complexityClassFont{L}-reduction of \problemFont{2STCON} to \minimalDFA{2}. Consequently, \minimalDFA{2} is \complexityClassFont{NL}-hard.

\subsection{Complexity of \sPrimeDFA{}}
\label{subsec:S-Prime-DFA}
We end our discussion by using the construction presented in \cref{subsec:2DFAMinimal} to establish complexity boundaries for \sPrimeDFA{}. First, we define the notion of S-compositionality.
\begin{definition2}
	A DFA $\mathcal{A}$ is \defHighlight{S-composite} if there is a $k \in \natNumGeq{1}, k < \size{\mathcal{A}}$ such that $\mathcal{A}$ is $k$-decomposable. Otherwise, $\mathcal{A}$ is \defHighlight{S-prime}.\lipicsEnd
\end{definition2}
We denote the problem of deciding S-primality for a given DFA with \sPrimeDFA{} and the restriction of \sPrimeDFA{} to DFAs with at most $k \in \natNumGeq{2}$ letters with \sPrimeDFA{k}.

Note that the proof used in \cite{DBLP:journals/iandc/KupfermanM15} to show that \primeDFA{} is in \complexityClassFont{ExpSpace} is applicable for \sPrimeDFA{} with only slight modifications. 
Next, note that the \complexityClassFont{L}-reduction of the emptiness problem for DFAs to \primeDFA{} used in \cite{DBLP:journals/iandc/KupfermanM15} to prove the \complexityClassFont{NL}-hardness of \primeDFA{} relies on the fact that every DFA recognizing the empty language is prime. Thus, it is not easily adaptable for \sPrimeDFA{}. Instead, the \complexityClassFont{NL}-hardness of \sPrimeDFA{2} is shown by using a reduction from \problemFont{2STCON}, which adapts the construction outlined in \cref{subsec:2DFAMinimal}. We get:
\begin{restatable}{theorem}{theSPrimeDFAComplexity}
	\label{the:SPrimeDFAComplexity}
	The problems \sPrimeDFA{} and \sPrimeDFA{k} for $k \in \natNumGeq{2}$ are in \complexityClassFont{ExpSpace} and they are \complexityClassFont{NL}-hard.\lipicsEnd
	\vspace{1.9ex}
\end{restatable}

Further, we denote with \primeDFA{k} the restriction of \primeDFA{} to DFAs with at most $k \in \natNumGeq{2}$ letters and remark that the results presented in \cite{DBLP:journals/iandc/KupfermanM15} can be expanded to:
\begin{restatable}{theorem}{thePrimeDFAComplexity}
	\label{the:PrimeDFAComplexity}
	The problems \primeDFA{} and \primeDFA{k} for $k \in \natNumGeq{2}$ are in \complexityClassFont{ExpSpace} and they are \complexityClassFont{NL}-hard.\lipicsEnd
	\vspace{1.9ex}
\end{restatable}

This ends our discussion of the complexity of \sPrimeDFA{} and its restrictions, in which we have applied the construction outlined in \cref{subsec:2DFAMinimal} to prove \complexityClassFont{NL}-hardness.
\section{Discussion}
\label{sec:Discussion}
We studied the intersection compositionality, also denoted with $\cap$-compositionality, of regular languages. We added to the existing line of research focusing on fragments of the regular languages by analyzing the $\cap$-compositionality of ADFAs and thereby of finite languages. This research was in part motivated by existing results concerning the concatenation compositionality of finite languages.

We completely characterized the $\cap$-compositionality of ADFAs and thus finite languages. Using this characterization we proved the \complexityClassFont{NL}-completeness of $\primeDFAfin{}$. Thus, finite languages are significantly easier to handle under $\cap$-compositionality than under concatenation compositionality, where the respective primality problem for finite languages is \complexityClassFont{NP}-hard \cite{DBLP:journals/corr/abs-1902-06253}.

With notions of compositionality using union and both union and intersection already suggested in \cite{DBLP:journals/iandc/KupfermanM15}, we formally introduced the notions of $\cup$- and DNF-compositionality. We characterized the $\cup$- and DNF-compositionality of finite languages, which proved to be far simpler than the characterization of $\cap$-compositionality. These results also imply the characterization of the $\cap$-, $\cup$- and DNF-compositionality of co-finite languages.

This suggests that the key feature of finite languages regarding compositionality is not the finiteness of the languages per se, but rather the existence of only finitely many meaningfully different runs of the respective DFAs, a feature finite languages have in common not only with co-finite languages, but also with languages whose minimal DFAs allow for cycles in both accepting and rejecting sinks.
A logical next step would therefore be the characterization of the compositionality of these DFAs.

We also note that in our proofs we employed $\cap$-compositionality results concerning a different language fragment, namely co-safety DFAs, studied in \cite{DBLP:journals/iandc/KupfermanM15}. This suggests the possibility of employing the results concerning finite languages in future analyses and stresses the usefulness of working with language fragments. We provided one application of the results concerning finite languages by using them to prove the existence of a language that is DNF-composite but $\cap$- and $\cup$-prime.

Furthermore, we presented a proof of the \complexityClassFont{NL}-hardness and thereby \complexityClassFont{NL}-completeness of the basic problem \minimalDFA{2}. While the \complexityClassFont{NL}-hardness of \minimalDFA{k} for $k \in \natNumGeq{3}$ is folklore, this result appears to be new.

We utilized this result to establish the known complexity boundaries of \primeDFA{} for the here newly introduced problem \sPrimeDFA{}. We extended these results to the restrictions \primeDFA{k} and \sPrimeDFA{k} for $k \in \natNumGeq{2}$.

While it is interesting that a slight variation in the definition of $\cap$-compositionality, which does not touch the validity of most results, requires a whole new approach to establish the known lower complexity boundary, the big task of closing the doubly exponential complexity gap for \primeDFA{} still remains. And now, this gap exists for \sPrimeDFA{} as well.

Therefore, with the analysis of language fragments, further notions of compositionality, and the complexity gaps for \primeDFA{} and \sPrimeDFA{}, there is still need for further research.



\bibliography{spenner-paperES-literatur.bib}

\appendix
\section{Proofs for \texorpdfstring{\cref{sec:fl_characterization}}{Section \ref{sec:fl_characterization}}}
\label{sec:fl_proofs}
In this section we provide detailed proofs for the results presented in \cref{sec:fl_characterization}. To increase readability and avoid overlong proofs, we introduce a number of additional lemmas.

Our goal is to prove \cref{the:fl_characterization}, which reads:
\theFlCharacterization*
In order to prove \cref{the:fl_characterization} we will follow the structure outlined in the opening of \cref{sec:fl_characterization}. That is, we prove the following five assertions one after another:
\claFlCharacterization*
Note that (\ref{cla_ass:fl_characterization_non-linear})-(\ref{cla_ass:fl_characterization_linear+safety+non-CEP}) cover the entire set of minimal ADFAs recognizing a non-empty language. Therefore, proving (\ref{cla_ass:fl_characterization_non-linear})-(\ref{cla_ass:fl_characterization_linear+safety+non-CEP}) is sufficient to prove \cref{the:fl_characterization}. Further, note that a DFA recognizing the empty language is trivially prime. Therefore, \cref{the:fl_characterization} indeed characterizes the compositionality of ADFAs and thereby of finite languages.

From here on, let $\mathcal{A} = (Q,\Sigma,q_I,\delta,F)$ be a minimal ADFA recognizing the non-empty language $L$. Let $n$ be the length of the longest word in $L$.

\subsection{Proofs of \texorpdfstring{\cref{cla:fl_characterization}}{Claim \ref{cla:fl_characterization}} (\ref{cla_ass:fl_characterization_non-linear}) and (\ref{cla_ass:fl_characterization_linear+sigmaN})}
\label{subsec:fl_charac_aAndb}
We begin by presenting the rather simple proofs of \cref{cla:fl_characterization} (\ref{cla_ass:fl_characterization_non-linear}) and (\ref{cla_ass:fl_characterization_linear+sigmaN}).

\begin{proof}[Proof of \cref{cla:fl_characterization} (\ref{cla_ass:fl_characterization_non-linear})]
	Assume that $\mathcal{A}$ is not linear. As mentioned in \cref{sec:preliminaries}, this implies $\ind{\mathcal{A}} > n+2$.
	
	For a word $w \in \Sigma^*$, let $\mathcal{A}_w$ be the minimal DFA with $\lang{\mathcal{A}_w} = \{w\}$. For a value $m \in \natNum$, let $\mathcal{A}_{\leq m}$ be the minimal DFA with $\lang{\mathcal{A}_{\leq m}} = \{w \in \Sigma^* \setDel |w| \leq m\}$. The trivial constructions are pictured in \cref{fig:fl_A_wA_leqm}. Note that $\size{\mathcal{A}_w} = |w| + 2$ and $\size{\mathcal{A}_{\leq m}} = m + 2$.
	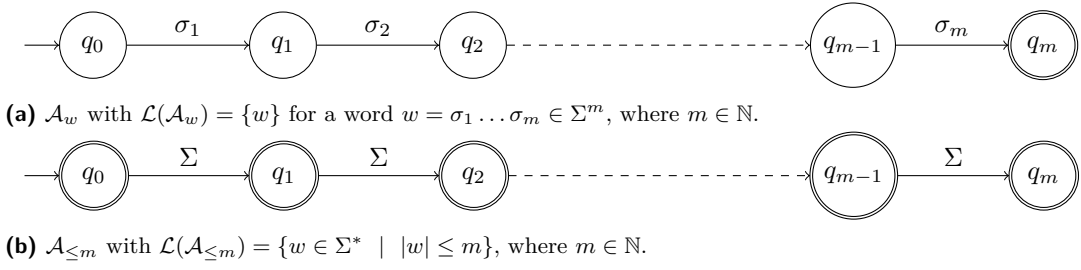
\begin{figure}[t]
		\begin{subfigure}[t]{0.75\textwidth}
			\centering
			\begin{tikzpicture}[node distance=2.5cm]
			\node[state, initial] 					(q0) 	{$q_0$};
			\node[state, right of=q0] 				(q1) 	{$q_1$};
			\node[state, right of=q1] 				(q2) 	{$q_2$};
			\node[state, draw=none, right of=q2]	(empty-node) 	{};
			\node[state, right of=empty-node] 		(qm-1) 	{$q_{m-1}$};
			\node[state, right of=qm-1, accepting] 	(qm) 	{$q_m$};
			
			\draw	(q0)	edge[above]				node{$\sigma_1$}	(q1);
			\draw	(q1)	edge[above]				node{$\sigma_2$}	(q2);
			\draw	(q2)	edge[dashed]			node{}				(qm-1);
			\draw	(qm-1)	edge[above]				node{$\sigma_m$}	(qm);
			\end{tikzpicture}
			\caption{$\mathcal{A}_{w}$ with $\lang{\mathcal{A}_w} = \{w\}$ for a word $w = \sigma_1\dots\sigma_m \in \Sigma^m$, where $m \in \natNum$.}
			\label{subfig:fl_A_w}
		\end{subfigure}\hfill
		\begin{subfigure}[t]{0.75\textwidth}
			\centering
			\begin{tikzpicture}[node distance=2.5cm]
			\node[state, accepting, initial] 					(q0) 	{$q_0$};
			\node[state, accepting, right of=q0] 				(q1) 	{$q_1$};
			\node[state, accepting, right of=q1] 				(q2) 	{$q_2$};
			\node[state, accepting, draw=none, right of=q2]		(empty-node) 	{};
			\node[state, accepting, right of=empty-node] 		(qm-1) 	{$q_{m-1}$};
			\node[state, accepting, right of=qm-1] 	(qm) 		{$q_m$};
			
			\draw	(q0)	edge[above]				node{$\Sigma$}	(q1);
			\draw	(q1)	edge[above]				node{$\Sigma$}	(q2);
			\draw	(q2)	edge[dashed]			node{}				(qm-1);
			\draw	(qm-1)	edge[above]				node{$\Sigma$}	(qm);
			\end{tikzpicture}
			\caption{$\mathcal{A}_{\leq m}$ with $\lang{\mathcal{A}_{\leq m}} = \{w \in \Sigma^* \setDel |w| \leq m\}$, where $m \in \natNum$.}
			\label{subfig:fl_A_leqm}
		\end{subfigure}
		\caption{DFAs $\mathcal{A}_{w}$ and $\mathcal{A}_{\leq m}$. Omitted transitions lead to a rejecting sink.}
		\label{fig:fl_A_wA_leqm}
	\end{figure}

	Define $L' = \{w \in \Sigma^* \setDel w \in \complementOp{\lang{\mathcal{A}}} \text{ and } |w| \leq n\}$. We proof: $\lang{\mathcal{A}} = \lang{\mathcal{A}_{\leq n}} \cap \bigcap_{w \in L'} \lang{\complementOp{\mathcal{A}_w}}$, where $\complementOp{\mathcal{A}_w}$ is the complement DFA of $\mathcal{A}_w$. Note that this immediately implies the compositionality of $\mathcal{A}$, since $\size{\mathcal{A}_{\leq n}} = n+2 < \ind{\mathcal{A}}$ and $\size{\complementOp{\mathcal{A}_w}} = |w| + 2 \leq n+2 < \ind{\mathcal{A}}$ for each $w \in L'$.
	
	We begin by proving $\lang{\mathcal{A}} \subseteq \lang{\mathcal{A}_{\leq n}} \cap \bigcap_{w \in L'} \lang{\complementOp{\mathcal{A}_w}}$. Let $w \in \lang{\mathcal{A}}$. Then $|w| \leq n$ holds and therefore we have $w \in \lang{\mathcal{A}_{\leq n}}$. Let $w' \in L'$. In particular, this implies $w' \notin \lang{\mathcal{A}}$ and thus $w' \neq w$. Therefore, $w \in \lang{\complementOp{\mathcal{A}_{w'}}}$. Taken together we get $w \in \lang{\mathcal{A}_{\leq n}} \cap \bigcap_{w \in L'} \lang{\complementOp{\mathcal{A}_w}}$. The first containment is shown.
	
	Next, we prove $\lang{\mathcal{A}_{\leq n}} \cap \bigcap_{w \in L'} \lang{\complementOp{\mathcal{A}_w}} \subseteq \lang{\mathcal{A}}$. Let $w \notin \lang{\mathcal{A}}$. If $|w| > n$ then $w \notin \lang{\mathcal{A}_{\leq n}}$. On the other hand, if $|w| \leq n$ we have $w \in L'$ and therefore $w \notin \bigcap_{w \in L'} \lang{\complementOp{\mathcal{A}_w}}$. Taken together we get $w \notin \lang{\mathcal{A}_{\leq n}} \cap \bigcap_{w \in L'} \lang{\complementOp{\mathcal{A}_w}}$. The second containment is shown and we are done.
\end{proof}

For the proof of \cref{cla:fl_characterization} (\ref{cla_ass:fl_characterization_linear+sigmaN}) we introduce an additional lemma, which we will reuse when proving \cref{the:fl_cupDNFCharacterization} (\ref{lem_ass:fl_cupDNFCharacterization_DNF}).
\begin{lemma}
	\label{lem:fl_sigmaN}
	Let $\mathcal{B} = (S,\Sigma,s_I,\eta,G)$ be a minimal DFA such that there exists an $m \in \natNum$ with $\ind{\mathcal{B}} < m+2$ and $\sigma^m \in \lang{\mathcal{B}}$ for a letter $\sigma \in \Sigma$. Then there exists an $l \in \{1,\dots,m+1\}$ such that $\eta(s_I,\sigma^{m+il}) = \eta(s_I,\sigma^m)$ for every $i \in \natNum$.\lipicsEnd
\end{lemma}
\begin{proof}
	Let $\mathcal{B}$ and $m$ be as required. Let $s = \eta(s_I,\sigma^m)$.
	\begin{description}
		\item[Case 1: \normalfont{In the initial run of $\mathcal{B}$ on $\sigma^m$ no state is visited more than once.}] Since $\ind{\mathcal{B}} < m+2$, we have $\ind{\mathcal{B}} = m+1$ and in the initial run of $\mathcal{B}$ on $\sigma^m$ every state of $\mathcal{B}$ is visited exactly once. Then there exists a $j \in \{0,\dots,m\}$ such that $\eta(s_I,\sigma^m\sigma) = \eta(s_I,\sigma^j)$. Select $l = m + 1 - j$. Clearly, $1 \leq l \leq m+1$ holds. Additionally, we have $\eta(s,\sigma^l) = \eta(s,\sigma^{m+1-j}) = \eta(s,\sigma\sigma^{m-j}) = \eta(\eta(s_I,\sigma^m),\sigma\sigma^{m-j}) = \eta(\eta(s_I,\sigma^m\sigma),\sigma^{m-j}) = \eta(\eta(s_I,\sigma^j),\sigma^{m-j}) = \eta(s_I,\sigma^j\sigma^{m-j}) = \eta(s_I,\sigma^m) = s$.
		
		We now argue that the selected $l$ satisfies the requirement using induction. For $i = 0$, we have $\eta(s_I,\sigma^{m+il}) = \eta(s_I,\sigma^m) = s$. Now assume that $\eta(s_I,\sigma^{m+il}) = s$ holds for a particular $i \in \natNum$. Then we have $\eta(s_I,\sigma^{m+(i+1)l}) = \eta(s_I,\sigma^{m+il}\sigma^l) = \eta(\eta(s_I,\sigma^{m+il}),\sigma^l) = \eta(s,\sigma^l) = s$. We are done with Case 1.
		
		\item[Case 2: \normalfont{Else.}] Then there exist $s,t \in \{0,\dots,m\}$ with $s < t$ such that $\eta(s_I,\sigma^s) = \eta(s_I,\sigma^t)$. Then $\eta(s_I,\sigma^{s+(m-t)+1}) = \eta(s_I,\sigma^{t+(m-t)+1}) = \eta(s_I,\sigma^m\sigma)$ holds. Select $j = s + (m-t) + 1$. Note that $1 \leq j \leq m$.
		
		Thus, we have a $j \in \{1,\dots,m\} \subseteq \{0,\dots,m\}$ such that $\eta(s_I,\sigma^m\sigma) = \eta(s_I,\sigma^j)$. We can select $l = m+1-j$ and trace back Case 2 to Case 1. Therefore, we are done with Case 2.
	\end{description}
	With Cases 1 and 2 the proof is complete.
\end{proof}

With \cref{lem:fl_sigmaN} in hand, we can now turn to \cref{cla:fl_characterization} (\ref{cla_ass:fl_characterization_linear+sigmaN}).
\begin{proof}[Proof of \cref{cla:fl_characterization} (\ref{cla_ass:fl_characterization_linear+sigmaN})]
	Assume that $\mathcal{A}$ is linear and that there exists a letter $\sigma \in \Sigma$ with $\sigma^n \in L$. We show that $\sigma^{n+(n+1)!}$ is a primality witness of $\mathcal{A}$. Note that the linearity of $\mathcal{A}$ implies $\ind{\mathcal{A}} = n+2$.
	
	Let $\mathcal{B} = (S,\Sigma,s_I,\eta,G) \in \alpha(\mathcal{A})$. By definition, $\mathcal{B}$ is a minimal DFA with $\ind{\mathcal{B}} < \ind{\mathcal{A}} = n+2$ and $\lang{\mathcal{A}} \subseteq \lang{\mathcal{B}}$. In particular, we have $\sigma^n \in \lang{\mathcal{B}}$.
	
	With \cref{lem:fl_sigmaN} this implies the existence of an $l \in \{1,\dots,n+1\}$ such that $\eta(s_I,\sigma^{n+il}) = \eta(s_I,\sigma^n)$ for every $i \in \natNum$. In particular, this implies $\sigma^{n+(n+1)!} \in \lang{\mathcal{B}}$.
	
	Thus, the word $\sigma^{n+(n+1)!}$ is accepted by every DFA in $\alpha(\mathcal{A})$. With $\sigma^{n+(n+1)!} \notin \lang{\mathcal{A}}$ the word $\sigma^{n+(n+1)!}$ is a primality witness of $\mathcal{A}$. We are done.
\end{proof}

To summarize, we have proven \cref{cla:fl_characterization} (\ref{cla_ass:fl_characterization_non-linear}) and (\ref{cla_ass:fl_characterization_linear+sigmaN}). Additionally, with \cref{lem:fl_sigmaN} we have established a result that we can reuse in the proof of \cref{the:fl_cupDNFCharacterization} (\ref{lem_ass:fl_cupDNFCharacterization_DNF}).

For now, we have to prove the remaining \cref{cla:fl_characterization} (\ref{cla_ass:fl_characterization_non-sigmaN+non-safety})-(\ref{cla_ass:fl_characterization_linear+safety+non-CEP}). From here on, we assume that $\mathcal{A}$ is linear and that $\sigma^n \notin \lang{\mathcal{A}}$ holds for every $\sigma \in \Sigma$, as otherwise we have already covered the compositionality of $\mathcal{A}$ with \cref{cla:fl_characterization} (\ref{cla_ass:fl_characterization_non-linear}) and (\ref{cla_ass:fl_characterization_linear+sigmaN}).

As in \cref{sec:fl_characterization}, we assume w.l.o.g. $Q = \{q_0,\dots,q_{n+1}\}$ with $q_j$ being reachable from $q_i$ for all $i < j$, which implies $q_I = q_0$ and $q_n \in F$ with $q_{n+1}$ being the rejecting sink. Furthermore, we define $\Sigma_{i,j} = \{\sigma \in \Sigma \setDel \delta(q_i,\sigma) = q_j\}$.
Note that for every $\sigma \in \Sigma$ there exists an $i \in \{1,\dots,n\}$ with $\sigma \notin \Sigma_{i-1,i}$. The form of $\mathcal{A}$ is pictured in \cref{fig:fl_minLinDFA}.
\begin{figure}[t]
	\centering
	\begin{tikzpicture}[node distance=2cm]
	\footnotesize
	\node[state, initial, dash dot] 					(q0) 	{$q_0$};
	\node[state, right of=q0, dash dot] 				(q1) 	{$q_1$};
	\node[state, right of=q1, dash dot] 				(q2) 	{$q_2$};
	\node[state, right of=q2, draw=none]				(empty-node) 	{};
	\node[state, right of=empty-node, dash dot] 		(qn-1) 	{$q_{n-1}$};
	\node[state, right of=qn-1, accepting] 				(qn) 	{$q_n$};
	\node[state, right of=qn] 							(qn+1) 	{$q_{n+1}$};
	
	\draw	(q0)	edge[below]							node{$\Sigma_{0,1}$}	(q1);
	\draw	(q0)	edge[above, bend left=30]			node{$\Sigma_{0,2}$}	(q2);
	\draw	(q0)	edge[above, bend left=45]			node{$\Sigma_{0,n-1}$}	(qn-1);
	\draw	(q0)	edge[above, bend left=60]			node{$\Sigma_{0,n}$}	(qn);
	\draw	(q0)	edge[above, bend left=75]			node{$\Sigma_{0,n+1}$}	(qn+1);
	
	\draw	(q1)	edge[below]							node{$\Sigma_{1,2}$}	(q2);
	\draw	(q1)	edge[above, bend right=30]			node{$\Sigma_{1,n-1}$}	(qn-1);
	\draw	(q1)	edge[above, bend right=45]			node{$\Sigma_{1,n}$}	(qn);
	\draw	(q1)	edge[above, bend right=60]			node{$\Sigma_{1,n+1}$}	(qn+1);
	
	\draw	(q2)	edge[dashed]						node{}	(qn-1);
	\draw	(q2)	edge[above, bend left=30]			node{$\Sigma_{2,n-1}$}	(qn-1);
	\draw	(q2)	edge[above, bend left=45]			node{$\Sigma_{2,n}$}	(qn);
	\draw	(q2)	edge[above, bend left=60]			node{$\Sigma_{2,n+1}$}	(qn+1);
	
	\draw	(qn-1)	edge[below]							node{$\Sigma_{n-1,n}$}	(qn);
	\draw	(qn-1)	edge[above, bend right=50]			node{$\Sigma_{n-1,n+1}$}(qn+1);
	
	\draw	(qn)	edge[below]							node{$\Sigma_{n,n+1}$}	(qn+1);
	
	\draw	(qn+1)	edge[loop below]					node{$\Sigma$}			(qn+1);
	\end{tikzpicture}
	\caption{Minimal linear ADFA recognizing a non-empty language, with $n$ being the length of the longest accepted word.}
	\label{fig:fl_minLinDFA}
\end{figure}
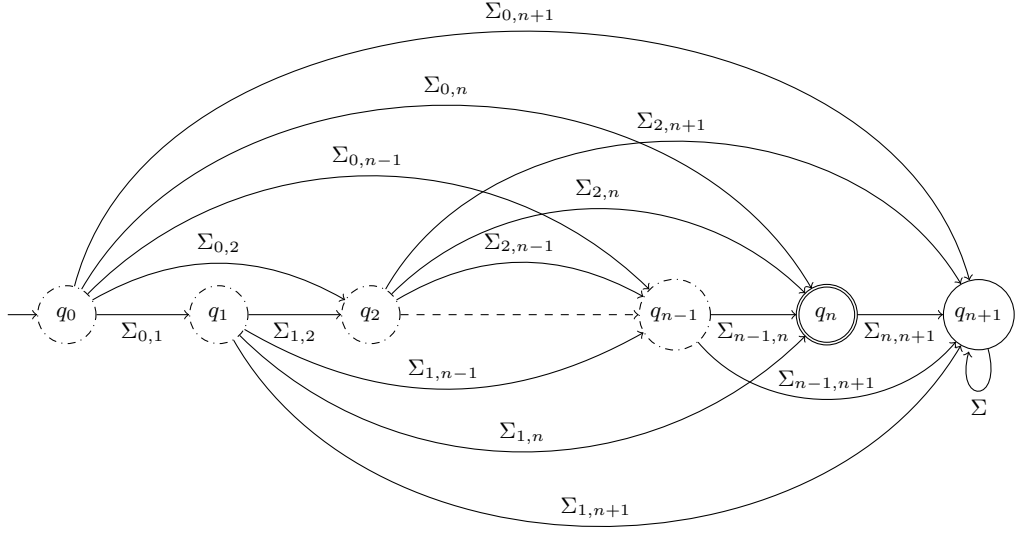

\subsection{Proof of \texorpdfstring{\cref{cla:fl_characterization}}{Claim \ref{cla:fl_characterization}} (\ref{cla_ass:fl_characterization_non-sigmaN+non-safety})}
\label{subsec:fl_charac_c}
We turn to the proof of \cref{cla:fl_characterization} (\ref{cla_ass:fl_characterization_non-sigmaN+non-safety}). We assume that $\mathcal{A}$ is not a safety DFA, which implies $\{q_n\} \subseteq F \subset Q \setminus \{q_{n+1}\}$. Let $d \in \{0,\dots,n-1\}$ with $q_d \notin F$.
We provide more details on the DFAs outlined in \cref{subsec:fl_linearNonSafetyDFAs}, which taken together witness the compositionality of $\mathcal{A}$.

We begin by proving \cref{lem:fl_A_0A_dA_myUnderbari,lem:fl_A_sigmai,lem:fl_A_w^!}. Afterwards, we will prove \cref{lem:fl_tildeA_w} by providing four additional lemmas for the four cases mentioned in \cref{subsec:fl_linearNonSafetyDFAs}.

Note \cref{fig:fl_automataBigPicture}, which is partly a repetition of \cref{fig:fl_A_myUnderbariA_sigmai} and which outlines the DFA constructions. Additionally, we now provide formal definitions for these DFAs.
\begin{description}
	\item[Definition of $\mathcal{A}_0$, outlined in \cref{subfig:fl_app_A_0}.]
	We define $\mathcal{A}_0 = (Q_0,\Sigma,q_0,\delta_0,F_0)$ where:
	\begin{align*}
		Q_0						&= Q \setminus \{q_n\}\\
		F_0						&= (F \cup \{q_0\}) \setminus \{q_n\}\\
		\delta_0(q,\sigma)		&= 
		\begin{cases}
			q_0					&\text{ if $\delta(q,\sigma) = q_n$}\\
			\delta(q,\sigma) 	&\text{ else, thus if $\delta(q,\sigma) \neq q_n$}
		\end{cases}.
	\end{align*}
	
	\item[Definition of $\hat{\mathcal{A}}_d$, outlined in \cref{subfig:fl_app_hatA_d}.]
	We define $\hat{\mathcal{A}} = (\hat{Q}_d,\Sigma,q_0,\hat{\delta}_d,\hat{F}_d)$ where:
	\begin{align*}
		\hat{Q}_d				&= Q \setminus \{q_{n+1}\}\\
		\hat{F}_d				&= F\\
		\hat{\delta}_d(q_i,\sigma)&=
		\begin{cases}
			q_d					&\text{ if $i = n$}\\
			q_n					&\text{ if $i < n$ and $\delta(q_i,\sigma) = q_{n+1}$}\\
								&\text{ (arbitrary definition possible)}\\
			\delta(q_i,\sigma)	&\text{ else, thus if $i < n$ and $\delta(q_i,\sigma) \neq q_{n+1}$}
		\end{cases}.
	\end{align*} 
	
	\item[Definition of $\mathcal{A}_{\myUnderbar{i}}$ for $m < n-1$, outlined in \cref{subfig:fl_app_A_myUnderbariA_m<n-1}.] Let $m < n-1$ and $\myUnderbar{i} = (i_0,\dots,i_m) \in I_m$. We define $\mathcal{A}_{\myUnderbar{i}} = (Q_{\myUnderbar{i}},\Sigma,q_0,\delta_{\myUnderbar{i}},F_{\myUnderbar{i}})$ where:
	\begin{align*}
		Q_{\myUnderbar{i}} 		&= \{q_{i_0},\dots,q_{i_m}\} \cup \{q_{n+1},q_+\}\\
		F_{\myUnderbar{i}}		&= Q_{\myUnderbar{i}} \setminus \{q_{n+1}\}\\
		\delta_{\myUnderbar{i}}(q,\sigma)&=
		\begin{cases}
			q_{i_{j+1}} 		&\text{ if $q = q_{i_j}$ for $j < m$ and $\sigma \in \Sigma_{i_j,i_{j+1}}$}\\
			q_+					&\text{ if $q = q_{i_j}$ for $j < m$ and $\sigma \notin \Sigma_{i_j,i_{j+1}}$}\\
			q_{n+1}				&\text{ if $q = q_{i_m}$}\\
			q					&\text{ else, thus if $q \in \{q_{n+1},q_+\}$}
		\end{cases}.
	\end{align*}
	
	\item[Definition of $\hat{\mathcal{A}}_d$ for $m = n-1$, outlined in \cref{subfig:fl_app_A_myUnderbariA_m=n-1}.]
	Let $m = n-1$ and $\myUnderbar{i} = (0,\dots,j-1,j+1,\dots,n) \in I_m$. We define $\mathcal{A}_{\myUnderbar{i}} = (Q_{\myUnderbar{i}},\Sigma,q_0,\delta_{\myUnderbar{i}},F_{\myUnderbar{i}})$ where:
	\begin{align*}
		Q_{\myUnderbar{i}}		&= \{q_0,\dots,q_{j-1},q_{j+1},\dots,q_{n+1}\}\\
		F_{\myUnderbar{i}}		&= Q_{\myUnderbar{i}} \setminus \{q_{n+1}\}\\
		\delta_{\myUnderbar{i}}(q_k,\sigma)=&
		\begin{cases}
			q_{k+1}				&\text{ if $k < n+1$ and $k \neq j-1$ and $\sigma \in \Sigma_{k,k+1}$}\\
			q_{k}				&\text{ if $k < n+1$ and $k \neq j-1$ and $\sigma \notin \Sigma_{k,k+1}$}\\
			q_{j+1}				&\text{ if $k = j-1$ and $\sigma \in \Sigma_{j-1,j+1}$}\\
			q_{j-1}				&\text{ if $k = j-1$ and $\sigma \notin \Sigma_{j-1,j+1}$}\\
			q_{n+1}				&\text{ else, thus if $k = n+1$}
		\end{cases}.
	\end{align*}
	
	\item[Definition of $\mathcal{A}_{\sigma,i}$, outlined in \cref{subfig:fl_app_A_sigmai}.]
	Let $\sigma \in \Sigma$ and $i \in \{1,\dots,n\}$. We define $\mathcal{A}_{\sigma,i} = (Q_{\sigma,i},\Sigma,q_0,\delta_{\sigma,i},F_{\sigma,i})$ where:
	\begin{align*}
		Q_{\sigma,i}			&= \{q_0,\dots,q_n\}\\
		F_{\sigma,i}			&= Q_{\sigma,i} \setminus \{q_n\}\\
		\delta(q_j,\sigma')		&=
		\begin{cases}
			q_i					&\text{ if $j = i-1$ and $\sigma' = \sigma$}\\
			q_{i-1}				&\text{ if $j = i-1$ and $\sigma' \neq \sigma$}\\
			q_{j+1}				&\text{ if $j \neq i-1$ and $j < n$}\\
			q_n					&\text{ else, thus if $j = n$}
		\end{cases}.
	\end{align*}
	
	\item[Definition of $\mathcal{A}_w^!$, outlined in \cref{subfig:fl_app_A_w^!}.]
	Let $w = \sigma_1\dots\sigma_m \in \Sigma^m$. We define $\mathcal{A}_w^! = (Q_w^!,\Sigma,q_0,\delta_w^!,F_w^!)$ where:
	\begin{align*}
		Q_w^!					&= \{q_0,\dots,q_m\}\\
		F_w^!					&= Q_w^! \setminus \{q_m\}\\
		\delta_w^!(q_i,\sigma)	&=
		\begin{cases}
			q_{i+1}				&\text{if $i < m$ and $\sigma = \sigma_{i+1}$}\\
			q_i					&\text{if $i < m$ and $\sigma \neq \sigma_{i+1}$}\\
			q_m					&\text{else, thus if $i = m$}
		\end{cases}.
	\end{align*}
\end{description}
\begin{figure}
	\begin{subfigure}[t]{0.75\textwidth}
		\centering
		\begin{tikzpicture}[node distance=2.5cm]
		\footnotesize
		\node[state, initial, accepting] 					(q0) 	{$q_0$};
		\node[state, right of=q0, dash dot] 				(q1) 	{$q_1$};
		\node[state, right of=q1, draw=none]				(empty-node) 	{};
		\node[state, right of=empty-node, dash dot] 		(qn-1) 	{$q_{n-1}$};
		\node[state, right of=qn-1] 						(qn+1) 	{$q_{n+1}$};
		
		\draw	(q0)	edge[above]							node{$\Sigma_{0,1}$}								(q1);
		\draw	(q0)	edge[bend left=45, above]				node{$\Sigma_{0,n-1}$}								(qn-1);
		\draw	(q0)	edge[bend left=60, above]				node{$\Sigma_{0,n+1}$}								(qn+1);
		\draw	(q0)	edge[loop above]					node{$\Sigma_{0,n}$}								(q0);
		\draw	(q1)	edge[dashed]						node{}												(qn-1);
		\draw	(q1)	edge[bend left=30, above]				node{$\Sigma_{1,n-1}$}								(qn-1);
		\draw	(q1)	edge[bend left=45, above]				node{$\Sigma_{1,n+1}$}								(qn+1);
		\draw	(q1)	edge[bend left=30, above]				node{$\Sigma_{1,n}$}								(q0);
		\draw	(qn-1)	edge[above]							node{$\Sigma_{n-1,n+1}$}							(qn+1);
		\draw	(qn-1)	edge[bend left=30, above]				node{$\Sigma_{n-1,n}$}								(q0);
		\draw	(qn+1)	edge[loop above]					node{$\Sigma$}										(qn+1);
		\end{tikzpicture}
		\caption{$\mathcal{A}_0$.}
		\label{subfig:fl_app_A_0}
	\end{subfigure}\hfill
	\begin{subfigure}[t]{0.75\textwidth}
		\centering
		\begin{tikzpicture}[node distance=2.5cm]
		\footnotesize
		\node[state, initial, dash dot] 					(q0) 	{$q_0$};
		\node[state, right of=q0, draw=none]				(empty-node1) 	{};
		\node[state, right of=empty-node1] 		(qd) 	{$q_d$};
		\node[state, right of=qd, draw=none]				(empty-node2) 	{};
		\node[state, right of=empty-node2, dash dot] 		(qn-1) 	{$q_{n-1}$};
		\node[state, right of=qn-1, accepting] 						(qn) 	{$q_n$};
		
		\draw	(q0)	edge[dashed]						node{}								(qd);
		\draw	(q0)	edge[bend left=45, above]			node{$\Sigma_{0,n},\Sigma_{0,n+1}$} (qn);
		\draw	(qd)	edge[dashed]						node{}								(qn-1);
		\draw	(qd)	edge[bend left=37.5, above]			node{$\Sigma_{d,n},\Sigma_{d,n+1}$}	(qn);
		\draw	(qn-1)	edge[bend right=30, below]			node{$\Sigma_{n-1,n},\Sigma_{n-1,n+1}$}(qn);
		\draw	(qn)	edge[bend right=22.5, above]			node{$\Sigma$}						(qd);
		\end{tikzpicture}
		\caption{$\hat{\mathcal{A}}_d$. Unaltered transitions are omitted.}
		\label{subfig:fl_app_hatA_d}
	\end{subfigure}\hfill
	\begin{subfigure}[t]{0.75\textwidth}
		\centering
		\begin{tikzpicture}[node distance=2.25cm]
		\footnotesize
		\node[state, initial, accepting] 					(q0) 	{$q_0$};
		\node[state, right of=q0, accepting] 				(qi1) 	{$q_{i_1}$};
		\node[state, right of=qi1, draw=none]				(empty-node) 	{};
		\node[state, right of=empty-node, accepting] 		(qim-1) {$q_{i_{m-1}}$};
		\node[state, right of=qim-1, accepting] 			(qn) {$q_n$};
		\node[state, right of=qn] 							(qn+1) {$q_{n+1}$};
		\node[state, below of=q0, accepting] 				(q+) 	{$q_+$};
		
		\draw	(q0)	edge[above]						node{$\Sigma_{0,i_1}$}								(qi1);
		\draw	(q0)	edge[left]						node{$\complementOp{\Sigma_{0,i_1}}$}								(q+);
		\draw	(qi1)	edge[dashed]					node{}								(qim-1);
		\draw	(qi1)	edge[right]						node{$\complementOp{\Sigma_{i_1,i_2}}$}								(q+);
		\draw	(qim-1)	edge[above]						node{$\Sigma_{i_{m-1},n}$}								(qn);
		\draw	(qim-1)	edge[below]						node{$\complementOp{\Sigma_{i_{m-1},n}}$}								(q+);
		\draw	(qn)	edge[above]						node{$\Sigma$}								(qn+1);
		\draw	(qn+1)	edge[loop right]				node{$\Sigma$}						(qn+1);
		\draw	(q+)	edge[loop right]				node{$\Sigma$}						(q+);
		\end{tikzpicture}
		\caption{$\mathcal{A}_{\myUnderbar{i}}$ if $m < n-1$.}
		\label{subfig:fl_app_A_myUnderbariA_m<n-1}
	\end{subfigure}\hfill
	\begin{subfigure}[t]{0.75\textwidth}
		\centering
		\begin{tikzpicture}[node distance=2cm]
		\footnotesize
		\node[state, initial, accepting] 					(q0) 	{$q_0$};
		\node[state, right of=q0, draw=none]				(empty-node1) 	{};
		\node[state, right of=empty-node1, accepting] 		(qj-1) {$q_{j-1}$};
		\node[state, right of=qj-1, accepting] 				(qj+1) {$q_{j+1}$};
		\node[state, right of=qj+1, draw=none]				(empty-node2) 	{};
		\node[state, right of=empty-node2, accepting] 		(qn) {$q_n$};
		\node[state, right of=qn] 							(qn+1) {$q_{n+1}$};
		
		\draw	(q0)	edge[dashed]						node{}								(qj-1);
		\draw	(q0)	edge[loop above]					node{$\complementOp{\Sigma_{0,1}}$}	(q0);
		\draw	(qj-1)	edge[below, bend right]							node{$\Sigma_{j-1,j+1}$}			(qj+1);
		\draw	(qj-1)	edge[loop above]					node{$\complementOp{\Sigma_{j-1,j+1}}$}	(qj-1);
		\draw	(qj+1)	edge[dashed]						node{}								(qn);
		\draw	(qj+1)	edge[loop above]					node{$\complementOp{\Sigma_{j+1,j+2}}$}	(qj-1);
		\draw	(qn)	edge[above]							node{$\Sigma$}						(qn+1);
		\draw	(qn+1)	edge[loop above]					node{$\Sigma$}						(qn+1);
		\end{tikzpicture}
		\caption{$\mathcal{A}_{\myUnderbar{i}}$ if $m = n-1$, where $\myUnderbar{i} = (0,\dots,j-1,j+1,\dots,n)$.}
		\label{subfig:fl_app_A_myUnderbariA_m=n-1}
	\end{subfigure}\hfill
	\begin{subfigure}[t]{0.75\textwidth}
		\centering
		\begin{tikzpicture}[node distance=1.725cm]
		\footnotesize
		\node[state, initial, accepting] 					(q0) 	{$q_0$};
		\node[state, right of=q0, accepting] 				(q1) 	{$q_1$};
		\node[state, right of=q1, draw=none]				(empty-node1) 	{};
		\node[state, right of=empty-node1, accepting] 		(qi-1) 	{$q_{i-1}$};
		\node[state, right of=qi-1, accepting] 				(qi) 	{$q_{i}$};
		\node[state, right of=qi, draw=none]				(empty-node2) 	{};
		\node[state, right of=empty-node2, accepting]		(qn-1) 	{$q_{n-1}$};
		\node[state, right of=qn-1] 						(qn) 	{$q_n$};
		
		\draw	(q0)	edge[above]							node{$\Sigma$}					(q1);
		\draw	(q1)	edge[above,dashed]					node{$\Sigma$}					(qi-1);
		\draw	(qi-1)	edge[loop above]					node{$\Sigma\setminus\{\sigma\}$}(qi-1);
		\draw	(qi-1)	edge[above]							node{$\sigma$}					(qi);
		\draw	(qi)	edge[above,dashed]					node{$\Sigma$}					(qn-1);
		\draw	(qn-1)	edge[above]							node{$\Sigma$}					(qn);
		\draw	(qn)	edge[loop above]					node{$\Sigma$}					(qn);
		\end{tikzpicture}
		\caption{$\mathcal{A}_{\sigma,i}$.}
		\label{subfig:fl_app_A_sigmai}
	\end{subfigure}\hfill
	\begin{subfigure}[t]{0.75\textwidth}
		\centering
		\begin{tikzpicture}[node distance=2.5cm]
		\footnotesize
		\node[state, initial, accepting] 					(q0) 	{$q_0$};
		\node[state, right of=q0, accepting] 				(q1) 	{$q_1$};
		\node[state, right of=q1, draw=none]				(empty-node) 	{};
		\node[state, right of=empty-node, accepting] 		(qm-2) 	{$q_{m-2}$};
		\node[state, right of=qm-2, accepting] 				(qm-1) 	{$q_{m-1}$};
		\node[state, right of=qm-1] 						(qm) 	{$q_m$};
		
		\draw	(q0)	edge[above]						node{$\sigma_1$}					(q1);
		\draw	(q0)	edge[loop above]				node{$\Sigma \setminus \{\sigma_1\}$}(q0);
		\draw	(q1)	edge[dashed]					node{}								(qm-2);
		\draw	(q1)	edge[loop above]				node{$\Sigma \setminus \{\sigma_2\}$}(q1);
		\draw	(qm-2)	edge[above]						node{$\sigma_{m-1}$}				(qm-1);
		\draw	(qm-2)	edge[loop above]				node{$\Sigma \setminus \{\sigma_{m-1}\}$}(qm-2);
		\draw	(qm-1)	edge[above]						node{$\sigma_{m}$}					(qm);
		\draw	(qm-1)	edge[loop above]				node{$\Sigma \setminus \{\sigma_m\}$}(qm-1);
		\draw	(qm)	edge[loop above]				node{$\Sigma$}						(qm);
		\end{tikzpicture}
		\caption{$\mathcal{A}_w^!$.}
		\label{subfig:fl_app_A_w^!}
	\end{subfigure}
	\caption{DFAs $\mathcal{A}_0$ and $\hat{\mathcal{A}}_d$. DFA $\mathcal{A}_{\myUnderbar{i}}$ for $\myUnderbar{i} \in I_m$ with $m \in \{1,\dots,n-1\}$. DFA $\mathcal{A}_{\sigma,i}$ for $\sigma \in \Sigma, i \in \{1,\dots,n\}$. DFA $\mathcal{A}_w^!$ for $w = \sigma_1\dots\sigma_m \in \Sigma^m$.}
	\label{fig:fl_automataBigPicture}
\end{figure}
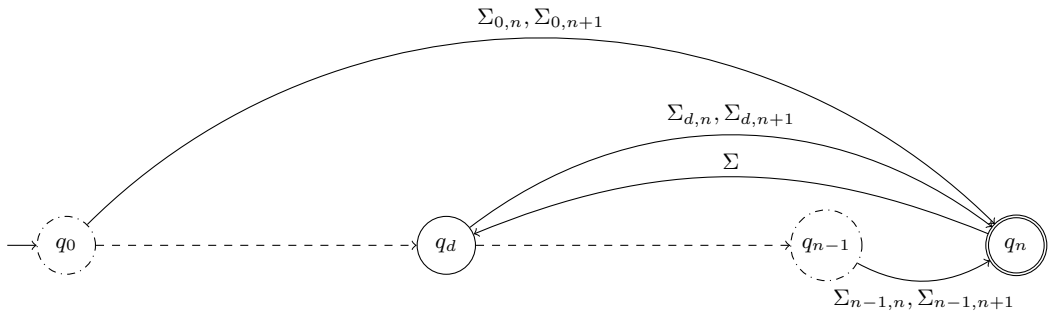
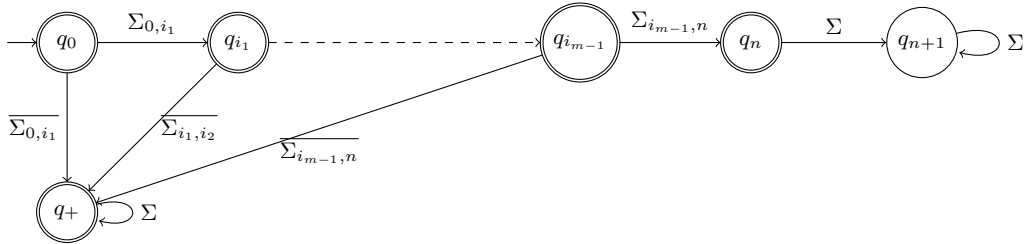
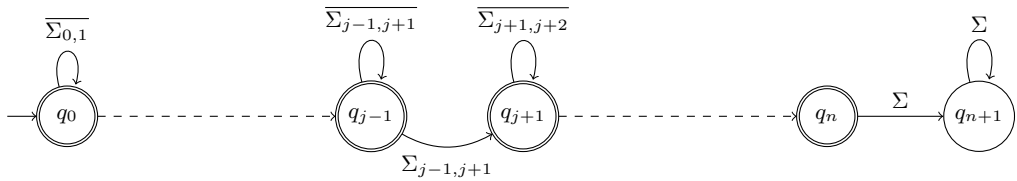
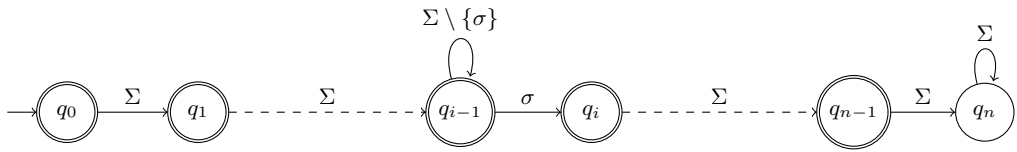
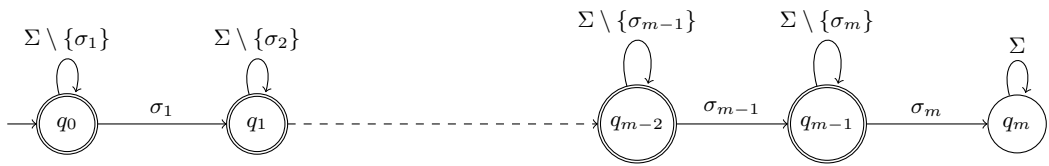

First, we consider \cref{lem:fl_A_0A_dA_myUnderbari}, which reads:
\lemFlANULLAdAMyUnderbari*
\begin{proof}[Proof of \cref{lem:fl_A_0A_dA_myUnderbari}]
	We begin by looking at (i).
	
	First, consider $\mathcal{A}_0 = (Q_0,\Sigma,q_0,\delta_0,F_0)$. As seen in \cref{subfig:fl_app_A_0}, we have $\size{\mathcal{A}_0} = n+1 < n+2 = \ind{\mathcal{A}}$. It remains to show that $\lang{\mathcal{A}} \subseteq \lang{\mathcal{A}_0}$. Let $w \in \lang{\mathcal{A}}$. Then there exists an $i \in \{0,\dots,n\}$ with $\delta(q_0,w) = q_i$ and $q_i \in F$. If $i < n$ we obviously have $\delta_0(q_0,w)=q_i$ and $q_i \in F_0$ as well. If $i = n$ then $\delta_0(q_0,w)=q_0 \in F_0$ holds. Therefore, we have $w \in \lang{\mathcal{A}_0}$. We have shown $\mathcal{A}_0 \in \alpha(\mathcal{A}_0)$.
	
	Second, consider $\hat{\mathcal{A}}_d = (\hat{Q}_d,\Sigma,q_0,\hat{\delta}_d,\hat{F}_d)$. As seen in \cref{subfig:fl_app_hatA_d}, we again have $\size{\hat{\mathcal{A}}_d} = n+1 < n+2 = \ind{\mathcal{A}}$. We need to show that $\lang{\mathcal{A}} \subseteq \lang{\hat{\mathcal{A}}_d}$. Let $w \in \lang{\mathcal{A}}$. Then there exists an $i \in \{0,\dots,n\}$ with $\delta(q_0,w) = q_i$ and $q_i \in F$. Obviously, we have $\hat{\delta}_d(q_0,w) = q_i \in \hat{F}_d$ and thus $w \in \lang{\hat{\mathcal{A}}_d}$. We have shown $\hat{\mathcal{A}}_d \in \alpha(\mathcal{A})$.
	
	Third, consider $\mathcal{A}_{\myUnderbar{i}} = (Q_{\myUnderbar{i}},\Sigma,q_0,\delta_{\myUnderbar{i}},F_{\myUnderbar{i}})$ for an $\myUnderbar{i} = (i_0,\dots,i_m) \in I_m$ where $m \in \{1,\dots,n-1\}$. 
	
	We begin by assuming $m < n-1$. As seen in \cref{subfig:fl_app_A_myUnderbariA_m<n-1}, we have $\size{\mathcal{A}_{\myUnderbar{i}}} = m+3 < (n-1)+3 = n+2 = \ind{\mathcal{A}}$. We need to show that $\lang{\mathcal{A}} \subseteq \lang{\mathcal{A}_{\myUnderbar{i}}}$. Let $w = \sigma_1\dots\sigma_l \in \Sigma^l$ with $w \notin \lang{\mathcal{A}_{\myUnderbar{i}}}$. Clearly, this implies $l > m$ with $\sigma_1 \in \Sigma_{i_0,i_1},\dots,\sigma_m \in \Sigma_{i_{m-1},i_m}$. Therefore, we have $\delta(q_0,\sigma_1\dots\sigma_m) = q_n$ and thus $\delta(q_0,w) = q_{n+1}$. Thus, we have $w \notin \lang{\mathcal{A}}$. We have shown $\mathcal{A}_{\myUnderbar{i}} \in \alpha(\mathcal{A})$ if $m < n-1$.
	
	Next, we assume $m = n-1$. As seen in \cref{subfig:fl_app_A_myUnderbariA_m=n-1}, we have $\size{\mathcal{A}_{\myUnderbar{i}}} = m+2 = (n-1)+2 = n+1 < n+2 = \ind{\mathcal{A}}$. Again, it remains to show that $\lang{\mathcal{A}} \subseteq \lang{\mathcal{A}_{\myUnderbar{i}}}$. Let $w = \sigma_1\dots\sigma_l \in \Sigma^l$ with $w \notin \lang{\mathcal{A}_{\myUnderbar{i}}}$. Clearly, this implies $l \geq n$. If $l > n$ it trivially follows that $w \notin \lang{\mathcal{A}}$. If $l = n$ we have $\sigma_1 \in \Sigma_{0,1},\dots,\sigma_{j-1} \in \Sigma_{j-2,j-1}$ and $\sigma_j \in \Sigma_{j-1,j+1}$ and $\sigma_{j+1} \in \Sigma_{j+1,j+2},\dots,\sigma_{n} \in \Sigma_{n,n+1}$, where $j$ is the one value in $\{0,\dots,n\} \setminus \{i_0,\dots,i_{n-1}\}$. Therefore, we have $\delta(q_0,w) = q_{n+1}$ and thus $w \notin \lang{\mathcal{A}}$. We have shown $\mathcal{A}_{\myUnderbar{i}} \in \alpha(\mathcal{A})$ if $m = n-1$.
	
	Together we have shown $\mathcal{A}_{\myUnderbar{i}} \in \alpha(\mathcal{A})$. This completes the proof of (i).
	
	We now turn to (ii). Consider a word $w = \sigma_1 \dots \sigma_l \in \Sigma^l$ with $w \notin L$ that is not an extension of a word $u \in L, |u|=n$. That is, we do not have $l>n$ with $\delta(q_0,\sigma_1\dots\sigma_n) = q_n$. We need to prove that $w \notin \lang{\mathcal{A}_0} \cap \lang{\hat{\mathcal{A}}_d} \cap \bigcap_{m=1}^{n-1}\bigcap_{\myUnderbar{i} \in I_m}\lang{\mathcal{A}_{\myUnderbar{i}}}$.
	
	If $w \notin \lang{\mathcal{A}_0} \cap \lang{\hat{\mathcal{A}}_d}$ we are done. Therefore, we assume $w \in \lang{\mathcal{A}_0} \cap \lang{\hat{\mathcal{A}}_d}$. Note that $w \in \lang{\hat{\mathcal{A}}_d}$ implies $\delta(q_0,w) = q_{n+1}$. Further, note that $w \in \lang{\mathcal{A}_0}$ implies that in the run of $\mathcal{A}$ on $w$ the sink $q_{n+1}$ is entered from $q_n$. That is, there exists an $m \in \natNum, m < l$ such that $\delta(q_0,\sigma_1\dots\sigma_m) = q_n$. Since by definition $\delta(q_0,\sigma_1\dots\sigma_n) \neq q_n$ holds, this implies $m < n$. This immediately implies the existence of an $\myUnderbar{i} = \{i_0,\dots,i_m\} \in I_m$ with $\sigma_1 \in \Sigma_{i_0,i_1},\dots,\sigma_m \in \Sigma_{i_{m-1},i_m}$. Thus, we have $\delta_{\myUnderbar{i}}(q_0,\sigma_1\dots\sigma_m) = q_n$ and $\delta_{\myUnderbar{i}}(q_0,w) = q_{n+1} \notin F_{\myUnderbar{i}}$ and therefore $w \notin \lang{\mathcal{A}_{\myUnderbar{i}}}$.
	
	We have shown $w \notin \lang{\mathcal{A}_0} \cap \lang{\hat{\mathcal{A}}_d} \cap \bigcap_{m=1}^{n-1}\bigcap_{\myUnderbar{i} \in I_m}\lang{\mathcal{A}_{\myUnderbar{i}}}$, which completes the proof of (ii). We are done with the proof of \cref{lem:fl_A_0A_dA_myUnderbari}.
\end{proof}

Next, we consider \cref{lem:fl_A_sigmai}, which reads:
\lemFlASigmai*
\begin{proof}[Proof of \cref{lem:fl_A_sigmai}]
	Let $\sigma \in \Sigma$ and $i \in \{1,\dots,n\}$ with $\sigma \notin \Sigma_{i-1,i}$. We consider the DFA $\mathcal{A}_{\sigma,i} = (Q_{\sigma,i},\Sigma,q_0,\delta_{\sigma,i},F_{\sigma,i})$.
	
	First, we consider (i) and prove $\mathcal{A}_{\sigma,i} \in \alpha(\mathcal{A})$. As seen in \cref{subfig:fl_app_A_sigmai}, we have $\size{\mathcal{A}_{\sigma,i}} = n+1 < n+2 = \ind{\mathcal{A}}$. It remains to show that $\lang{\mathcal{A}} \subseteq \lang{\mathcal{A}_{\sigma,i}}$. Let $w=\sigma_1\dots\sigma_m \in \Sigma^m$ with $w \notin \lang{\mathcal{A}_{\sigma,i}}$. Note that this implies $|w| \geq n$. If $|w| > n$ then $w \notin \lang{\mathcal{A}}$ holds trivially. If $|w| = n$ then for every $j \in \{0,\dots,n\}$ we have $\delta_{\sigma,i}(q_0,\sigma_1\dots\sigma_j) = q_j$. In particular, this implies $\sigma_i = \sigma$. Therefore, we have a word $w = \sigma_1\dots\sigma_n$ with $\sigma_i \notin \Sigma_{i-1,i}$, which implies $w \notin \lang{\mathcal{A}}$. We have shown $\mathcal{A}_{\sigma,i} \in \alpha(\mathcal{A})$.
	
	Second, we consider (ii). Let $m \in \natNum$ and let $w = \sigma_1\dots\sigma_m \in \Sigma^m$ such that there exists a $j \in \{1,\dots,m\}$ with $j \geq i$, $\sigma_j = \sigma$ and $m \geq j + (n-i)$. We need to show $w \notin \lang{\mathcal{A}_{\sigma,i}}$.
	
	Note that $\delta_{\sigma,i}(q_0,\sigma_1\dots\sigma_{i-1}) = q_{i-1}$. Clearly, this implies $\delta_{\sigma,i}(q_0,\sigma_1\dots\sigma_j) = q_k$ for a $k \in \{0,\dots,n\}, k \geq i$. Thus, the DFA $\mathcal{A}_{\sigma,i}$ will enter the rejecting sink $q_n$ after reading $n-k$ additional letters. Note that with $m \geq j+(n-i)$ and $k \geq i$ we have $|\sigma_{j+1}\dots\sigma_m| = m-j \geq n-i \geq n-k$. Thus, we have $w \notin \lang{\mathcal{A}_{\sigma,i}}$. We are done with (ii), which completes the proof of \cref{lem:fl_A_sigmai}.
\end{proof}

Finally, we complete our first step of the proof of \cref{cla:fl_characterization} (\ref{cla_ass:fl_characterization_non-sigmaN+non-safety}) by proving \cref{lem:fl_A_w^!}, which reads:
\lemFlAWEXCL*
\begin{proof}[Proof of \cref{lem:fl_A_w^!}]
	Let $w = \sigma_1\dots\sigma_n \in \Sigma^n$ with $w \notin L$. We consider the DFA $\mathcal{A}_w^! = (Q_w^!,\Sigma,q_0,\delta_w^!,F_w^!)$.
	
	As seen in \cref{subfig:fl_app_A_w^!}, we have $\size{\mathcal{A}_w^!} = |w|+1 = n+1 < n+2 = \ind{\mathcal{A}}$. It remains to show that $\lang{\mathcal{A}} \subseteq \lang{\mathcal{A}_w^!}$. Let $w'=\sigma_1'\dots\sigma_m' \in \Sigma^m$ with $w' \notin \lang{\mathcal{A}_w^!}$. Obviously, this implies $|w'| \geq n$. If $|w'| > n$ then $w' \notin \lang{\mathcal{A}}$ holds trivially. If $|w'| = n$ we have $\delta_w^!(q_0,\sigma_1'\dots\sigma_i') = q_i$ for every $i \in \{0,\dots,n\}$, which implies $w' = w$. Since $w \notin \lang{\mathcal{A}}$ holds by definition, we have $w' \notin \lang{\mathcal{A}}$. We are done.
\end{proof}

So far, we have proven \cref{lem:fl_A_0A_dA_myUnderbari,lem:fl_A_sigmai,lem:fl_A_w^!}. Next, we turn to \cref{lem:fl_tildeA_w}, which reads:
\lemFlTildeAW*
We will provide four additional lemmas implementing the four cases mentioned in \cref{subsec:fl_linearNonSafetyDFAs}. The lemmas get increasingly technical, but the main idea in all of them is to design for a given word $w$ with $n < |w| \leq n+(n-2)$, which is an extension of a word $u \in L, |u|=n$, a DFA $\tilde{\mathcal{A}}_w \in \alpha(\mathcal{A})$ rejecting $w$. The DFA $\tilde{\mathcal{A}}_w$ simulates the behavior of $\mathcal{A}$ for states $q_0,\dots,q_{d-1}$ and the initial run of $\tilde{\mathcal{A}}_w$ on $w$ ends in $q_d$, which is the only rejecting state of $\tilde{\mathcal{A}}_w$.

We begin with:
\begin{lemma}
	\label{lem:fl_A_w'}
	Let $m \in \natNum$ with $n < m \leq n+(n-2)$. Let $w = \sigma_1\dots\sigma_m \in \Sigma^m$ with $\sigma_1\dots\sigma_n \in L$ and $|\sigma_{d+1}\dots\sigma_m|_{\sigma_m} \leq n-d$. 
	
	Let $l = |\sigma_{d+1}\dots\sigma_{m-1}|_{\sigma_m}$. Note that $l = |\sigma_{d+1}\dots\sigma_{m-1}|_{\sigma_m} < |\sigma_{d+1}\dots\sigma_m|_{\sigma_m} \leq n-d$.
	
	Let $S$ be the set of values $i \in \{d+1,\dots,m-1\}$ with $\sigma_i = \sigma_m$, that is $S = \{i \in \{d+1,\dots,m-1\} \setDel \sigma_i = \sigma_m\}$. Let $T$ be the set of the $n-d-l$ smallest values $i \in \{d+1,\dots,m-1\}$ with $\sigma_i \neq \sigma_m$, that is $T = \{i \in \{d+1,\dots,m-1\} \setDel |\{j \in \{d+1,\dots,m-1\} \setDel \sigma_j \neq \sigma_m \wedge j \leq i\}| \leq n-d-l\}$. Note that $|S| = l < n-d$ and $|T| = n-d-l > 0$ and $S \cap T = \emptyset$.
	
	Let $I = \{i_1,\dots,i_{n-d}\} = S \cup T$ with $d+1 \leq i_1 < \dots < i_{n-d} \leq m-1$. Note that $i_1 = d+1$ and $|\sigma_{d+1}\dots\sigma_{i_{n-d}}|_{\sigma_m} = l = |\sigma_{d+1}\dots\sigma_{m-1}|_{\sigma_m}$ and $|\sigma_{i_{n-d}+1}\dots\sigma_{m}|_{\sigma_m} = 1 = |\sigma_{m}|_{\sigma_m}$.
	
	We define $\mathcal{A}_w' = (Q_w',\Sigma,q_0,\delta_w',F_w')$ as follows:
	\begin{align*}
		Q_w' &= \{q_0,\dots,q_n\} \\
		F_w' &= Q_w' \setminus \{q_d\}\\
		\delta_w'(q_j,\sigma) &= \begin{cases}
			\delta(q_j,\sigma) 	&\text{ if $j < d$ and $\delta(q_j,\sigma) \neq q_{n+1}$}\\
			q_j					&\text{ if $j < d$ and $\delta(q_j,\sigma) = q_{n+1}$}\\ 
								&\text{ (arbitrary definition possible)}\\
			q_{d+1}				&\text{ if $j=d$}\\
			q_{j+1}				&\text{ if $d < j \leq n-1$ and $\sigma = \sigma_{i_{j+1-d}}$}\\
			q_j					&\text{ if $d < j \leq n-1$ and $\sigma \neq \sigma_{i_{j+1-d}}$}\\
			q_d					&\text{ if $j=n$ and $\sigma = \sigma_m$}\\
			q_n					&\text{ else, thus if $j=n$ and $\sigma \neq \sigma_m$}
		\end{cases}.
	\end{align*}

	Then the following assertions hold:
	\begin{romanenumerate}
		\item $\mathcal{A}_w' \in \alpha(\mathcal{A})$.
		\item $w \notin \lang{\mathcal{A}_w'}$.
	\end{romanenumerate}
	Making (ii) more specific, the following holds:
	\begin{align*}
		\delta_w'(q_0,\sigma_1 \dots \sigma_d) &= q_d \\
		\delta_w'(q_0,\sigma_1 \dots \sigma_{d+1}) &= q_{d+1} \\
		\delta_w'(q_0,\sigma_1 \dots \sigma_{i_{n-d}}) &= q_n \\
		\delta_w'(q_0,\sigma_1 \dots \sigma_{m'}) &= q_n \text{ for each $m' \in \{i_{n-d},\dots,m-1\}$} \\
		\delta_w'(q_0,\sigma_1 \dots \sigma_m) &= q_d.
	\end{align*}\lipicsEnd
\end{lemma}
\begin{proof}
	Let $m \in \natNum$ and $w = \sigma_1\dots\sigma_m \in \Sigma^m$ be as required. Let $l,S,T,I=\{i_1,\dots,i_{n-d}\}$ be as defined above.
	
	We consider (i). We have $\size{\mathcal{A}_w'} = n+1 < n+2 = \ind{\mathcal{A}}$. It remains to show that $\lang{\mathcal{A}} \subseteq \lang{\mathcal{A}_w'}$. Let $w' \in \Sigma^*$ with $w' \notin \lang{\mathcal{A}_w'}$. Since $q_d$ is the only rejecting state of $\mathcal{A}_w'$, we have $\delta_w'(q_0,w') = q_d$. Note that $\mathcal{A}_w'$ simulates the behavior of $\mathcal{A}$ for the states $q_0,\dots,q_{d-1}$ and advances exactly one state at a time for the states $q_d,\dots,q_n$.
	This clearly implies that either $\delta(q_0,w') = q_d$ or $\delta(q_0,w') = q_{n+1}$ holds. Therefore, have $w' \notin \lang{\mathcal{A}}$. We are done with (i).
	
	Next, we consider (ii). Since $\sigma_1\dots\sigma_n \in \lang{\mathcal{A}}$, we have $\delta_w'(q_0,\sigma_1\dots\sigma_d) = q_d$. From here on, the transitions are chosen in such a way that $\mathcal{A}_w'$ advances for every occurrence of $\sigma_m$ and the first $n-d-l$ occurrences of letters other than $\sigma_m$. This guarantees $\delta_w'(q_0,\sigma_1 \dots \sigma_{i_{n-d}}) = q_n$. Since $|\sigma_{i_{n-d}+1}\dots\sigma_{m}|_{\sigma_m} = 1 = |\sigma_{m}|_{\sigma_m}$, the state $q_n$ is only left when reading the final letter of $w$. That is, $\delta_w'(q_0,\sigma_1 \dots \sigma_{m-1}) = q_n$ and $\delta_w'(q_0,\sigma_1 \dots \sigma_m) = q_d$. Thus, we have $w \notin \lang{\mathcal{A}_w'}$. We are done with (ii). The proof of \cref{lem:fl_A_w'} is complete.
\end{proof}

Note that the construction of \cref{lem:fl_A_w'} critically hinges on the property $|\sigma_{d+1}\dots\sigma_m|_{\sigma_m} \leq n-d$, which allows to advance for every occurrence of $\sigma_m$ in $\sigma_{d+1}\dots\sigma_{m-1}$. If $|\sigma_{d+1}\dots\sigma_m|_{\sigma_m} > n-d$ such a simple construction is not possible because state $q_n$ would be reached with more than one $\sigma_m$ left to read. We will handle words with $|\sigma_{d+1}\dots\sigma_m|_{\sigma_m} > n-d$ with different construction, starting with:
\begin{lemma}
	\label{lem:fl_A_w''}
	Let $m \in \natNum$ with $n < m \leq n+(n-2)$. Let $w = \sigma_1\dots\sigma_m \in \Sigma^m$ such that the following conditions hold:
	\begin{enumerate}
		\item $|\sigma_{d+1}\dots\sigma_{m}|_{\sigma_m} > n-d$.
		\item $\forall v \in \Sigma^n. v \notin \lang{\mathcal{A}} \Rightarrow w \in \lang{\mathcal{A}_v^!}$. Note that this implies that every subsequence of $w$ of length $n$ is in $\lang{\mathcal{A}}$.
		\item $\sigma_{n+1} \neq \sigma_m$. Note that this implies $m \geq n+2$. Thus, we have $n+2 \leq m \leq n+(n-2)$. This further implies $n \geq 4$.
	\end{enumerate}

	We define $\mathcal{A}_w'' = (Q_w'',\Sigma,q_0,\delta_w'',F_w'')$ as follows:
	\begin{align*}
		Q_w'' &= \{q_0,\dots,q_n\} \\
		F_w'' &= Q_w'' \setminus \{q_d\} \\
		\delta_w''(q_i,\sigma) &= \begin{cases}
			\delta(q_i,\sigma) 		&\text{ if $i < n$ and $\delta(q_i,\sigma) \neq q_{n+1}$}\\
			q_i						&\text{ if $i < n$ and $\delta(q_i,\sigma) = q_{n+1}$}\\
									&\text{ (arbitrary definition possible)}\\
			q_{n-[(m-1)-(n+2)+1]}	&\text{ if $i = n$ and $\sigma = \sigma_{n+1}$}\\
			q_d						&\text{ if $i = n$ and $\sigma = \sigma_m$} \\
			q_n						&\text{ else, thus if $i = n$ and $\sigma \notin \{\sigma_{n+1},\sigma_m\}$}\\
									&\text{ (arbitrary definition possible)}
		\end{cases}.
	\end{align*}
	Note that $|\sigma_{n+2}\dots\sigma_{m-1}| = (m-1)-(n+2)+1$ holds.
	
	Then the following assertions hold:
	\begin{romanenumerate}
		\item $\mathcal{A}_w'' \in \alpha(\mathcal{A})$.
		\item $w \notin \lang{\mathcal{A}_w''}$.
	\end{romanenumerate}
	Making (ii) more specific, the following holds:
	\begin{align*}
		\delta_w''(q_0,\sigma_1 \dots \sigma_n) &= q_n \\
		\delta_w''(q_0,\sigma_1 \dots \sigma_{n+1}) &= q_{n-[(m-1)-(n+2)+1]}\\
		\delta_w''(q_0,\sigma_1 \dots \sigma_{m'}) &= q_{n-[(m-1)-(n+2)+1] + (m'-(n+1))} \text{ for each $m' \in \{n+2,\dots,m-1\}$} \\
		\delta_w''(q_0,\sigma_1 \dots \sigma_m) &= q_d.
	\end{align*}
	\lipicsEnd
\end{lemma}
\begin{proof}
	Let $m \in \natNum$ and $w = \sigma_1\dots\sigma_m \in \Sigma^m$ be as required.
	
	We begin by showing that $\mathcal{A}_w''$ is well-defined. That is, we show that $n-[(m-1)-(n+2)+1] \in \{0,\dots,n\}$. Note that $n-[(m-1)-(n+2)+1] = 2n+2-m$. Since $n+2 \leq m \leq n+(n-2)$, we have $n-[(m-1)-(n+2)+1] = 2n+2-m \geq 2n+2-(n+(n-2)) = 4$ and $n-[(m-1)-(n+2)+1] = 2n+2-m \leq 2n+2-(n+2) = n$. Thus, we have $4 \leq n-[(m-1)-(n+2)+1] \leq n$. Since $n \geq 4$, this implies that $\mathcal{A}_w''$ is well-defined.
	
	Next, we turn to (i). We argue that $\mathcal{A}_w'' \in \alpha(\mathcal{A})$. The argumentation here is similar to \cref{lem:fl_A_w'} (i). We have $\size{\mathcal{A}_w''} = n+1 < n+2 = \ind{\mathcal{A}}$. Additionally, for a word $w' \in \Sigma^*$ we have $\delta_w''(q_0,w') = q_d$ only if $\delta(q_0,w')=q_d$ or $\delta(q_0,w')=q_{n+1}$, which implies $\lang{\mathcal{A}}\subseteq\lang{\mathcal{A}_w''}$. Thus, we have $\mathcal{A}_w'' \in \alpha(\mathcal{A})$.
	
	Finally, we turn to (ii).
	
	The crucial observation is that, since every subsequence of $w$ of length $n$ is in $\lang{\mathcal{A}}$, the subsequences $\sigma_1\dots\sigma_n$ and $\sigma_1\dots\sigma_{n-[(m-1)-(n+2)+1]}\sigma_{n+2}\dots\sigma_{m-1}$ are in $\lang{\mathcal{A}}$.
	
	This implies $\sigma_w''(q_0,\sigma_1\dots\sigma_n) = q_n$ and $\delta_w''(q_{n-[(m-1)-(n+2)+1]},\sigma_{n+2}\dots\sigma_{m-1}) = q_n$, which implies $\delta_w''(q_0,w) = q_d$. We are done with (ii). The proof of \cref{lem:fl_A_w''} is complete.
\end{proof}

Note that the construction of \cref{lem:fl_A_w''} critically hinges on the property $\sigma_{n+1} \neq \sigma_m$. This allows to circle back to different states when reaching $q_n$ after reading $\sigma_1\dots\sigma_n$ and $\sigma_1\dots\sigma_{m-1}$. If $\sigma_{n+1} = \sigma_m$ this construction is not possible. Therefore, we still have to handle words $w = \sigma_1\dots\sigma_m \in \Sigma_m$ with $n < m \leq n+(n-2)$, $|\sigma_{d+1}\dots\sigma_{m}|_{\sigma_m} > n-d$ and $\sigma_{n+1} = \sigma_m$ for which every subsequence of length $n$ is in $\lang{\mathcal{A}}$. As mentioned in \cref{subsec:fl_linearNonSafetyDFAs}, we will differentiate between two more cases, thus introducing two further lemmas.

We begin with:
\begin{lemma}
	\label{lem:fl_A_w'''}
	Let $m \in \natNum$ with $n < m \leq n+(n-2)$. Let $w = \sigma_1\dots\sigma_m \in \Sigma^m$ such that the following conditions hold:
	\begin{enumerate}
		\item $|\sigma_{d+1}\dots\sigma_{m}|_{\sigma_m} > n-d$.
		\item $\forall v \in \Sigma^n. v \notin \lang{\mathcal{A}} \Rightarrow w \in \lang{\mathcal{A}_v^!}$. Note that this implies that every subsequence of $w$ of length $n$ is in $\lang{\mathcal{A}}$.
		\item $\sigma_{n+1} = \sigma_m$.
		\item Let $x \in \natNumGeq{1}$ and $u \in \Sigma^*$ with $u \neq v\sigma_m$ for each $v \in \Sigma^*$ such that $w = \sigma_1\dots\sigma_du\sigma_m^x$. Let $b \in \{0,1\}$ with $b = 1$ if $\sigma_{d+1} \neq \sigma_m$ and $b = 0$ otherwise. Then $|u|_{\sigma_m} = |\sigma_{d+1}\dots\sigma_{m-x}|_{\sigma_m} < n-d-b$ holds.
	\end{enumerate}
	Let $l = |\sigma_{d+1}\dots\sigma_m|_{\sigma_m}$.
	
	We define $\mathcal{A}_w''' = (Q_w''',\Sigma,q_0,\delta_w''',F_w''')$ as follows:
	\begin{align*}
		Q_w''' &= \{q_0,\dots,q_n\} \\
		F_w''' &= Q_w''' \setminus \{q_d\} \\
		\delta_w'''(q_i,\sigma) &= \begin{cases}
			\delta(q_i,\sigma)	&\text{ if $i < d$ and $\delta(q_i,\sigma) \neq q_{n+1}$}\\
			q_i					&\text{ if $i < d$ and $\delta(q_i,\sigma) = q_{n+1}$}\\ 
								&\text{ (arbitrary definition possible)}\\
			q_{d+1}				&\text{ if $i = d$}\\
			q_{i+1}				&\text{ if $d < i < n-1$ and $\sigma = \sigma_m$}\\
			q_i					&\text{ if $d < i < n-1$ and $\sigma \neq \sigma_m$}\\
			q_{d-[x-(n-(d+b+(l-x))+1)]} &\text{ if $i = n$ and $\sigma = \sigma_m$}\\
			q_n					&\text{ else, thus if $i = n$ and $\sigma \neq \sigma_m$}
		\end{cases}.
	\end{align*}

	Then the following assertions hold:
	\begin{romanenumerate}
		\item $\mathcal{A}_w''' \in \alpha(\mathcal{A})$.
		\item $w \notin \lang{\mathcal{A}_w'''}$.
	\end{romanenumerate}
	Making (ii) more specific, the following holds:
	\begin{align*}
		\delta_w'''(q_0,\sigma_1 \dots \sigma_d) &= q_d \\
		\delta_w'''(q_0,\sigma_1 \dots \sigma_{d+1}) &= q_{d+1} \\
		\delta_w'''(q_0,\sigma_1 \dots \sigma_{m-x}) &= q_{d+b+(l-x)} \\
		\delta_w'''(q_0,\sigma_1 \dots \sigma_{m-x} \sigma_m^{n-(d+b+(l-x))}) &= q_n \\
		\delta_w'''(q_0,\sigma_1 \dots \sigma_{m-x} \sigma_m^{n-(d+b+(l-x))} \sigma_m) &= q_{d-[x-(n-(d+b+(l-x))+1)]}\\
		\delta_w'''(q_0,\sigma_1 \dots \sigma_m) &= q_d.
	\end{align*}
	\lipicsEnd
\end{lemma}
\begin{proof}
	Let $m \in \natNum$ and $w = \sigma_1\dots\sigma_m \in \Sigma^m$ be as required. Let $l,u,x,b$ be as defined above.
	
	Again we begin by showing that $\mathcal{A}_w'''$ is well defined. That is, we show that $d-[x-(n-(d+b+(l-x))+1)] \in \{0,\dots,n\}$. Note that $d-[x-(n-(d+b+(l-x))+1)] = n-l-b+1$.
	
	On the one hand, since $w \in \lang{\mathcal{A}_{\sigma_m^n}^!}$ holds, we have $l = |\sigma_{d+1}\dots\sigma_m|_{\sigma_m} \leq |w|_{\sigma_m} \leq n-1$. We further have $b \leq 1$. Therefore, we have $d-[x-(n-(d+b+(l-x))+1)] = n-l-b+1 \geq n-(n-1)-(1)+1 = 1$. On the other hand, we have $l > n-d$ and $b \geq 0$. Therefore, we have $d-[x-(n-(d+b+(l-x))+1)] = n-l-b+1 \leq n-(n-d+1)-(0)+1 = d < n$. Taken together we have $1 \leq d-[x-(n-(d+b+(l-x))+1)] \leq d < n$.
	
	First, note that per definition $n < n+(n-2)$ holds, which implies $2 < n$. Second, note that per definition $|\sigma_{d+1}\dots\sigma_m|_{\sigma_m} > n-d$ holds and, again, that we have $|\sigma_{d+1}\dots\sigma_m|_{\sigma_m} \leq n-1$. This implies $n-1 > n-d$ and thus $d > 1$.
	
	Therefore, we have $1 \leq d-[x-(n-(d+b+(l-x))+1)] \leq d < n$ and $1 < d$. The DFA $\mathcal{A}_w'''$ is well-defined.
	
	Next, we turn to (i). Again, we refer to \cref{lem:fl_A_w'}. We have $\size{\mathcal{A}_w'''} = n+1 < n+2 = \ind{\mathcal{A}}$ and for each $w' \in \Sigma^*$ we have $\delta_w'''(q_0,w') = q_d$ only if $\delta(q_0,w') = q_d$ or $\delta(q_0,w') = q_{n+1}$, which implies $w' \notin \lang{\mathcal{A}}$. Therefore, we have $\mathcal{A}_w''' \in \alpha(\mathcal{A})$.
	
	Finally, we turn to (ii). 
	
	Clearly, we have $\delta_w'''(q_0,\sigma_1\dots\sigma_d) = q_d$ and $\delta_w'''(q_0,\sigma_1\dots\sigma_{d+1}) = q_{d+1}$. Note that the first occurrence of a $\sigma_m$ in the subword $\sigma_{d+1}\dots\sigma_m$ is read in $q_{d+b}$. Note further that with $l > n-d$ there exist enough occurrences of $\sigma_m$ in $\sigma_{d+1}\dots\sigma_m$ to reach $q_n$ and afterwards $q_{d-[x-(n-(d+b+(l-x))+1)]}$. Also, by definition we have $l-x = |u|_{\sigma_m} < n-d-b$ and therefore $d+b+(l-x) < n$. This implies that $q_n$ is reached only after subword $u$ is read completely. Thus, we have $\delta_w'''(q_0,\sigma_1\dots\sigma_du) = q_{d+b+(l-x)}$. To then reach $q_n$, the subword $\sigma_m^{n-(d+b+(l-x))}$ has to be read. Therefore, we have $\delta_w'''(q_0,\sigma_1\dots\sigma_du\sigma_m^{n-(d+b+(l-x))}) = q_n$ and $\delta_w'''(q_0,\sigma_1\dots\sigma_du\sigma_m^{n-(d+b+(l-x))+1}) = q_{d-[x-(n-(d+b+(l-x))+1)]}$.
	
	After this, only the suffix $\sigma_m^{x-(n-(d+b+(l-x))+1)}$ remains, while the DFA is in the state $q_{d-[x-(n-(d+b+(l-x))+1)]}$. 
	
	Now we only need to show that $\delta_w'''(q_{d-[x-(n-(d+b+(l-x))+1)]},\sigma_m^{x-(n-(d+b+(l-x))+1)}) = q_d$ holds. It is sufficient to proof the following: $\forall i \in \{d-[x-(n-(d+b+(l-x))+1)]+1,\dots,d\}.\sigma_m \in \Sigma_{i-1,i}$.
	
	Note that $\sigma_1\dots\sigma_d\sigma_m^l$ with $l > n-d$ is a subsequence of $w$. This implies that $v = \sigma_1\dots\sigma_{n-l}\sigma_m^l$ is a subsequence of $w$. Note that $|v| = n$. Therefore, we have $v \in \lang{\mathcal{A}}$. This implies: $\forall i \in \{n-l+1,\dots,n\}.\sigma_m \in \Sigma_{i-1,i}$. 
	
	Now note that $b \leq 1$ implies $d-[x-(n-(d+b+(l-x))+1)]+1 = n-l-b+1+1 \geq n-l+1$. Therefore, we have: $\forall i \in \{d-[x-(n-(d+b+(l-x))+1)]+1,\dots,d\}.\sigma_m \in \Sigma_{i-1,i}$.
	
	As argued above, this implies $\delta_w'''(q_{d-[x-(n-(d+b+(l-x))+1)]},\sigma_m^{x-(n-(d+b+(l-x))+1)}) = q_d$ and therefore $\delta_w'''(q_0,w) = q_d$. We are done with (ii). The proof of \cref{lem:fl_A_w'''} is complete.
\end{proof}

Finally, we introduce:
\begin{lemma}
	\label{lem:fl_A_w''''}
	Let $m \in \natNum$ with $n < m \leq n+(n-2)$. Let $w = \sigma_1\dots\sigma_m \in \Sigma^m$ such that the following conditions hold:
	\begin{enumerate}
		\item $|\sigma_{d+1}\dots\sigma_{m}|_{\sigma_m} > n-d$.
		\item $\forall v \in \Sigma^n. v \notin \lang{\mathcal{A}} \Rightarrow w \in \lang{\mathcal{A}_v^!}$. Note that this implies that every subsequence of $w$ of length $n$ is in $\lang{\mathcal{A}}$.
		\item $\sigma_{n+1} = \sigma_m$.
		\item Let $x \in \natNumGeq{1}$ and $u \in \Sigma^*$ with $u \neq v\sigma_m$ for each $v \in \Sigma^*$ such that $w = \sigma_1\dots\sigma_du\sigma_m^x$. Let $b \in \{0,1\}$ with $b = 1$ if $\sigma_{d+1} \neq \sigma_m$ and $b = 0$ otherwise. Then $|u|_{\sigma_m} = |\sigma_{d+1}\dots\sigma_{m-x}|_{\sigma_m} \geq n-d-b$ holds.
		\item Let $i \in \{1,\dots,n\}$ be the maximal value with $\sigma_m \notin \Sigma_{i-1,i}$. Then $w \notin \lang{\mathcal{A}_{\sigma_m,i}}$ holds. 
	\end{enumerate}

	We define $\mathcal{A}_w'''' = (Q_w'''',\Sigma,q_0,\delta_w'''',F_w'''')$ as follows:
	\begin{align*}
		Q_w'''' &= \{q_0,\dots,q_n\} \\
		F_w'''' &= Q_w'''' \setminus \{q_d\} \\
		\delta_w''''(q_i,\sigma) &= \begin{cases}
			\delta(q_i,\sigma) 			&\text{ if $i < n$ and $\delta(q_i,\sigma) \neq q_{n+1}$}\\
			q_i							&\text{ if $i < n$ and $\delta(q_i,\sigma) = q_{n+1}$}\\
										&\text{ (arbitrary definition possible)}\\
			q_{n-[(m-x-1) - (n+2) + 1]}	&\text{ if $i = n$ and $\sigma = \sigma_m$}\\
			q_{d-x}						&\text{ if $i = n$ and $\sigma = \sigma_{m-x}$}\\
			q_n							&\text{ else, thus if $i = n$ and $\sigma \notin \{\sigma_m,\sigma_{m-x}$\}}\\
										&\text{ (arbitrary definition possible)}
		\end{cases}.
	\end{align*}

	Then the following assertions hold:
	\begin{romanenumerate}
		\item $\mathcal{A}_w'''' \in \alpha(\mathcal{A})$.
		\item $w \notin \lang{\mathcal{A}_w''''}$.
	\end{romanenumerate}
	Making (ii) more specific, the following holds:
	\begin{align*}
		\delta_w''''(q_0,\sigma_1 \dots \sigma_n) 		&= q_n \\
		\delta_w''''(q_0,\sigma_1 \dots \sigma_{n+1}) 	&= q_{n-[(m-x-1) - (n+2) + 1]} \\
		\delta_w''''(q_0,\sigma_1 \dots \sigma_{m'}) 	&= q_{n-[(m-x-1) - (n+2) + 1] + (m'-(n+1))}\\ 
														&\text{ for each $m' \in \{n+2,\dots,m-x-1\}$} \\
		\delta_w''''(q_0,\sigma_1 \dots \sigma_{m-x}) 	&= q_{d-x} \\
		\delta_w''''(q_0,\sigma_1 \dots \sigma_m) 		&= q_d.
	\end{align*}
	\lipicsEnd
\end{lemma}
\begin{proof}
	Let $m \in \natNum$ and $w = \sigma_1\dots\sigma_m \in \Sigma^m$ be as required. Let $u,x,b$ be as defined above. Define $l = |\sigma_{d+1}\dots\sigma_m|_{\sigma_m}$.
	
	Again, we show that $\mathcal{A}_w''''$ is well-defined. That is, we show that $n-[(m-x-1) - (n+2) + 1],d-x \in \{0,\dots,n\}$.
	
	We begin by showing $n-[(m-x-1) - (n+2) + 1] \in \{0,\dots,n\}$. Note that with $w = \sigma_1\dots\sigma_du\sigma_m^x$ we have $1 \leq x \leq m-d$, which implies $d \leq m-x \leq m-1$.
	
	To show that $n-[(m-x-1) - (n+2) + 1] \in \{0,\dots,n\}$ we prove $m-x \geq n+2$ by contradiction. Assume $m-x < n+2$.
	\begin{description}
		\item[Case 1: \normalfont{$m-x = d$.}] Then $w = \sigma_1\dots\sigma_d\sigma_m^x$ holds. Since $\sigma_{d+1} = \sigma_m$, we have $b = 0$. This implies $|u|_{\sigma_m} = 0 < 1 = n-(n-1) \leq n-d = n-d-b$, but $|u|_{\sigma_m} < n-d-b$ is a contradiction to the condition $|u|_{\sigma_m} \geq n-d-b$. We are done with Case 1.
		\item[Case 2: \normalfont{$m-x > d$.}] Then $d < m-x < n+2$ holds. Because of $\sigma_{m-x} \neq \sigma_m = \sigma_{n+1}$ we have $m-x \neq n+1$ and therefore $d < m-x < n+1$. This implies:
		\begin{align*}
			|u|_{\sigma_m}	&= |\sigma_{d+1}\dots\sigma_{m-x}|_{\sigma_m} \\
			&= (1-b) + |\sigma_{d+2}\dots\sigma_{m-x}|_{\sigma_m}\\
			&= (1-b) + |\sigma_{d+2}\dots\sigma_{m-x-1}|_{\sigma_m}\\
			&\leq (1-b) + |\sigma_{d+2}\dots\sigma_{m-x-1}|\\
			&= (1-b) + (m-x-1) - (d+2) + 1\\
			&= (m-x)-b-d-1\\
			&< (m-x)-b-d\\
			&\leq n-b-d.
		\end{align*}
		Again, $|u|_{\sigma_m} < n-d-b$ is a contradiction to the condition $|u|_{\sigma_m} \geq n-d-b$. We are done with Case 2.
	\end{description}
	With Cases 1 and 2 we have shown $m-x \geq n+2$ by contradiction. This means that the last occurrence of a letter unequal to $\sigma_m$ in $w$ is in the subword $\sigma_{n+2}\dots\sigma_{m-x}$.
	
	Now we return to our proof of $n-[(m-x-1) - (n+2) + 1] \in \{0,\dots,n\}$. Since $m-x \geq n+2$, we have $n-[(m-x-1) - (n+2) + 1] = 2n+2-(m-x) \leq 2n+2-(n+2) = n$. Additionally, since $m \leq n+(n-2)$, we have $n-[(m-x-1) - (n+2) + 1] = 2n+2-(m-x) \geq 2n+2-(n+(n-2))+x = x+4$. Taken together we have $x+4 \leq n-[(m-x-1) - (n+2) + 1] \leq n$.
	
	Note that $m-x \geq n+2$ and $m \leq n+(n-2)$, which implies $n+2+x \leq n+(n-2)$ and therefore $x+4 \leq n$. We have proven $n-[(m-x-1) - (n+2) + 1] \in \{0,\dots,n\}$.
	
	Now we proof $d-x \in \{0,\dots,n\}$. Note that $l = |\sigma_{d+1}\dots\sigma_m|_{\sigma_m} = |u\sigma_m^x|_{\sigma_m} = |u|_{\sigma_m}+x$ and therefore $x = l - |u|_{\sigma_m}$. Additionally, because of $w \in \lang{\mathcal{A}_{\sigma_m^n}^!}$ we have $l \leq n-1$. Finally, per requirement $|u|_{\sigma_m} \geq n-d-b$ holds. Taken together we have:
	\begin{align*}
		x 	&= l - |u|_{\sigma_m}\\
		&\leq  (n-1) - (n-d-b)\\
		&= d+b-1\\
		&\leq d+1-1\\
		&= d. 
	\end{align*}
	This implies $d-x \geq d-d = 0$.
	
	Further, because of $d \leq n-1$ and $x \geq 1$ we have $d-x \leq (n-1)-1 = n-2$.
	
	Taken together we have $0 \leq d-x \leq n-2$. Again, note that $n < m \leq n+(n-2)$ implies $2 < n$. We have proven $d-x \in \{0,\dots,n\}$.
	
	With $n-[(m-x-1) - (n+2) + 1],d-x \in \{0,\dots,n\}$ the DFA $\mathcal{A}_w''''$ is well-defined.
	
	Next, we consider (i). The proof of $\mathcal{A}_w'''' \in \alpha(\mathcal{A})$ is again similar to \cref{lem:fl_A_w'}, since $\size{\mathcal{A}_w''''} = n+1 < n+2 = \ind{\mathcal{A}}$ holds and for each $w' \in \Sigma^*$ we have $\delta_w''''(q_0,w') = q_d$ only if $\delta(q_0,w') = q_d$ or $\delta(q_0,w') = q_{n+1}$.
	
	Finally, we look at (ii). Obviously, we have $\delta_w''''(q_0,\sigma_1\dots\sigma_n) = q_n$ and $\delta_w''''(q_0,\sigma_1\dots\sigma_{n+1}) = q_{n-[(m-x-1) - (n+2) + 1]}$.
	
	Now note that $v = \sigma_1\dots\sigma_{n-[(m-x-1) - (n+2) + 1]}\sigma_{n+2}\dots\sigma_{m-x-1}$ is a subsequence of $w$ with $|v| = n$. This implies $v \in \lang{\mathcal{A}}$ and therefore $\delta_w''''(q_{n-[(m-x-1) - (n+2) + 1]},\sigma_{n+2}\dots\sigma_{m-x-1}) = q_n$. Therefore, we have $\delta_w''''(q_0,\sigma_1\dots\sigma_{m-x-1}) = q_n$ and $\delta_w''''(q_0,\sigma_1\dots\sigma_{m-x}) = q_{d-x}$.
	
	After reading $\sigma_1\dots\sigma_{m-x}$, only the suffix $\sigma_m^x$ is left to be read and we have to show that $\delta_w''''(q_{d-x},\sigma_m^x) = d$. It is sufficient to prove: $\forall j \in \{d-x+1,\dots,d\}. \sigma_m \in \Sigma_{j-1,j}$.
	
	Let $i \in \{1,\dots,n\}$ be the largest value with $\sigma_m \notin \Sigma_{i-1,i}$. We have to show that $i < d-x+1$. To do so we consider the DFA $\mathcal{A}_{\sigma_m,i}$. Per requirement $w \in \lang{\mathcal{A}_{\sigma_m,i}}$ holds.
	
	Again, we use that $\sigma_1\dots\sigma_d\sigma_m^l$ with $l > n-d$ is a subsequence of $w$, which implies that $v = \sigma_1\dots\sigma_{n-l}\sigma_m^l$ is a subsequence of $w$ with $|v| = n$. Therefore, we have $v \in \lang{\mathcal{A}}$. This implies: $\forall j \in \{n-l+1,\dots,n\}.\sigma_m \in \Sigma_{j-1,j}$. It follows that $i < n-l+1$ and thus $i-1 < n-l$.
	
	Now note that $l > n-d$ and $n-l > i-1$ implies $d > n-l > i-1$. This immediately implies $\delta_{\sigma_m,i}(q_0,\sigma_1\dots\sigma_d) = q_j$ for some $j \in \{0,\dots,n\}, j \geq i-1$.
	
	Next, let $t \in \{d+1,\dots,m\}$ be the smallest value such that $\sigma_t = \sigma_m$. Since $\sigma_{n+1} = \sigma_m$, we have $t \leq n+1$. Since $m-x \geq n+2$, we additionally have $t < m-x$.
	
	Since $\delta_{\sigma_m,i}(q_0,\sigma_1\dots\sigma_d) = q_j$ for some $j \in \{0,\dots,n\}, j \geq i-1$ and $w \in \lang{\mathcal{A}_{\sigma_m,i}}$, after the first occurrence of $\sigma_m$ in $\sigma_{d+1}\dots\sigma_m$ there can only be $(n-1)-(j+1) \leq (n-1)-i$ additional letters. That is: $|\sigma_{t}\dots\sigma_{m}| \leq (n-1) - j \leq (n-1) - (i-1) = n-i$.
	
	This implies:
	\begin{align*}
		l	&= |\sigma_{d+1}\dots\sigma_m|_{\sigma_m}\\
		&= |\sigma_{t}\dots\sigma_m|_{\sigma_m}\\
		&= |\sigma_{t}\dots\sigma_{m-x-1}\sigma_{m-x}\sigma_{m-x+1}\dots\sigma_m|_{\sigma_m}\\
		&= |\sigma_{t}\dots\sigma_{m-x-1}\sigma_{m-x+1}\dots\sigma_m|_{\sigma_m}\\
		&\leq |\sigma_{t}\dots\sigma_{m-x-1}\sigma_{m-x+1}\dots\sigma_m|\\
		&= |\sigma_t \dots \sigma_m|-1\\
		&\leq (n-i)-1.
	\end{align*}

	With $x = l - |u|_{\sigma_m}$ and $l \leq n-i-1$ and $|u|_{\sigma_m} \geq n-d-b$ we get:
	\begin{align*}
		x	&= l - |u|_{\sigma_m}\\
		&\leq (n-i-1) - (n-d-b)\\
		&= d+b-i-1\\
		&\leq d+1-i-1\\
		&= d-i.
	\end{align*}
	Thus, we have $1 \leq x \leq d-i$. Note that with $i-1 < n-l < d$, which implies $i \leq d-1$, we indeed have $1 \leq d-i$.
	
	In conclusion, we have $x \leq d-i$ and therefore $i \leq d-x < d-x+1$, which implies: $\forall j \in \{d-x+1,\dots,d\}.\sigma_m \in \Sigma_{j-1,j}$. This implies $\delta_w''''(q_{d-x},\sigma_m^x) = q_d$ and thus $\delta_w''''(q_0,w) = q_d$. We are done with (ii). The proof of \cref{lem:fl_A_w''''} is complete.
\end{proof}

Note that the case distinction in \cref{lem:fl_A_w''',lem:fl_A_w''''} is found in the fourth condition. That is, $|u|_{\sigma_m} < n-d-b$ and $|u|_{\sigma_m} \geq n-d-b$. The fifth condition in \cref{lem:fl_A_w''''} uses the fact that our decomposition contains the DFAs of the form $\mathcal{A}_{\sigma,i}$. This allows us to restrict the construction of DFAs $\tilde{\mathcal{A}}_w$ in general to words $w$ which are not rejected by the DFAs of the form $\mathcal{A}_{\sigma,i}$. This can be seen in \cref{lem:fl_tildeA_w}, which explicitly states this condition. But since we only need this condition in \cref{lem:fl_A_w''''}, while for \cref{lem:fl_A_w',lem:fl_A_w'',lem:fl_A_w'''} it is enough to require $|w| \leq n+(n-2)$, we only state this requirement in \cref{lem:fl_A_w''''}.

Note further that we indeed need the two lemmas, \cref{lem:fl_A_w''',lem:fl_A_w''''}. 

The construction in \cref{lem:fl_A_w'''} does not work for words with $|u|_{\sigma_m} \geq n-d-b$, since after arriving in $q_n$ and circling back to $q_{d-[x-(n-(d+b+(l-x))+1)]}$, the remaining suffix is not necessarily of the form $\sigma_m^k$ for $k \in \natNum$. This is problematic, since we cannot guarantee that for the remaining letters $\sigma \neq \sigma_m$ the DFA $\mathcal{A}_w'''$ advances exactly one state. Additionally, we cannot even be certain how many letters are left to be read.

The construction in \cref{lem:fl_A_w''''} does not work for words with $|u|_{\sigma_m} < n-d-b$, since after arriving in $q_n$ for the second time and circling back to $q_{d-x}$, we cannot use the subsequence argument to ensure that $\sigma_m \in \Sigma_{j-1,j}$ for each $j \in \{d-x+1,d\}$. In other words, it is possible that not enough occurences of the letter $\sigma_m$ were read before $\mathcal{A}_w''''$ circles back to $q_{d-x}$.

Therefore, we indeed need the case distinction captured by \cref{lem:fl_A_w''',lem:fl_A_w''''}.

With \cref{lem:fl_A_w',lem:fl_A_w'',lem:fl_A_w''',lem:fl_A_w''''} in hand, the proof of \cref{lem:fl_tildeA_w} is obvious:
\begin{proof}[Proof of \cref{lem:fl_tildeA_w}]
	Let $w \in \Sigma^*$ be as required. That is, $|w| > n$ with $w \in \lang{\mathcal{A}_v^!}$ for each $v \notin L,|v|=n$ and $w \in \bigcap_{\sigma \in \Sigma} \lang{\mathcal{A}_{\sigma,i_\sigma}}$. Note that this implies $n < |w| \leq n+(n-2)$.
	
	Then the word $w$ satisfies the conditions of one lemma out of \cref{lem:fl_A_w',lem:fl_A_w'',lem:fl_A_w''',lem:fl_A_w''''}. Select the respective DFA as $\tilde{\mathcal{A}}_w$. Then DFA $\tilde{\mathcal{A}}_w$ witnesses the validity of \cref{lem:fl_tildeA_w}. We are done.
\end{proof}

We have proven \cref{lem:fl_A_0A_dA_myUnderbari,lem:fl_A_sigmai,lem:fl_A_w^!,lem:fl_tildeA_w}. This leads to a simple proof of \cref{cla:fl_characterization} (\ref{cla_ass:fl_characterization_non-sigmaN+non-safety}):
\begin{proof}[Proof of \cref{cla:fl_characterization} (\ref{cla_ass:fl_characterization_non-sigmaN+non-safety})]
	We have to show that $\mathcal{A}$ is composite.
	
	We start by defining $X',X'',X''',X''''$ as the sets of words for which the conditions of the respective lemma out of \cref{lem:fl_A_0A_dA_myUnderbari,lem:fl_A_sigmai,lem:fl_A_w^!,lem:fl_tildeA_w} hold. 
	We further define $X^! = \{w \in \Sigma^n \setDel w \notin L\}$. 
	Finally, for each $\sigma \in \Sigma$ we define $i_\sigma = \maxOp{\{i \in \{1,\dots,n\} \setDel \sigma \notin \Sigma_{i-1,i}\}}$.
	
	We prove:
	\begin{align*}
		L 	&= \lang{\mathcal{A}_0}
		\cap \lang{\hat{\mathcal{A}}_d}\\
		&\cap \bigcap_{m=1}^{n-1} \bigcap_{\myUnderbar{i} \in I_m} \lang{\mathcal{A}_{\myUnderbar{i}}}
		\cap \bigcap_{w \in X^!} \lang{\mathcal{A}_w^!}
		\cap \bigcap_{\sigma\in\Sigma} \lang{\mathcal{A}_{\sigma,i_\sigma}}\\
		&\cap \bigcap_{w \in X'} \lang{\mathcal{A}_w'}
		\cap \bigcap_{w \in X''} \lang{\mathcal{A}_w''}
		\cap \bigcap_{w \in X'''} \lang{\mathcal{A}_w'''}
		\cap \bigcap_{w \in X''''} \lang{\mathcal{A}_w''''}.
	\end{align*}
	We denote the language created by the decomposition on the right hand side with $L_\cap$.

	Note that with \cref{lem:fl_A_0A_dA_myUnderbari,lem:fl_A_sigmai,lem:fl_A_w^!,lem:fl_tildeA_w} each of the DFAs used for the decomposition is in $\alpha$($\mathcal{A}$). Therefore, they are sufficiently small and $L \subseteq L_\cap$ holds. It remains to show that $L_\cap \subseteq L$.
	
	Let $w \in \Sigma^*$ with $w \notin L$. First, assume that:
	\begin{align*}
		w &\notin \lang{\mathcal{A}_0}
		\cap \lang{\hat{\mathcal{A}}_d}\\
		&\cap \bigcap_{m=1}^{n-1} \bigcap_{\myUnderbar{i} \in I_m} \lang{\mathcal{A}_{\myUnderbar{i}}}
		\cap \bigcap_{w \in X^!} \lang{\mathcal{A}_w^!}
		\cap \bigcap_{\sigma\in\Sigma} \lang{\mathcal{A}_{\sigma,i_\sigma}}.
	\end{align*}
	Then we are done immediately. 
	
	Thus, we assume that $w$ is not rejected by these DFAs. This clearly implies $w \in X' \cup X'' \cup X''' \cup X''''$. Therefore, we have:
	\begin{align*}
		w \notin \bigcap_{w \in X'} \lang{\mathcal{A}_w'}
		\cap \bigcap_{w \in X''} \lang{\mathcal{A}_w''}
		\cap \bigcap_{w \in X'''} \lang{\mathcal{A}_w'''}
		\cap \bigcap_{w \in X''''} \lang{\mathcal{A}_w''''}.
	\end{align*}

	This implies $w \notin L_\cap$. We have shown $L_\cap \subseteq L$. The proof of \cref{cla:fl_characterization} (\ref{cla_ass:fl_characterization_non-sigmaN+non-safety}) is complete.
\end{proof}

To summarize, we have introduced four additional lemmas, \cref{lem:fl_A_w',lem:fl_A_w'',lem:fl_A_w''',lem:fl_A_w''''}, and using these have proven \cref{lem:fl_A_0A_dA_myUnderbari,lem:fl_A_sigmai,lem:fl_A_w^!,lem:fl_tildeA_w}. After this, the proof of \cref{cla:fl_characterization} (\ref{cla_ass:fl_characterization_non-sigmaN+non-safety}) was obvious.

\subsection{Proof of \texorpdfstring{\cref{cla:fl_characterization}}{Claim \ref{cla:fl_characterization}} (\ref{cla_ass:fl_characterization_non-sigmaN+CEP})}
\label{subsec:fl_charac_d}
Next, we consider \cref{cla:fl_characterization} (\ref{cla_ass:fl_characterization_non-sigmaN+CEP}). We still assume that $\mathcal{A}$ is of the form described at the end of \cref{subsec:fl_charac_aAndb}. That is, $\mathcal{A}$ is of the form displayed in \cref{fig:fl_minLinDFA}. 

We prove that $\mathcal{A}$ is composite if it has the CEP. That is, if for every $w \in \lang{\mathcal{A}}$ with $|w| = n$ there exists a compression $w'$ of $w$ such that every extension of $w'$ is rejected by $\mathcal{A}$. Note that this is equivalent to: For every $w = \sigma_1\dots\sigma_n \in \lang{\mathcal{A}}$ there exist $i \in \{0,\dots,n-2\}, l \in \{2,\dots,n-i\}$ such that $\delta(q_0,\sigma_1\dots\sigma_i\sigma_{i+l}\dots\sigma_n) \in \{q_n,q_{n+1}\}$.

We prove the compositionality of $\mathcal{A}$ by specifying a DFA $\mathcal{A}_{i,l}$ for each possible pair $i \in \{0,\dots,n-2\}, l \in \{2,\dots,n-i\}$. These DFAs will reject the extensions of words $w \in L, |w|=n$. In addition to these DFAs we will use the DFAs $\mathcal{A}_0,\hat{\mathcal{A}}_d,\mathcal{A}_{\myUnderbar{i}}$ which we have discussed in \cref{subsec:fl_charac_c}.
\begin{figure}[t]
	\centering
	\begin{tikzpicture}[node distance=1.875cm]
	\tiny
	\node[state, initial, accepting] 					(q0) 	{$q_0$};
	\node[state, right of=q0, accepting]				(q1)	{$q_1$};
	\node[state, right of=q1, draw=none]				(empty-node1) 	{};
	\node[state, right of=empty-node1, accepting] 		(qi-1) 	{$q_{i-1}$};
	\node[state, right of=qi-1, accepting] 				(qi) 	{$q_{i}$};
	\node[state, right of=qi, accepting] 				(qi+1) 	{$q_{i+1}$};
	\node[state, right of=qi+1, draw=none]				(empty-node2) 	{};
	\node[state, right of=empty-node2, accepting] 		(qi+l-2){$q_{i+l-2}$};
	\node[state, below of=qi, accepting] 				(qi+l){$q_{i+l}$};
	\node[state, right of=qi+l, accepting] 				(qi+l+1){$q_{i+l+1}$};
	\node[state, right of=qi+l+1, draw=none]			(empty-node3) 	{};
	\node[state, right of=empty-node3, accepting] 		(qn-1){$q_{n-1}$};
	\node[state, below of=qi+l, accepting] 				(qn){$q_{n}$};
	\node[state, below of=q0, draw=none]				(empty-node4) 	{};
	\node[state, below of=empty-node4] 					(qn+1){$q_{n+1}$};
	
	\draw	(q0)	edge[above]							node{$\Sigma_{0,1}$}								(q1);
	\draw	(q0)	edge[above, bend left=30]			node{$\Sigma_{0,i-1}$}								(qi-1);
	\draw	(q0)	edge[above, bend left=45]			node{$\Sigma_{0,i},\Sigma_{0,i+l-1}$}								(qi);
	\draw	(q0)	edge[above, bend left=60]			node{$\Sigma_{0,i+1}$}								(qi+1);
	\draw	(q0)	edge[above, bend left=75]			node{$\Sigma_{0,i+l-2}$}							(qi+l-2);
	\draw	(q0)	edge[below, bend right=15]			node{$\Sigma_{0,i+l}$}							(qi+l);
	\draw	(q0)	edge[below]							node{$\Sigma_{0,i+l+1}$}							(qi+l+1);
	\draw	(q0)	edge[above]							node{$\Sigma_{0,n-1}$}							(qn-1);
	\draw	(q0)	edge[below, bend right=15]			node{$\Sigma_{0,n}$}							(qn);
	\draw	(q0)	edge[right]							node{$\Sigma_{0,n+1}$}							(qn+1);
	\draw	(q1)	edge[dashed]						node{}								(qi-1);
	\draw	(qi-1)	edge[above]							node{$\Sigma_{i-1,i}$}								(qi);
	\draw	(qi)	edge[above]							node{$\complementOp{\Sigma'}$}	(qi+1);
	\draw	(qi)	edge[right]							node{$\Sigma'$}	(qi+l);
	\draw	(qi+1)	edge[above, dashed]					node{$\Sigma$}								(qi+l-2);
	\draw	(qi+l-2)edge[above, bend right=45]			node{$\Sigma$}							(qi);
	\draw	(qi+l)	edge[below]							node{$\Sigma$}							(qi+l+1);
	\draw	(qi+l+1)edge[below, dashed]					node{$\Sigma$}								(qn-1);
	\draw	(qn-1)	edge[below]							node{$\Sigma$}							(qn);
	\draw	(qn)	edge[above]							node{$\Sigma$}							(qn+1);
	\end{tikzpicture}
	\caption{DFA $\mathcal{A}_{i,l}$ for $i \in \{0,\dots,n-2\}, l \in \{2,\dots,n-i\}$. The transitions exiting $q_0$ are given. For $j \in \{1,\dots,i-1\}$, the transitions are omitted for readability purposes. It is $\Sigma' = \bigcup_{j=i+l}^{n+1}\Sigma_{i,j}$.}
	\label{fig:A_il}
\end{figure}
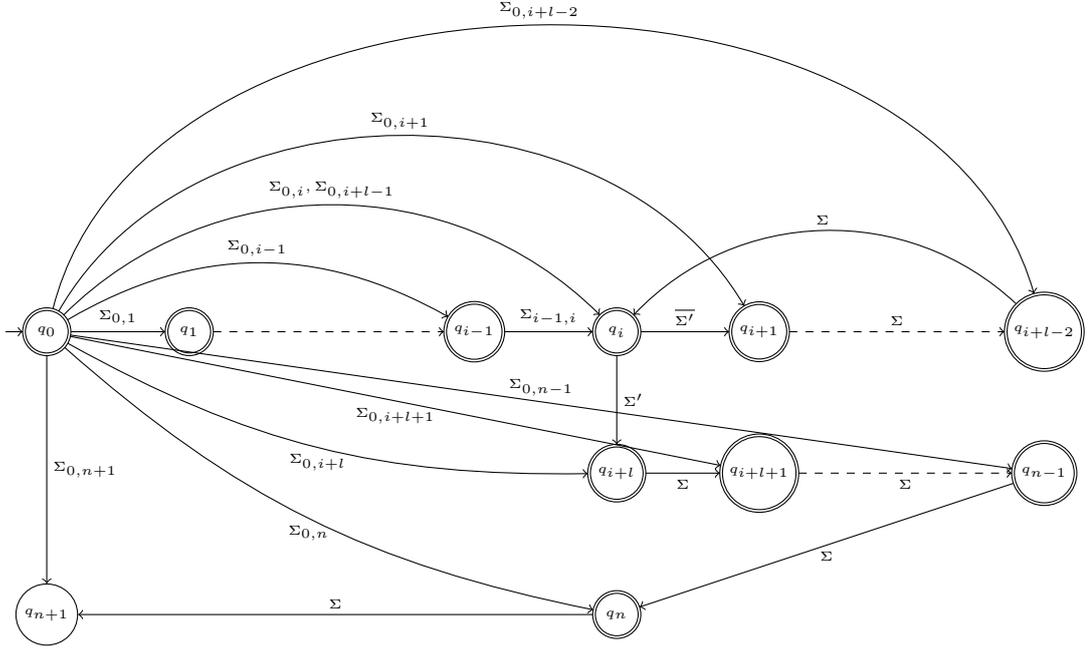

The construction of DFA $\mathcal{A}_{i,l}$ is specified in \cref{fig:A_il}. Additionally, we provide a formal definition. Let $i \in \{0,\dots,n-2\}, l \in \{2,\dots,n-i\}$ and let $\Sigma' = \bigcup_{j=i+l}^{n+1}\Sigma_{i,j}$. We define $\mathcal{A}_{i,l} = (Q_{i,l},\Sigma,q_0,\delta_{i,l},F_{i,l})$ where:
\begin{align*}
	Q_{i,l}			&= \{q_0,\dots,q_{i+l-2},q_{i+l},\dots,q_{n+1}\}\\
	F_{i,l}			&= Q_{i,l} \setminus \{q_{n+1}\}\\
	\delta_{i,l}(q_j,\sigma)&=
	\begin{cases}
		\delta(q_j,\sigma)		&\text{ if $j < i$ and $\delta(q_j,\sigma) \neq q_{i+l-1}$}\\
		q_i						&\text{ if $j < i$ and $\delta(q_j, \sigma) = q_{i+l-1}$}\\
		q_{i+1}					&\text{ if $j = i$ and $\sigma \notin \Sigma'$}\\
		q_{i+l}					&\text{ if $j = i$ and $\sigma \in \Sigma'$}\\
		q_{j+1}					&\text{ if $i < j < n+1$ and $j \neq i+l-2$}\\
		q_i						&\text{ if $j = i+l-2$}\\
		q_{n+1}					&\text{else, thus if $j = n+1$}
	\end{cases}.
\end{align*}

The following lemma states properties of $\mathcal{A}_{i,l}$:
\begin{lemma}
	\label{lem:fl_A_il}
	Let $i \in \{0,\dots,n-2\}$ and $l \in \{2,\dots,n-i\}$. Consider the DFA $\mathcal{A}_{i,l}$. The following assertions hold:
	\begin{romanenumerate}
		\item $\mathcal{A}_{i,l} \in \alpha(\mathcal{A})$.
		\item Let $w = \sigma_1\dots\sigma_n \in \Sigma^n$ with $\delta(q_0,\sigma_1\dots\sigma_i) = q_i$ and $\sigma_{i+l} \in \Sigma'$. The DFA $\mathcal{A}_{i,l}$ then rejects every extension of $w$. That is, $wv \notin \lang{\mathcal{A}_{i,l}}$ for each $v \in \Sigma^+$.\lipicsEnd
	\end{romanenumerate}
\end{lemma}
\begin{proof}
	Let $i \in \{0,\dots,n-2\}, l \in \{2,\dots,n-i\}$. Consider the DFA $\mathcal{A}_{i,l} = (Q_{i,l},\Sigma,q_0,\delta_{i,l},F_{i,l})$.
	
	First, we turn to (i) and argue that $\mathcal{A}_{i,l} \in \alpha(\mathcal{A})$. Note that $\size{\mathcal{A}_{i,l}} = n+1 < n+2 = \ind{\mathcal{A}}$. Therefore, we only have to argue that $\lang{\mathcal{A}} \subseteq \lang{\mathcal{A}_{i,l}}$. But this is easy to see, since for every $w \in \Sigma^*$ and $s,t \in \{0,\dots,n+1\}, t \neq i+l-1$ such that $\delta(q_0,w) = q_s$ and $\delta_{i,l}(q_0,w) = q_t$, we have $s \geq t$. 
	Since this implies $\delta_{i,l}(q_0,w) = q_{n+1}$ only if $\delta(q_0,w) = q_{n+1}$ and since $q_{n+1}$ is the only rejecting state of $\mathcal{A}_{i,l}$, we have $\lang{\mathcal{A}} \subseteq \lang{\mathcal{A}_{i,l}}$.
	We have shown that $\mathcal{A}_{i,l} \in \alpha(\mathcal{A})$ and are done with (i).
	
	Second, we turn to (ii). Let $w = \sigma_1\dots\sigma_n \in \Sigma^n$ with $\delta(q_0,\sigma_1\dots\sigma_i) = q_i$ and $\sigma_{i+l} \in \Sigma'$. Then we have $\delta_{i,l}(q_0,\sigma_1\dots\sigma_{i+l}) = q_j$ for a $j \geq i+l$. To be more precise, we have $j = i+l$ if $\sigma_{i+1} \notin \Sigma'$, otherwise we have $j > i+l$. 
	
	Note that when in state $q_j$ the DFA $\mathcal{A}_{i,l}$ rejects after reading $(n+1)-j \leq (n+1)-(i+l)$ additional letters. Therefore, every extension of $w$ is rejected by $\mathcal{A}_{i,l}$. We are done with (ii). The proof of \cref{lem:fl_A_il} is complete.
\end{proof}

Note that \cref{lem:fl_A_il} (ii) critically hinges on $\sigma_{i+l} \in \Sigma'$. We now introduce a rather technical lemma, which allows us to use the DFAs $\mathcal{A}_{i,l}$ to reject the extensions of words $w \in L, |w|=n$.
\begin{lemma}
	\label{lem:fl_i+lCondition}
	Let $w = \sigma_1\dots\sigma_n \in \Sigma^n$ with $w \in \lang{\mathcal{A}}$ such that there exist $i \in \{0,\dots,n-2\},l \in \{2,\dots,n-i\}$ with $\delta(q_0,\sigma_1\dots\sigma_{i}\sigma_{i+l}\dots\sigma_n) \in \{q_n,q_{n+1}\}$. Let $i$ be the maximal value for which such an $l$ exists. Then there exists a $j \in \{i+l,\dots,n+1\}$ such that $\sigma_{i+l} \in \Sigma_{i,j}$. That is, $\sigma_{i+l} \in \Sigma'$ holds.\lipicsEnd
\end{lemma}
\begin{proof}
	Before we start, we define $\indexOp{q_j} = j$ for each $q_j \in Q$.
	
	Let $w,i,l$ be as required.
	
	Note that there is a $k \in \{0,\dots,n-(i+l)\}$ such that $\indexOp{\delta(q_0,\sigma_1\dots\sigma_i\sigma_{i+l}\dots\sigma_{i+l+k})} \geq i+l+k$. This is obvious, since for $k = n-(i+l)$ we have:
	\begin{align*}
		&\indexOp{\delta(q_0,\sigma_1\dots\sigma_i\sigma_{i+l}\dots\sigma_{i+l+k})}\\
		= 		&\indexOp{\delta(q_0,\sigma_1\dots\sigma_i\sigma_{i+l}\dots\sigma_n)}\\
		\geq 	&n\\ 
		= 		&(i+l) + (n-(i+l))\\
		= 		&(i+l) + k.
	\end{align*}

	With this observation in hand, we can turn to the actual proof.	
	We employ a proof by contradiction and therefore assume $\sigma_{i,l} \in \Sigma_{i,j}$ for a $j \in \{i+1,\dots,i+l-1\}$.
	
	We briefly consider the case $i = n-2$. Then $\delta(q_0,\sigma_1\dots\sigma_i\sigma_{i+l}) = \delta(q_0,\sigma_1\dots\sigma_{n-2}\sigma_n) = q_{n-1}$ holds. This is a contradiction to $\delta(q_0,\sigma_1\dots\sigma_i\sigma_{i+l}) \in \{q_n,q_{n+1}\}$. Therefore, we can assume $i < n-2$.
	
	We will show that $\sigma_{i,l} \in \Sigma_{i,j}$ for a $j \in \{i+1,\dots,i+l-1\}$ and $i < n-2$ implies the existence of $i' \in \{0,\dots,n-2\},l' \in \{2,\dots,n-i'\}$ such that $\delta(q_0,\sigma_1\dots\sigma_{i'}\sigma_{i'+l'}\dots\sigma_n) \in \{q_n,q_{n+1}\}$ and $i' > i$, which contradicts the selection of $i$ as the largest possible value.
	
	Note that with $\sigma_{i+l} \in \Sigma_{i,j}$ we have $\indexOp{\delta(q_0,\sigma_1\dots\sigma_i\sigma_{i+l})} = \indexOp{\delta(q_i,\sigma_{i+l})} = \indexOp{q_j} = j < i+l$. Additionally, with the above observation there exists a $k \in \{0,\dots,n-(i+l)\},k>0$ such that $\indexOp{\delta(q_0,\sigma_1\dots\sigma_i\sigma_{i+l}\dots\sigma_{i+l+k})} \geq i+l+k$. Let $k$ be the minimal value for which this holds. Then we have:
	\begin{align*}
		&k \geq 1,\\
		&\indexOp{\delta(q_0,\sigma_1\dots\sigma_i\sigma_{i+l}\dots\sigma_{i+l+k-1})} < i+l+k-1,\\
		&\indexOp{\delta(q_0,\sigma_1\dots\sigma_i\sigma_{i+l}\dots\sigma_{i+l+k})} \geq i+l+k.
	\end{align*}

	We define $i' = \indexOp{\delta(q_0,\sigma_1\dots\sigma_i\sigma_{i+l}\dots\sigma_{i+l+k-1})}$. Note that $i' < i+l+k-1 < i+l+k \leq n$ and thus $i' \leq n-2$. Further, we have:
	\begin{align*}
		i' 		&= \indexOp{\delta(q_0,\sigma_1\dots\sigma_i\sigma_{i+l}\dots\sigma_{i+l+k-1})}\\
		&\geq \indexOp{\delta(q_0,\sigma_1\dots\sigma_i\sigma_{i+l}\dots\sigma_{i+l+1-1})}\\
		&= \indexOp{\delta(q_0,\sigma_1\dots\sigma_i\sigma_{i+l})}\\
		&= \indexOp{q_j}\\
		&= j\\
		&> i.
	\end{align*}
	Thus, we have $i' \in \{i+1,\dots,n-2\}$ and $i < n-2$. Therefore, we have $i' \in \{0,\dots,n-2\}$ with $i' > i$.
	
	Now note that:
	\begin{align*}
		&\delta(q_0,\sigma_1\dots\sigma_{i'}\sigma_{i+l+k}\dots\sigma_n)\\
		=		&\delta(q_{i'},\sigma_{i+l+k}\dots\sigma_n)\\
		=		&\delta(\delta(q_0,\sigma_1\dots\sigma_i\sigma_{i+l}\dots\sigma_{i+l+k-1}),\sigma_{i+l+k}\dots\sigma_n)\\
		=		&\delta(q_0,\sigma_1\dots\sigma_i\sigma_{i+l}\dots\sigma_{i+l+k-1}\sigma_{i+l+k}\dots\sigma_n)\\
		=		&\delta(q_0,\sigma_1\dots\sigma_i\sigma_{i+l}\dots\sigma_n) \in \{q_n,q_{n+1}\}.
	\end{align*}

	Let $l' = (i+l+k) - i'$. Then $\delta(q_0,\sigma_1\dots\sigma_{i'}\sigma_{i'+l'}\dots\sigma_n) \in \{q_n,q_{n+1}\}$ holds. Additionally, as explained above, we have $i' \in \{0,\dots,n-2\}$ with $i' > i$. Further, because of $i' < i+l+k-1$ we have $l' = (i+l+k)-i' > (i+l+k) - (i+l+k-1) = 1$. That is, $l' \geq 2$. Finally, because of $k \leq n-(i+l)$ we have $l'+i' = (i+l+k-i') + i' = i+l+k \leq i+l+(n-(i+l)) = n$. That is, $l'+i' \leq n$ and therefore $l' \leq n-i'$. Taken together we get $l' \in \{2,\dots,n-i'\}$.
	
	To summarize, we have $i' \in \{0,\dots,n-2\},l' \in \{2,\dots,n-i'\}$ with $i' > i$ such that  $\delta(q_0,\sigma_1\dots\sigma_{i'}\sigma_{i'+l'}\dots\sigma_n) \in \{q_n,q_{n+1}\}$. This is a contradiction to the selection of $i$ as the largest possible value.
	
	In conclusion, we have proven by contradiction that there exists a $j \in \{i+l,\dots,n+1\}$ such that $\sigma_{i+l} \in \Sigma_{i,j}$. That is, $\sigma_{i+l} \in \Sigma'$. The proof of \cref{lem:fl_i+lCondition} is complete.
\end{proof}

With \cref{lem:fl_A_il,lem:fl_i+lCondition} in hand, it is easy to prove \cref{cla:fl_characterization} (\ref{cla_ass:fl_characterization_non-sigmaN+CEP}).
\begin{proof}[Proof of \cref{cla:fl_characterization} (\ref{cla_ass:fl_characterization_non-sigmaN+CEP})]
	We assume that $\mathcal{A}$ has the CEP. We prove the compositionality of $\mathcal{A}$ by showing:
	\begin{align*}
		L &= \lang{\mathcal{A}_0}
		\cap \lang{\hat{\mathcal{A}}_d}\\
		&\cap \bigcap_{m=1}^{n-1} \bigcap_{\myUnderbar{i} \in I_m} \lang{\mathcal{A}_{\myUnderbar{i}}}
		\cap \bigcap_{i\in\{0,\dots,n-2\},l\in\{2,\dots,n-i\}} \lang{\mathcal{A}_{i,l}},
	\end{align*}
	where $d \in \{0,\dots,n\}$ can be arbitrarily selected. We denote the language created by the decomposition on the right hand side with $L_\cap$.
	
	The proof is similar to the proof of \cref{cla:fl_characterization} (\ref{cla_ass:fl_characterization_non-sigmaN+non-safety}).
	
	Note that with \cref{lem:fl_A_0A_dA_myUnderbari,lem:fl_A_il} each of the DFAs used for the decomposition is in $\alpha(\mathcal{A})$, which implies that they are sufficiently small and that $L \subseteq L_\cap$ holds. Therefore, we only need to show $L_\cap \subseteq L$.
	
	Let $w = \sigma_1\dots\sigma_m \in \Sigma^m$ with $w \notin L$. Similar to the proof of \cref{cla:fl_characterization} (\ref{cla_ass:fl_characterization_non-sigmaN+non-safety}) we begin by assuming:
	\begin{align*}
		w &\notin \lang{\mathcal{A}_0}
		\cap \lang{\hat{\mathcal{A}}_d}\\
		&\cap \bigcap_{m=1}^{n-1} \bigcap_{\myUnderbar{i} \in I_m} \lang{\mathcal{A}_{\myUnderbar{i}}}.
	\end{align*}
	Then we are done immediately.
	
	Thus, we assume that $w$ is not rejected by these DFAs. With \cref{lem:fl_A_0A_dA_myUnderbari} (ii) this implies that $w$ is an extension of a word $u \in L, |u| = n$. That is, $m > n$ and $\sigma_1\dots\sigma_n \in L$. Then per requirement there exist $i \in \{0,\dots,n-2\},l \in \{2,\dots,n-i\}$ such that $\delta(q_0,\sigma_1\dots\sigma_i\sigma_{i+l}\dots\sigma_n) \in \{q_n,q_{n+1}\}$. Let $i$ be the largest value for which such an $l$ exists. With \cref{lem:fl_i+lCondition} this implies $\sigma_{i+l} \in \Sigma_{i,j}$ for a $j \in \{i+l,\dots,n+1\}$. Then with \cref{lem:fl_A_il} we have $w \notin \lang{\mathcal{A}_{i,l}}$, since $\delta(q_0,\sigma_1\dots\sigma_i) = q_i$ and $\sigma_{i+l} \in \bigcup_{j=i+l}^{n+1}\Sigma_{i,j}$. Therefore, we have:
	\begin{align*}
		w \notin \bigcap_{i\in\{0,\dots,n-2\},l\in\{2,\dots,n-i\}} \lang{\mathcal{A}_{i,l}}.
	\end{align*}
	
	This implies $w \notin L_\cap$. Thus, we have $L_\cap \subseteq L$. The proof of \cref{cla:fl_characterization} (\ref{cla_ass:fl_characterization_non-sigmaN+CEP}) is complete.
\end{proof}

We have proven that $\mathcal{A}$ is composite if it has the CEP.
It is noteworthy that we have made no requirements regarding the accepting and rejecting states of $\mathcal{A}$, since the CEP implies compositionality regardless of these states. We have introduced a new type of DFA, $\mathcal{A}_{i,l}$, which rejects the extensions of words $w \in L, |w| = n$. Therefore, if $\mathcal{A}$ has the CEP then the more complicated construction for non-safety DFAs detailed in \cref{subsec:fl_charac_c} is not necessary.

\subsection{Proof of \texorpdfstring{\cref{cla:fl_characterization}}{Claim \ref{cla:fl_characterization}} (\ref{cla_ass:fl_characterization_linear+safety+non-CEP})}
\label{subsec:fl_charac_e}
Finally, we consider \cref{cla:fl_characterization} (\ref{cla_ass:fl_characterization_linear+safety+non-CEP}). Our goal is to show that $\mathcal{A}$ is prime if it is a safety DFA and does not have the CEP.

Therefore, we assume that $\mathcal{A}$ is a safety DFA. That is, $F = Q \setminus \{q_{n+1}\}$. We begin by proving primality of $\mathcal{A}$ if $\mathcal{A}$ does not have the CEP and if another condition is met as well. Then we will show that this condition is implied by $\mathcal{A}$ not having the CEP. Thus, the ADFA $\mathcal{A}$ is prime if it does not have the CEP.

First, we prove:
\begin{lemma}
	\label{lem:fl_strongNecessaryCondition}
	The ADFA $\mathcal{A}$ is prime, if:
	\begin{enumerate}
		\item $\neg (\Sigma_{n-1,n} \subseteq \bigcup_{j=0}^{n-1} \Sigma_{j,n+1})$, and
		\item it does not have the CEP.\lipicsEnd
	\end{enumerate}
\end{lemma}
\begin{proof}
	Assume that the two conditions outlined in the lemma hold. 
	
	With the first condition there is a $\sigma \in \Sigma_{n-1,n}$ such that $\sigma \notin \bigcup_{j=0}^{n-1} \Sigma_{j,n+1}$. With the second condition there is a word $w = \sigma_1\dots\sigma_n \in \Sigma^n$ with $w \in \lang{\mathcal{A}}$ such that 
	$\delta(q_0,\sigma_1\dots\sigma_i\sigma_{i+l}\dots\sigma_n) \notin \{q_n,q_{n+1}\}$ holds for every $i \in \{0,\dots,n-2\}, l \in \{2,\dots,n-i\}$. We will show that $w\sigma$ is a primality witness of $\mathcal{A}$.
	
	Let $\mathcal{B} = (S,\Sigma,s_0,\eta,G) \in \alpha(\mathcal{A})$. It is easy to see that, since $\mathcal{A}$ is a safety DFA, we can assume w.l.o.g. that $\mathcal{B}$ is a safety DFA as well. See \cite{netser18decomposition} for more details. We will show $w\sigma \in \lang{\mathcal{B}}$.
	
	If $\lang{\mathcal{B}} = \Sigma^*$ then $w\sigma \in \lang{\mathcal{B}}$ trivially holds. Therefore, we assume $\lang{\mathcal{B}} \subset \Sigma^*$. Since $\mathcal{B}$ is a minimal safety DFA and has therefore only one rejecting state, which is a rejecting sink, the DFA $\mathcal{B}$ does not enter this sink in the run on $w$. This implies that in the run on $w$ the DFA $\mathcal{B}$ can pass only through $\size{\mathcal{B}}-1 \leq (n+1)-1 = n$ different states. Therefore, there are $i,j \in \{0,\dots,n\}$ with $i<j$ such that $\eta(s_0,\sigma_1\dots\sigma_i) = \eta(s_0,\sigma_1\dots\sigma_j)$.
	\begin{description}
		\item[Case 1: \normalfont{$j=n$.}] Then we have:
		\begin{align*}
			&\eta(s_0,w\sigma) \\
			=	&\eta(\eta(s_0,w),\sigma)\\
			=	&\eta(\eta(s_0,\sigma_1\dots\sigma_n),\sigma)\\
			=	&\eta(\eta(s_0,\sigma_1\dots\sigma_i),\sigma)\\
			=	&\eta(s_0,\sigma_1\dots\sigma_i\sigma).
		\end{align*}
		Since $w \in \lang{\mathcal{A}}$, we further have $\delta(q_0,\sigma_1\dots\sigma_i) = q_i$ for $i < j = n$. Additionally, we have per requirement $\sigma \notin \bigcup_{j=0}^{n-1} \Sigma_{j,n+1}$ and therefore in particular $\sigma \notin \Sigma_{i,n+1}$. Thus, we have $\delta(q_0,\sigma_1\dots\sigma_i\sigma) = \delta(q_i,\sigma) \neq q_{n+1}$ and therefore $\sigma_1\dots\sigma_i\sigma \in \lang{\mathcal{A}}$. Since $\eta(s_0,w\sigma) = \eta(s_0,\sigma_1\dots\sigma_i\sigma)$, this implies $w\sigma \in \lang{\mathcal{B}}$. We are done with Case 1.
		
		\item[Case 2: \normalfont{$j<n$.}] Then we have:
		\begin{align*}
			&\eta(s_0,w\sigma) \\
			=	&\eta(\eta(s_0,w),\sigma)\\
			=	&\eta(\eta(s_0,\sigma_1\dots\sigma_j),\sigma_{j+1}\dots\sigma_n\sigma)\\
			=	&\eta(\eta(s_0,\sigma_1\dots\sigma_i),\sigma_{j+1}\dots\sigma_n\sigma)\\
			=	&\eta(s_0,\sigma_1\dots\sigma_i\sigma_{j+1}\dots\sigma_n\sigma).
		\end{align*}
		Additionally, we have $0 \leq i < j < n$ and therefore $i \leq n-2$. Further, with $i < j < n$ we have $i+1 < j+1 < n+1$ and therefore $i+2 \leq j+1 \leq n$ which implies $2 \leq (j+1)-i \leq n-i$. Select $l = (j+1)-i$.
		Now we have $i \in \{0,\dots,n-2\}, l \in \{2,\dots,n-i\}$ with $\eta(s_0,w) = \eta(s_0,\sigma_1\dots\sigma_i\sigma_{i+l}\dots\sigma_n)$. 
		
		With the second condition $\delta(q_0,\sigma_1\dots\sigma_i\sigma_{i+l}\dots\sigma_n) = q_k$ holds for a $k \in \{0,\dots,n-1\}$. With condition one we then have $\sigma \notin \Sigma_{k,n+1}$ and therefore $\delta(q_0,\sigma_1\dots\sigma_i\sigma_{i+l}\dots\sigma_n\sigma) = \delta(q_k,\sigma) \neq q_{n+1}$. Thus, we have $\sigma_1\dots\sigma_i\sigma_{i+l}\dots\sigma_n\sigma \in \lang{\mathcal{A}}$. Since $\eta(s_0,w\sigma) = \eta(s_0,\sigma_1\dots\sigma_i\sigma_{i+l}\dots\sigma)$ holds, this implies $w\sigma \in \lang{\mathcal{B}}$. We are done with Case 2.
	\end{description}
	With Cases 1 and 2 we have shown $w\sigma \in \lang{\mathcal{B}}$. Therefore, every DFA in $\alpha(\mathcal{A})$ accepts $w\sigma$. This means that $w\sigma$ is a primality witness of $\mathcal{A}$, which implies the primality of $\mathcal{A}$. The proof of \cref{lem:fl_strongNecessaryCondition} is complete.
\end{proof}

We have established a sufficient condition for the primality of $\mathcal{A}$. Now we proof that the second condition, the ADFA $\mathcal{A}$ not having the CEP, implies the first condition. More precisely, we prove that already a weakened form of the second condition, the ADFA $\mathcal{A}$ not having a stronger property than the CEP, already implies the first condition.
\begin{lemma}
	\label{lem:fl_conditionImplication}
	The following assertion holds:
	\begin{align*}
		&\neg (\forall w = \sigma_1 \dots \sigma_n \in L. \exists i \in \{0,\dots, n-2\}. \exists l \in \{2,\dots,n-i\}. \\
		&\hspace{2cm} \delta(q_0,\sigma_1 \dots \sigma_i \sigma_{i+l} \dots \sigma_n) = q_{n+1})\\
		\Rightarrow
		&\neg (\Sigma_{n-1,n} \subseteq \bigcup_{j=0}^{n-1} \Sigma_{j,n+1})
	\end{align*}
	\lipicsEnd
\end{lemma}
\begin{proof}
	We proof the contraposition. Therefore we assume that $\Sigma_{n-1,n} \subseteq \bigcup_{j=0}^{n-1} \Sigma_{j,n+1}$.
	Let $w = \sigma_1\dots\sigma_n \in \Sigma^n$ with $w \in \lang{\mathcal{A}}$. We need to show that there are $i \in \{0,\dots,n-2\}, l \in \{2,\dots,n-i\}$ such that $\delta(q_0,\sigma_1\dots\sigma_i\sigma_{i+l}\dots\sigma_n) = q_{n+1}$.
	
	Clearly, we have $\sigma_j \in \Sigma_{j-1,j}$ for each $j \in \{1,\dots,n\}$. In particular, we have $\sigma_n \in \Sigma_{n-1,n}$. Since $\Sigma_{n-1,n} \subseteq \bigcup_{j=0}^{n-1} \Sigma_{j,n+1}$, there then exists an $i \in \{0,\dots,n-1\}$ with $\sigma_n \in \Sigma_{i,n+1}$. It clearly holds that $i \neq n-1$ and therefore we have $i \in \{0,\dots,n-2\}$.
	
	Now we select $l = n-i$. Note that this implies $l = n-i \in \{2,\dots,n-i\}$. Then we have:
	\begin{align*}
		&\delta(q_0,\sigma_1\dots\sigma_i\sigma_{i+l}\dots\sigma_n)\\
		=	&\delta(q_0,\sigma_1\dots\sigma_i\sigma_n)\\
		=	&\delta(q_i,\sigma_n)\\
		=	&q_{n+1}.
	\end{align*}
	
	In conclusion, we have proven the existence of values $i \in \{0,\dots,n-2\}, l \in \{2,\dots,n-i\}$ such that $\delta(q_0,\sigma_1\dots\sigma_i\sigma_{i+l}\dots\sigma_n) = q_{n+1}$. Thus, we have proven the contraposition of the implication and therefore the implication itself. The proof of \cref{lem:fl_conditionImplication} is complete.
\end{proof}

With \cref{lem:fl_strongNecessaryCondition,lem:fl_conditionImplication} in hand, the proof of \cref{cla:fl_characterization} (\ref{cla_ass:fl_characterization_linear+safety+non-CEP}) is trivial:
\begin{proof}[Proof of \cref{cla:fl_characterization} (\ref{cla_ass:fl_characterization_linear+safety+non-CEP})]
	Assume that the safety DFA $\mathcal{A}$ does not have the CEP. With \cref{lem:fl_conditionImplication} this implies $\neg (\Sigma_{n-1,n} \subseteq \bigcup_{j=0}^{n-1} \Sigma_{j,n+1})$. Therefore, both conditions of \cref{lem:fl_strongNecessaryCondition} are satisfied and $\mathcal{A}$ is prime. We are done.
\end{proof}

We have proven \cref{cla:fl_characterization} (\ref{cla_ass:fl_characterization_linear+safety+non-CEP}). That is, we have proven the primality of $\mathcal{A}$ if $\mathcal{A}$ is a safety DFA and does not have the CEP. 

\subsection{Concluding remarks}
Our goal for \cref{sec:fl_proofs} was to prove \cref{the:fl_characterization}, thereby completely characterizing the compositionality of ADFAs and thus of finite languages. To do so, we set out to prove \cref{cla:fl_characterization} (\ref{cla_ass:fl_characterization_non-linear})-(\ref{cla_ass:fl_characterization_linear+safety+non-CEP}), which taken together imply \cref{the:fl_characterization}.

In \cref{subsec:fl_charac_aAndb,subsec:fl_charac_c,subsec:fl_charac_d,subsec:fl_charac_e} we have proven (\ref{cla_ass:fl_characterization_non-linear})-(\ref{cla_ass:fl_characterization_linear+safety+non-CEP}) one after the other. Note that while (\ref{cla_ass:fl_characterization_non-linear}) and (\ref{cla_ass:fl_characterization_linear+sigmaN}), which cover the cases of non-linear ADFAs and linear ADFAs with a $\sigma^n \in \lang{\mathcal{A}}$, were fairly simple to prove, the remaining (\ref{cla_ass:fl_characterization_non-sigmaN+non-safety})-(\ref{cla_ass:fl_characterization_linear+safety+non-CEP}) covering linear ADFAs with $\sigma^n \notin \lang{\mathcal{A}}$ for all $\sigma \in \Sigma$ required a lot more work. The difficulty arose from extensions of words $w \in \lang{\mathcal{A}}, |w|=n$.

First, we have seen that if such an ADFA is not a safety DFA then it is composite. This holds because one of the additional rejecting states can be used to construct DFAs rejecting the mentioned extensions. These DFAs do not need a rejecting sink and instead circle back from their last state $q_n$ to an earlier state after reading an appropriate prefix of the extension.

Second, we have seen that if the ADFA has the CEP then it is composite, regardless of its accepting and rejecting states. This holds because using the CEP we can construct DFAs rejecting the extensions. They essentially omit one state of the original ADFA and can thus employ a rejecting sink.

Finally, we have seen that if the ADFA is a safety DFA and does not have the CEP then the ADFA is prime. This holds because a safety DFA can be decomposed into safety DFAs. Therefore, the DFAs used in the decomposition have to employ a rejecting sink and are therefore, intuitively speaking, one state short to read words of length $n$. Thus, they are necessarily confused about at least two prefixes of a word of length $n$. With the CEP not holding, we have shown that this implies primality.

This concludes the proof of \cref{the:fl_characterization} and thereby the characterization of the compositionality of ADFAs and thus of finite languages. This also completes the proofs for \cref{sec:fl_characterization}.

\section{Proofs for \texorpdfstring{\cref{sec:fl_complexity}}{Section \ref{sec:fl_complexity}}}
We use the characterization of the compositionality of ADFAs to prove:
\theFlPrimeDFAFinComplexity*

We will start by showing that $\primeDFAfin{}$ is in \complexityClassFont{NL}, before proving the \complexityClassFont{NL}-hardness.
To prove that $\primeDFAfin{}$ is in \complexityClassFont{NL}, we argue:
\begin{lemma}
	\label{lem:fl_PrimeDFAFinInNL}
	\cref{alg:fl_primeDFAFinNLalgorithm} is an \complexityClassFont{NL}-algorithm for $\primeDFAfin{}$.\lipicsEnd
\end{lemma}
\begin{proof}
	We begin by arguing that \cref{alg:fl_primeDFAFinNLalgorithm} indeed decides $\primeDFAfin{}$. Afterwards, we argue that \cref{alg:fl_primeDFAFinNLalgorithm} can be implemented in logarithmic space.
	
	First, note that it can obviously be decided in \complexityClassFont{NL} whether a state $q$ is reachable from a state $p$ in a given DFA. This further implies that it can be decided in \complexityClassFont{NL} whether a state $q$ is reachable in a given DFA and whether a given DFA recognizes a non-empty language. Second, it can obviously be decided in \complexityClassFont{NL} whether $\lang{\mathcal{A}^{q}} \neq \lang{\mathcal{A}^{p}}$ for a given DFA $\mathcal{A}$ and two states $p,q$.
	
	Now note that with the well-known Immerman-Szelepcsényi theorem $\complexityClassFont{NL} = \complexityClassFont{co-NL}$ holds \cite{DBLP:journals/siamcomp/Immerman88}. Therefore, it can also be decided in \complexityClassFont{NL} wether a state $q$ is unreachable from a state $p$ in a given DFA, whether a state $q$ is unreachable in a given DFA, and whether a given DFA recognizes the empty language. Further, it can be decided in \complexityClassFont{NL} whether $\lang{\mathcal{A}^{q}} = \lang{\mathcal{A}^{p}}$ for a given DFA $\mathcal{A}$ and two states $p,q$.
	
	With these observations in hand, we argue that \cref{alg:fl_primeDFAFinNLalgorithm} decides $\primeDFAfin{}$.
	
	Let $\mathcal{A} = (Q,\Sigma,q_0,\delta,F)$ with $Q = \{q_0,\dots,q_m\}$ be a DFA recognizing a finite language $L$. With $\hat{\mathcal{A}} = (\hat{Q},\Sigma,q_0,\hat{\delta},\hat{F})$ we denote the minimal DFA of $\mathcal{A}$. Note that $\hat{\mathcal{A}}$ is an ADFA. With $\Sigma_{i,j}$ we denote the usual subsets of $\Sigma$ in the ADFA $\hat{\mathcal{A}}$.
	
	We begin by making a couple of observations about the behavior of the algorithm.
	
	First, note that the algorithm accepts in line 1 if $L = \emptyset$. Otherwise, it resumes.
	
	Second, we consider the values $c$ and $n$ calculated in lines 2-18. 
	We argue that, if the algorithm is not to reject in line 18, then after line 16 the variable $c$ has to store the number of states of the given DFA $\mathcal{A}$ which can be removed because they are unreachable or can be merged with a state with a smaller subscript. 
	Thus, we argue that $c$ is the number of the removable states of $\mathcal{A}$, meaning $\ind{\mathcal{A}} = |\mathcal{A}| - c$. This implies $n = (m+1)-c-2 = |\mathcal{A}|-c-2 = \ind{\mathcal{A}}-2$.
	We begin by inspecting lines 17-18. 
	In line 17 the value $n$ is calculated depending on the value $c$. In line 18 the algorithm rejects if no word $w \in \Sigma^n$ with $w \in L$ exists.
	Note that the length of the longest word in $L$ is $\ind{\mathcal{A}}-2$ if $\hat{\mathcal{A}}$ is linear and is strictly smaller than $\ind{\mathcal{A}}-2$ otherwise. Therefore, to avoid being forced to reject in line 18 the algorithm has to achieve $n \leq \ind{\mathcal{A}}-2$. Since in line 17 the algorithm defines $n = (m+1)-c-2 = \size{\mathcal{A}}-c-2$, to avoid rejection in line 18 it is necessary that $\size{\mathcal{A}}-c \leq \ind{\mathcal{A}}$. Thus, $c$ has to be at least the number of removable states of $\mathcal{A}$ to avoid rejection in line 18.
	Consider the calculation of $c$ in lines 2-16.
	Note that $c$ is incremented in line 5 only if the current state $q_i$ is unreachable and in line 10 only if $q_i$ is reachable and a reachable state $q_j$ with $j < i$ is found that is equivalent to $q_i$. 
	Also, note that $c$ is incremented at most once for each $q_i$. 
	Therefore, $c$ is smaller or equal the number of removable states of $\mathcal{A}$, meaning $\size{\mathcal{A}} - c \geq \ind{\mathcal{A}}$, with equality being achieved only if the algorithm increments $c$ for each removable $q_i$.
	Thus, to avoid rejection in line 18 the algorithm has to increment $c$ for each removable $q_i$, so that $c$ is exactly the number of removable states of $\mathcal{A}$. This then implies $n = (m+1)-c-2 = \ind{\mathcal{A}} - 2$.
	
	Third, note that if the algorithm reaches line 18, that is, if $L \neq \emptyset$, then it rejects in line 18 if there exists no word $w \in \Sigma^n$ with $w \in L$. Otherwise, it resumes. 
	Since we just argued that $n = \ind{\mathcal{A}}-2$, this means that the algorithm rejects if $\hat{\mathcal{A}}$ is not linear. Otherwise, it resumes. 
	
	Fourth, note that if the algorithm reaches line 19, that is, if $L \neq \emptyset$ and $\hat{\mathcal{A}}$ is linear, then it accepts if there is a word $\sigma^n \in L$. Otherwise, it resumes.
	
	Fifth, note that if the algorithm reaches line 20 then in lines 20-22 it checks whether $\mathcal{A}$ is a safety DFA. It rejects if $\mathcal{A}$ is not a safety DFA. Otherwise, it resumes. Note here that in order to check whether $\mathcal{A}$ is a safety DFA it is sufficient to ensure that for each reachable state $q_i$ it holds that $q_i \notin F \Rightarrow \lang{\mathcal{A}^{q_i}} = \emptyset$, which implies that each reachable rejecting state can be replaced by a rejecting sink. Further, note that deciding whether $q_i$ is unreachable and whether $\lang{\mathcal{A}^{q_i}} \neq \emptyset$ can both be done in \complexityClassFont{NL}.
	
	We briefly summarize our observations so far. The algorithm terminates before reaching line 23 iff:
	\begin{itemize}
		\item $L = \emptyset$, in which case it accepts, or
		\item $L \neq \emptyset$ and $\hat{\mathcal{A}}$ is not linear, in which case it rejects, or
		\item $L \neq \emptyset$ and $\hat{\mathcal{A}}$ is linear and there exists a $\sigma \in \Sigma$ with $\sigma^n \in \Sigma$, in which case it accepts, or
		\item $L \neq \emptyset$ and $\hat{\mathcal{A}}$ is linear and there exists no $\sigma \in \Sigma$ with $\sigma^n \in \Sigma$ and $\mathcal{A}$ is not a safety DFA, in which case it rejects.
	\end{itemize}
	Therefore, line 23 is reached iff the following holds: $L \neq \emptyset$ and $\hat{\mathcal{A}}$ is linear and there exists no $\sigma \in \Sigma$ with $\sigma^n \in \Sigma$ and $\mathcal{A}$ is a safety DFA. 
	In this case the DFA $\mathcal{A}$ is prime iff the minimal DFA $\hat{\mathcal{A}}$ does not have the CEP.
	We argue that this is checked in lines 23-28. 
	To be more precise, we argue that the algorithm rejects in line 26 iff $\hat{\mathcal{A}}$ has the CEP.
	The somewhat strange fashion in which the condition is checked, with two separate selections of words in lines 24 and 26, is motivated by the need to achieve an \complexityClassFont{NL}-algorithm. We will inspect this later.
	
	Assume that $\hat{\mathcal{A}}$ does not have the CEP. We show that this implies that the algorithm does not reject.
	Note that a word $w = \sigma_1\dots\sigma_n$ can be selected in line 24, which witnesses that $\hat{\mathcal{A}}$ does not have the CEP. 
	This selection can be done anew for each $x$. 
	Then the same word can be chosen as $w'$ in line 26. Again, this selection can be done anew for each pair $i,l$. With this selection the state $\delta(q_0,\sigma_1\dots\sigma_i\sigma_{i+l}\dots\sigma_n)$ is neither a rejecting state equivalent to a rejecting sink nor an accepting state from which only rejecting states are reachable. Therefore, a word $v \in \Sigma^+$ with $\delta(q_0,\sigma_1\dots\sigma_i\sigma_{i+l}\dots\sigma_nv) \in F$ can be selected in line 26. That is, $\sigma_1\dots\sigma_i\sigma_{i+l}\dots\sigma_nv \in L$. Thus, the algorithm does not reject for any combination of values $x$ and $i,l$.
	
	Now assume that the algorithm does not reject in line 26. We show that this implies that $\hat{\mathcal{A}}$ does not have the CEP.
	For each $x \in \{1,\dots,n\}$ a word $w_x = \sigma_{1,x}\dots\sigma_{n,x}$ can be chosen in line 24 so that the algorithm does not reject. Define $w = \sigma_{1,1}\dots\sigma_{n,n}$. We argue that this $w$ breaches the CEP.
	
	Since $w_x \in L$ holds for each $x$, we have $\sigma_{x,x} \in \Sigma_{x-1,x}$ for each $x$. This implies $w \in L$.
	
	Now let $i \in \{0,\dots,n-2\}, l \in \{2,\dots,n-i\}$. Let $x = i+l$. Since the algorithm does not reject in line 26, we have $\sigma_{x,x} \in \Sigma_{i,j}$ for a $j < i+l$. Expressing this more formally, we have: $\forall x \in \{1,\dots,n\}, i \in \{0,\dots,n-2\}, l \in \{2,\dots,n-i\}. x = i+l \Rightarrow \exists j \in \{1,\dots,n+1\}. i< j < i+l \wedge \sigma_{x,x} \in \Sigma_{i,j}$. Getting rid of the variable $x$, this is clearly equivalent to: $\forall i \in \{0,\dots,n-2\}, l \in \{2,\dots,n-i\}. \exists j \in \{1,\dots,n+1\}. i < j < i+l \wedge \sigma_{i+l,i+l} \in \Sigma_{i,j}$.
	
	Note that with \cref{lem:fl_i+lCondition} this implies that there are no $i \in \{0,\dots,n-2\}, l \in \{2,\dots,n-i\}$ such that $\hat{\delta}(q_0,\sigma_{1,1}\dots\sigma_{i,i}\sigma_{i+l,i+l}\dots\sigma_{n,n}) \in \{q_n,q_{n+1}\}$. Therefore, the word $w$ breaches the CEP.
	
	Thus we have shown that the algorithm rejects in line 26 iff $\hat{\mathcal{A}}$ has the CEP.
	
	So far, we have made observations about the behavior of the algorithm. Now we argue how the correctness of the algorithm arises from our observations.
	
	Let $\mathcal{A}$ be prime. Then with \cref{the:fl_characterization} we have:
	\begin{bracketenumerate}
		\item $L = \emptyset$, or
		\item $L \neq \emptyset$ and $\hat{\mathcal{A}}$ is linear and there is a $\sigma \in \Sigma$ with $\sigma^n \in L$, or
		\item $L \neq \emptyset$ and $\hat{\mathcal{A}}$ is linear and there is no $\sigma \in \Sigma$ with $\sigma^n \in L$ and $\mathcal{A}$ is a safety DFA and $\hat{\mathcal{A}}$ does not have the CEP.
	\end{bracketenumerate}
	If (1) holds then the algorithm accepts in line 1. If (2) holds then the algorithm does not reject in line 18 and accepts in line 19. If (3) holds then the algorithm does not reject in lines 18, 21 or 26 and accepts in line 29.
	
	Now let $\mathcal{A}$ be composite. Then with \cref{the:fl_characterization} we have:
	\begin{bracketenumerate}
		\item $L \neq \emptyset$ and $\hat{\mathcal{A}}$ is not linear, or
		\item $L \neq \emptyset$ and $\hat{\mathcal{A}}$ is linear and there is no $\sigma \in \Sigma$ with $\sigma^n \in L$ and $\mathcal{A}$ is not a safety DFA, or
		\item $L \neq \emptyset$ and $\hat{\mathcal{A}}$ is linear and there is no $\sigma \in \Sigma$ with $\sigma^n \in L$ and $\mathcal{A}$ is a safety DFA and $\hat{\mathcal{A}}$ has the CEP.
	\end{bracketenumerate}
	If (1) holds then the algorithm does not accept in line 1 and rejects in line 18. If (2) holds then the algorithm does not accept in lines 1 or 19 and rejects in line 21. If (3) holds then the algorithm does not accept in lines 1 or 19 and rejects in line 26.
	
	Therefore, the algorithm indeed decides $\primeDFAfin{}$.
	
	Now we have to argue that the algorithm can be implemented in logarithmic space.
	
	We have already argued that the conditions in lines 1-17 can be decided in \complexityClassFont{NL}. In line 18 the word $w$ does not have to be stored completely. Instead, the algorithm can nondeterministically select one letter after the other, holding only one letter, a counter and the current state in memory. Therefore, line 18 only needs logarithmic space. Since in line 19 the algorithm only needs to nondeterministically select a letter and store it and can then proceed analogously to line 18, that is, holding a counter and the current state in memory, line 19 only needs logarithmic memory as well.
	Again, we have already argued that the conditions in lines 20-22 can be decided in \complexityClassFont{NL}.
	
	This leaves us with lines 23-29.
	The algorithm can store the value $x$. It can then nondeterministically select a word $w \in \Sigma^n$ and check $w \in L$ analogously to line 18. While doing so, it can store the letter $\sigma_x$.
	The algorithm can then store the values $i,l$. It can nondeterministically select a word $w' \in \Sigma^n$ with $\sigma_{i+l}' = \sigma_x$ and check $w' \in L$ analogously to line 18. While doing so, it can store the state reached after reading the prefix $\sigma_1'\dots\sigma_i'$. It can then start a second simulation of a run, beginning in state $q_i'$, once the suffix $\sigma_{i+l}'\dots\sigma_n'$ is reached. Since all this can be done with a constant number of counters, the algorithm only needs logarithmic space here as well.
	
	Thus, the algorithm only needs logarithmic space.
	
	In conclusion, \cref{alg:fl_primeDFAFinNLalgorithm} nondeterministically decides $\primeDFAfin{}$ in logarithmic space. Therefore, \cref{alg:fl_primeDFAFinNLalgorithm} is an \complexityClassFont{NL}-algorithm for $\primeDFAfin{}$. We are done.
\end{proof}

We have proven that $\primeDFAfin{}$ is in \complexityClassFont{NL}. 
Next, we prove that $\primeDFAfin{}$ is \complexityClassFont{NL}-hard. In fact, we prove that $\primeDFAfin{2}$ is \complexityClassFont{NL}-hard, where $\primeDFAfin{2}$ denotes the restriction of $\primeDFAfin{}$ to DFAs with at most two letters. Formally, we prove:
\begin{lemma}
	\label{lem:fl_PrimeDFAFinNLHard}
	The problem $\primeDFAfin{2}$ is \complexityClassFont{NL}-hard.\lipicsEnd
\end{lemma}
\begin{proof}
	We introduce a number of problems, which we will use in the \complexityClassFont{NL}-hardness proof. We do this locally, since we will not use these problems anywhere else.
	
	In \cref{sec:preliminaries} we have introduced the problem \problemFont{STCON}, which is \complexityClassFont{NL}-complete \cite{DBLP:books/daglib/0095988}. We now introduce a restriction of \problemFont{STCON}. With \problemFont{STCONDAG} we denote the restriction of STCON to directed acyclic graphs. With \problemFont{2STCONDAG} we denote the restriction of \problemFont{STCONDAG} to graphs with a maximum outdegree of two.
	
	With \emptyDFA{} we denote the emptiness problem for DFAs, that is, the problem of deciding emptiness for the language recognized by a given DFA. It is known that \emptyDFA{} is \complexityClassFont{NL}-complete \cite{DBLP:journals/jcss/Jones75}. With $\emptyDFAfin{}$ we denote the restriction of \emptyDFA{} to DFAs recognizing finite languages. With $\emptyDFAfin{2}$ we denote the restriction of $\emptyDFAfin{}$ to DFAs with at most two letters.
	
	First, we will argue that \problemFont{STCONDAG} is \complexityClassFont{NL}-complete. Clearly, this implies the \complexityClassFont{NL}-completeness of \problemFont{2STCONDAG}.
	Second, we will argue that $\emptyDFAfin{2}$ is \complexityClassFont{NL}-complete as well by \complexityClassFont{L}-reducing \problemFont{2STCONDAG} to $\emptyDFAfin{2}$. 
	Third and finally, we will argue that $\primeDFAfin{2}$ is \complexityClassFont{NL}-hard by \complexityClassFont{L}-reducing $\emptyDFAfin{2}$ to $\primeDFAfin{2}$.
	
	We begin by considering \problemFont{STCONDAG}. Since \problemFont{STCON} is \complexityClassFont{NL}-complete, the restriction \problemFont{STCONDAG} is in \complexityClassFont{NL} as well. We only have to show \complexityClassFont{NL}-hardness. 
	We will provide a sketch of how the \complexityClassFont{NL}-hardness proof of \problemFont{STCON} can be adapted for \problemFont{STCONDAG}.
	
	The general idea of the \complexityClassFont{NL}-hardness proof of \problemFont{STCON} is to turn the \complexityClassFont{NL}-Turing maschine of the given problem in \complexityClassFont{NL} into a graph. The configurations of the Turing maschine translate to the nodes of the graph. The connections between configurations translate to the edges of the graph. 
	
	Note that we can adapt any given \complexityClassFont{NL}-Turing maschine by introducing a configuration counter, which simply counts the number of calculation steps of the original \complexityClassFont{NL}-Turing maschine. Since the original \complexityClassFont{NL}-Turing maschine can only go through polynomially many configurations before terminating, this counter can be implemented using logarithmic space. The adapted Turing maschine therefore is an \complexityClassFont{NL}-Turing maschine as well. Thus, we can use this adapted \complexityClassFont{NL}-Turing maschine and translate it into a graph.
	
	Now note that turning this adapted \complexityClassFont{NL}-Turing maschine into a graph clearly results in a directed acyclic graph. This immediately implies the \complexityClassFont{NL}-hardness of \problemFont{STCONDAG}.
	
	We have argued that \problemFont{STCONDAG} is in \complexityClassFont{NL} and is \complexityClassFont{NL}-hard. Thus, it is \complexityClassFont{NL}-complete.
	
	Since \problemFont{STCONDAG} is in \complexityClassFont{NL}, the restriction \problemFont{2STCONDAG} is in \complexityClassFont{NL} as well. Additionally, it is easy to \complexityClassFont{L}-reduce \problemFont{STCONDAG} to \problemFont{2STCONDAG}, which implies the \complexityClassFont{NL}-hardness of \problemFont{2STCONDAG}. Therefore, the restriction \problemFont{2STCONDAG} is \complexityClassFont{NL}-complete.
	
	Next, we will argue that $\emptyDFAfin{2}$ is \complexityClassFont{NL}-complete. Note that, since \emptyDFA{} is in \complexityClassFont{NL}, the restriction $\emptyDFAfin{2}$ is in \complexityClassFont{NL} as well. We only have to show \complexityClassFont{NL}-hardness. To do this, we \complexityClassFont{L}-reduce \problemFont{2STCONDAG} to $\emptyDFAfin{2}$ in practically the same manner as \problemFont{STCON} is \complexityClassFont{L}-reduced to \emptyDFA{}.
	
	We can turn any given directed acyclic graph into a DFA using the usual construction, translating nodes to states and edges to transitions. The starting node translates to the initial state of the DFA. The target node translates to the only accepting state. Additionally, we can introduce a rejecting sink and add transitions into this sink for any nodes without sufficiently many edges. 
	
	Note that this construction results in an ADFA with exactly one accepting state.
	Therefore, the constructed ADFA recognizes a finite language and it recognizes a non-empty language iff the target node is reachable from the starting node in the graph.
	Additionally, the number of letters of the constructed DFA is equal to the maximum outdegree in the underlying graph.
	Therefore, this construction witnesses the \complexityClassFont{NL}-hardness of $\emptyDFAfin{2}$.
	
	We have argued that $\emptyDFAfin{2}$ is in \complexityClassFont{NL} and is \complexityClassFont{NL}-hard. Thus, it is \complexityClassFont{NL}-complete.
	
	After having established the \complexityClassFont{NL}-completeness of $\emptyDFAfin{2}$, we can turn to the \complexityClassFont{NL}-hardness of $\primeDFAfin{2}$. We will \complexityClassFont{L}-reduce $\emptyDFAfin{2}$ to $\primeDFAfin{2}$, building on the idea used in \cite{DBLP:journals/iandc/KupfermanM15} to prove the \complexityClassFont{NL}-hardness of \problemFont{Prime-DFA}, and additionally employing our characterization of the compositionality of finite languages.
	
	Let $\mathcal{A} = (Q,\Sigma,q_0,\delta,F)$ be an input DFA for $\emptyDFAfin{2}$. We construct a DFA $\mathcal{A}' = (Q',\Sigma',q_0,\delta',F')$. We introduce four new states: $Q' = Q \cup \{p_0,p_1,p_2,p_-\}$. We set $\Sigma' = \Sigma$ if $|\Sigma| = 2$. Otherwise, we select a $\Sigma'$ with $\Sigma \subseteq \Sigma'$ and $|\Sigma'| = 2$. W.l.o.g. we assume $\Sigma' = \{a,b\}$. We define: $F' = F \cup \{p_2\}$. Finally, we define $\delta'$ for each $q \in Q', \sigma \in \Sigma$ as follows:
	\begin{align*}
		\delta'(q,\sigma) = \begin{cases}
			p_-						&\text{ if $q \in \{p_2,p_-\}$}\\
			p_2						&\text{ if $q = p_1$ and $\sigma = b$}\\
			p_-						&\text{ if $q = p_1$ and $\sigma = a$}\\
			p_1						&\text{ if $q = p_0$ and $\sigma = a$}\\
			p_-						&\text{ if $q = p_0$ and $\sigma = b$}\\
			p_0						&\text{ if $q \in F$}\\
			\delta(q,\sigma)		&\text{ if $q \notin F$ and $\sigma \in \Sigma$}\\
			p_-						&\text{ else, thus if $q \notin F$ and $\sigma \notin \Sigma$}
		\end{cases}.
	\end{align*}
	Note that the last case of this definition is only relevant if $\Sigma \subset \Sigma'$.
	
	The main idea of the construction is to plug the series of states $p_0,p_1,p_2,p_-$ behind every accepting state of $\mathcal{A}$. That is, for every original accepting state of $\mathcal{A}$ every transition leads to $p_0$. These transitions are the only transitions into $p_0$. Reading $ab$ when in $p_0$ the DFA $\mathcal{A}'$ advances over $p_1$ to $p_2$. The other transitions exiting $p_0$ and $p_1$ lead into the rejecting sink $p_-$. From $p_2$ every transition leads into the rejecting sink $p_-$. The other transitions of $\mathcal{A}$ are replicated in $\mathcal{A}'$.
	
	Obviously, the DFA $\mathcal{A}'$ recognizes a finite language as well. Indeed, it is easy to see that $\lang{\mathcal{A}'} = \{w, w \sigma a b \setDel \sigma \in \Sigma' \wedge w \in X\}$, where $X$ is the set of words recognized by $\mathcal{A}$ which have no real prefix which is recognized by $\mathcal{A}$ as well.
	
	Now we prove that $\mathcal{A}'$ is prime iff $\mathcal{A}$ recognizes the empty language.
	
	First, assume that $\mathcal{A}$ recognizes the empty language. Then $\lang{\mathcal{A}'} = \{w, w \sigma a b \setDel \sigma \in \Sigma' \wedge w \in X\} = \emptyset$ clearly holds. Thus, $\mathcal{A}'$ recognizes the empty language as well. Therefore, $\mathcal{A}'$ is prime.
	
	Second, assume that $\mathcal{A}$ does not recognize the empty language. Clearly, this implies that $\mathcal{A}'$ does not recognize the empty language either. Further, every longest word in $\lang{\mathcal{A}'}$ ends on $ab$. Therefore, no longest word in $\lang{\mathcal{A}'}$ consists of the repetition of the same letter. And finally, $\mathcal{A}'$ is not a safety DFA, since the rejecting state $p_0$ - and, for good measure, also the rejecting state $p_1$ - is reachable in $\mathcal{A}'$, from which the accepting state $p_2$ is reachable. With \cref{the:fl_characterization} this implies the compositionality of $\mathcal{A}'$.
	
	We have shown that $\mathcal{A}'$ is prime iff $\mathcal{A}$ recognizes the empty language. Since $\mathcal{A}'$ can clearly be constructed out of $\mathcal{A}$ in logarithmic space, we have established an \complexityClassFont{L}-reduction from the \complexityClassFont{NL}-complete problem $\emptyDFAfin{2}$ to $\primeDFAfin{2}$. Thus, we have shown the \complexityClassFont{NL}-hardness of $\problemFont{2Prime-DFA}_\text{fin}$. The proof of \cref{lem:fl_PrimeDFAFinNLHard} is complete.
\end{proof}

We have argued that $\primeDFAfin{2}$ is \complexityClassFont{NL}-hard. Clearly, this implies the \complexityClassFont{NL}-hardness of $\primeDFAfin{}$ as well.

With \cref{lem:fl_PrimeDFAFinInNL,lem:fl_PrimeDFAFinNLHard} we have shown that $\primeDFAfin{}$ as well as $\primeDFAfin{2}$ are in \complexityClassFont{NL} and are \complexityClassFont{NL}-hard. This immediately implies the \complexityClassFont{NL}-completeness of $\primeDFAfin{}$ and $\primeDFAfin{2}$. This proves \cref{the:fl_primeDFAFinComplexity}. We are done with the proofs for \cref{sec:fl_complexity}.

\section{Proofs for \texorpdfstring{\cref{sec:fl_differentNotionsOfCompositionality}}{Section \ref{sec:fl_differentNotionsOfCompositionality}}}
We finish the proofs for our results concerning finite languages by proving the theorems in \cref{sec:fl_differentNotionsOfCompositionality}, in which finite languages are analyzed under different notions of compositionality.

We begin by proving:
\theFlCupDNFCharacterization*
\begin{proof}[Proof of \cref{the:fl_cupDNFCharacterization}]
	Consider a minimal ADFA $\mathcal{A} = (Q,\Sigma,q_I,\delta,F)$ recognizing a non-empty language. Let $n \in \natNum$ be the length of the longest word in $\lang{\mathcal{A}}$.
	
	We begin by proving (i). 
	
	First, assume that $\mathcal{A}$ is not linear. This implies $\ind{\mathcal{A}} > n+2$. We reuse the DFA $\mathcal{A}_w$ introduced in the proof of \cref{cla:fl_characterization} (\ref{cla_ass:fl_characterization_non-linear}) and pictured in \cref{subfig:fl_A_w}, which is the minimal DFA recognizing the language $\{w\}$ for a word $w \in \Sigma^*$. Clearly, we have $\lang{\mathcal{A}} = \bigcup_{w \in \lang{\mathcal{A}}} \lang{\mathcal{A}_w}$. And since $\ind{\mathcal{A}} > n+2 \geq \size{\mathcal{A}_w}$ for each $w \in \lang{\mathcal{A}}$, this already implies the $\cup$-compositionality of $\mathcal{A}$.
	
	Second, assume that $\mathcal{A}$ is linear. This implies $\ind{\mathcal{A}} = n+2$. 
	Let $w \in \Sigma^*$ be a word with $w \in \lang{\mathcal{A}},|w| = n$.
	Let $\mathcal{B} = (S,\Sigma,s_I,\eta,G)$ be a minimal DFA with $\lang{\mathcal{B}} \subseteq \lang{\mathcal{A}}$ and $\size{\mathcal{B}} < \ind{\mathcal{A}}$.
	We prove $w \notin \lang{\mathcal{B}}$, which implies the $\cup$-primality of $\mathcal{A}$.
	
	We begin by showing that $\mathcal{B}$ has to possess a rejecting sink. Let $u \in \Sigma^*$ with $|u| > n$. Then $uu' \notin \lang{\mathcal{A}}$ holds for each $u' \in \Sigma^*$. Since $\lang{\mathcal{B}} \subseteq \lang{\mathcal{A}}$, this implies $\eta(s_0,uu') = \eta(\eta(s_0,u),u') \notin G$ for each $u' \in \Sigma^*$. Since $\mathcal{B}$ is minimal, this implies that $\eta(s_0,u)$ is a rejecting sink.
	
	With this result in hand, we prove $w \notin \lang{\mathcal{B}}$ by contradiction. Assume $w \in \lang{\mathcal{B}}$. Note that $\mathcal{B}$ has at most $n+1$ states and that one of these states is a rejecting sink. Since $w \in \lang{\mathcal{B}}$, this rejecting sink is not visited by $\mathcal{B}$ in its initial run on $w$. Therefore, only $n$ different states are visited in the initial run on $w$, which implies that one state is visited twice. Clearly, this implies that $\mathcal{B}$ recognizes an infinite language. This contradicts $\lang{\mathcal{B}} \subseteq \lang{\mathcal{A}}$. Our proof by contradiction of $w \notin \lang{\mathcal{B}}$ is done.
	
	We have shown that there exists no DFA with strictly less than $\ind{\mathcal{A}}$ states that recognizes a subset of $\lang{\mathcal{A}}$ and which accepts $w$. Therefore, $\mathcal{A}$ is $\cup$-prime.
	
	We have proven that $\mathcal{A}$ is $\cup$-prime iff $\mathcal{A}$ is linear. We are done with (i).
	
	Next, we consider (ii).
	
	First, we prove the DNF-compositionality of $\mathcal{A}$ if $\mathcal{A}$ is not linear or if there exists no $\sigma \in \Sigma$ with $\sigma^n \in \lang{\mathcal{A}}$.
	
	If $\mathcal{A}$ is not linear then, as we have shown in (i), $\mathcal{A}$ is $\cup$-composite. This implies that $\mathcal{A}$ is DNF-composite. Therefore, we only have to consider the case where $\mathcal{A}$ is linear and where no $\sigma \in \Sigma$ with $\sigma^n \in \lang{\mathcal{A}}$ exists.
	
	Let $w \in \lang{\mathcal{A}}$. 
	
	If $|w|<n$ we have $\size{\mathcal{A}_w} = |w|+2 < n+2 = \ind{\mathcal{A}}$. Therefore, for each such word we can simply use the DFA $\mathcal{A}_w$ in the $\cup$-decomposition. 
	
	If $|w|=n$ then we can utilize the idea outlined in \cite[Example 3.2]{DBLP:journals/iandc/KupfermanM15} to build two DFAs $\mathcal{A}_w^1,\mathcal{A}_w^2$ with $\lang{\mathcal{A}_w^1} \cap \lang{\mathcal{A}_w^2} = \{w\}$ and $\size{\mathcal{A}_w^1},\size{\mathcal{A}_w^2} < n+2$, since per requirement there are at least two different letters in $w$.
	
	Let $\mathcal{A}_w^1$ be the minimal DFA with $\lang{\mathcal{A}_w^1} = \{w\}^*$. Let $\sigma \in \Sigma$ such that $|w|_{\sigma} > 0$. Let $\mathcal{A}_w^2$ be the minimal DFA with $\lang{\mathcal{A}_w^2} = \{u \in \Sigma^* \setDel |u|_{\sigma} = |w|_{\sigma}\}$. Obviously, we have $\size{\mathcal{A}_w^1} = |w|+1 = n+1 < n+2 = \ind{\mathcal{A}}$ and $\size{\mathcal{A}_w^2} = |w|_{\sigma}+2 < n+2 = \ind{\mathcal{A}}$. It is also clear that $\lang{\mathcal{A}_w^1} \cap \lang{\mathcal{A}_w^2} = \{w\}$.
	
	Now we can prove DNF-compositionality of $\mathcal{A}$. Define $L' = \{w \in L \setDel |w|<n\}$. Then $\lang{\mathcal{A}} = \bigcup_{w \in L'} \lang{\mathcal{A}_w} \cup \bigcup_{w \in \lang{\mathcal{A}} \setminus L'}(\lang{\mathcal{A}_w^1} \cap \lang{\mathcal{A}_w^2})$ obviously holds.
	This completes the DNF-compositionality proof of $\mathcal{A}$ if $\mathcal{A}$ is linear and there exists no $\sigma \in \Sigma$ with $\sigma^n \in \lang{\mathcal{A}}$.
	
	In summary, we have shown that $\mathcal{A}$ is DNF-composite if $\mathcal{A}$ is not linear or if there exists no $\sigma \in \Sigma$ with $\sigma^n \in \lang{\mathcal{A}}$.
	
	Second, we show that $\mathcal{A}$ is DNF-prime if $\mathcal{A}$ is linear and there exists a $\sigma \in \Sigma$ with $\sigma^n \in \lang{\mathcal{A}}$.
	
	Assume that $\mathcal{A}$ is as required. Let $\sigma \in \Sigma$ with $\sigma^n \in \lang{\mathcal{A}}$. Let $s \in \natNumGeq{1}$ and $t_1,\dots,t_s \in \natNumGeq{1}$ such that there exist DFAs $\mathcal{A}_{1,1},\dots,\mathcal{A}_{1,t_1},\dots,\mathcal{A}_{s,1},\dots,\mathcal{A}_{s,t_s}$ with $\lang{\mathcal{A}} = \bigcup_{i=1}^{s}\bigcap_{j=1}^{t_i}\lang{\mathcal{A}_{i,j}}$. Then there exists an $i \in \{1,\dots,s\}$ with $\sigma^n \in \bigcap_{j=1}^{t_i}\lang{\mathcal{A}_{i,j}}$. Now note that with \cref{lem:fl_sigmaN} there is a $j \in \{1,\dots,t_i\}$ with $|\mathcal{A}_{i,j}| \geq n+2$, since otherwise we would have $\sigma^{n+(n+1)!} \in \bigcap_{j=1}^{t_i}\lang{\mathcal{A}_{i,j}}$, which would be a contradiction to $\lang{\mathcal{A}} = \bigcup_{i=1}^{s}\bigcap_{j=1}^{t_i}\lang{\mathcal{A}_{i,j}}$. Therefore, $\mathcal{A}$ is not $(\ind{\mathcal{A}}-1)$-DNF-decomposable. Thus, $\mathcal{A}$ is DNF-prime.
	
	We have shown that $\mathcal{A}$ is DNF-prime if $\mathcal{A}$ is linear and there exists a $\sigma \in \Sigma$ with $\sigma^n \in \lang{\mathcal{A}}$.
	
	In conclusion, we have shown that $\mathcal{A}$ is DNF-prime iff $\mathcal{A}$ is linear and there exists a $\sigma \in \Sigma$ with $\sigma^n \in \lang{\mathcal{A}}$. We are done with (ii).
	
	The proof of \cref{the:fl_cupDNFCharacterization} is complete.
\end{proof}

Finally, we prove the last remaining result of \cref{sec:fl_differentNotionsOfCompositionality}.
\theFlCapCupVsDNF*
\begin{proof}[Proof of \cref{the:fl_capCupVsDNF}]
	We use \cref{the:fl_characterization,the:fl_cupDNFCharacterization} to construct a DFA recognizing a finite language that is DNF-composite but $\cap$- and $\cup$-prime. 
	
	\begin{figure}[t]
		\centering
		\begin{tikzpicture}[node distance=2.5cm]
		\node[state, initial, accepting] 					(q0) 	{$q_0$};
		\node[state, right of=q0, accepting] 				(q1) 	{$q_1$};
		\node[state, right of=q1, accepting] 				(q2) 	{$q_2$};
		\node[state, right of=q2, accepting] 				(q3) 	{$q_3$};
		\node[state, right of=q3]			 				(q4) 	{$q_4$};
		
		\draw	(q0)	edge[above]								node{$a_1,a_2$}		(q1);
		\draw	(q0)	edge[bend left=45, above]				node{$a_3$}			(q2);
		
		\draw	(q1)	edge[above]								node{$a_2,a_3$}	(q2);
		\draw	(q1)	edge[bend right, below]					node{$a_1$}		(q4);
		
		\draw	(q2)	edge[above]								node{$a_3$}		(q3);
		\draw	(q2)	edge[bend left=45, above]				node{$a_1,a_2$}	(q4);
		
		\draw	(q3)	edge[above]								node{$\Sigma$}	(q4);
		
		\draw	(q4)	edge[loop right]						node{$\Sigma$}	(q4);
		\end{tikzpicture}
		\caption{DFA recognizing a finite language that is DNF-composite, but $\cap$- and $\cup$-prime.}
		\label{fig:fl_DNFCompositeCapPrimeCupPrime}
	\end{figure}
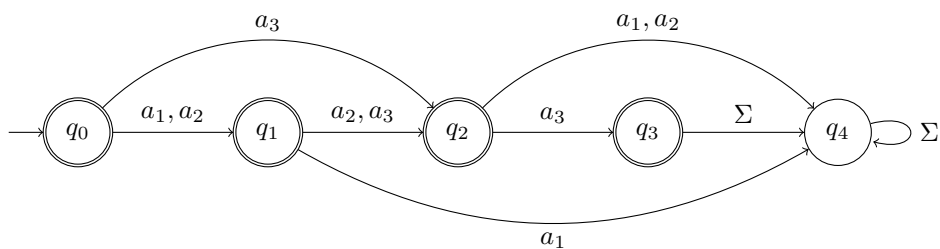
	Consider the ADFA $\mathcal{A} = (Q,\Sigma,q_0,\delta,F)$ pictured in \cref{fig:fl_DNFCompositeCapPrimeCupPrime}. Clearly, $\mathcal{A}$ is minimal and recognizes a non-empty language. Additionally, $\mathcal{A}$ is linear, there exists no $\sigma \in \Sigma$ with $\sigma^3 \in \lang{\mathcal{A}}$, and $\mathcal{A}$ is a safety DFA. Finally, $\mathcal{A}$ does not have the CEP, which is witnessed by the word $a_1a_2a_3 \in \lang{\mathcal{A}}$, since $\delta(q_0,a_2a_3), \delta(q_0,a_1a_3), \delta(q_0,a_3) = q_2 \notin \{q_3,q_4\}$ holds.
	
	Using \cref{the:fl_characterization,the:fl_cupDNFCharacterization}, this immediately implies that $\mathcal{A}$ is DNF-composite, but $\cap$- and $\cup$-prime. The same holds for the language $\lang{\mathcal{A}}$. We are done.
\end{proof}

This ends our discussion of the proofs for \cref{sec:fl_differentNotionsOfCompositionality}.
\section{Proofs for \texorpdfstring{\cref{sec:2DFAMinimalAndSPrimeDFA}}{Section \ref{sec:2DFAMinimalAndSPrimeDFA}}}
\label{sec:2DFAMinimal_proofs}
We provide proofs for the \complexityClassFont{NL}-completeness of \minimalDFA{2} formalized in \cref{the:2DFAMinimalNLComplete} and for the complexity boundaries of \sPrimeDFA{} and \sPrimeDFA{k} and of \primeDFA{} and \primeDFA{k} formalized in \cref{the:SPrimeDFAComplexity,the:PrimeDFAComplexity}.

We begin by proving:
\theTwoDFAMinimalNLComplete*
\begin{proof}[Proof of \cref{the:2DFAMinimalNLComplete}]
	The \complexityClassFont{NL}-completeness of \minimalDFA{} and its restrictions \minimalDFA{k} for $k \in \natNumGeq{3}$ is folklore. This immediately implies that \minimalDFA{2} is in \complexityClassFont{NL}. Thus, we only have to prove the \complexityClassFont{NL}-hardness of \minimalDFA{2}. To do so, we will \complexityClassFont{L}-reduce the \complexityClassFont{NL}-complete problem \problemFont{2STCON} to \minimalDFA{2} using the construction outlined in \cref{fig:2STCONto2DFAMinimalReduction}.
	
	Let $(G,s,t)$ be an input for the problem \problemFont{2STCON}. That is, $G = (V,E)$ is a directed graph with a maximum outdegree of two, and $s,t \in V$ are nodes of the graph. We construct a DFA $\mathcal{A}$ with at most two letters that is minimal iff $t$ is reachable from $s$ in $G$.
	
	If $s = t$ then $t$ is trivially reachable from $s$. In this case, we can construct an arbitrary minimal DFA with at most two letters. For example, we can construct the minimal DFA for the empty language. This case is done.
	
	From here on, we assume $s \neq t$. Further, w.l.o.g. we assume $V = \{0,\dots,n-1\}$ and $s=0,t=n-1$.
	
	Let $\mathcal{A}' = (Q',\Sigma,0,\delta',F')$ be the DFA constructed out of $(G,s,t)$ in the usual manner. That is, the nodes are translated to states, so $Q' = V'$, and the edges are translated to transitions. Further, the starting node $s=0$ is translated to the initial state, and the target node $t=n-1$ is translated to the only accepting state, so $F' = \{n-1\}$. If any node has not enough outgoing edges, self-loops are added. We will base the construction of $\mathcal{A}$ on $\mathcal{A}'$. 
	
	From here on, we use $0,1$ as the two letters of the alphabet, so $\Sigma = \{0,1\}$.
	
	We are ready to start with the construction of $\mathcal{A} = (Q,\Sigma,p_0,\delta,F)$. 
	We introduce the new states $p_0,\dots,p_{n-1}$ and $q_0,\dots,q_{n-1}$ and $q_0'$. Further, for each $i \in Q' = V = \{0,\dots,n-1\}$ we introduce the states $i_0',i_1',i_0,i_1$.
	We refer to the states $i,i_0',i_1',i_0,i_1$ for $i \in V$ as $v$-states. We refer to the states $p_0,\dots,p_{n-1}$ as $p$-states. And we refer to the states $q_0,\dots,q_{n-1}$ and $q_0'$ as $q$-states. We define $Q_p = \{p_0,\dots,p_{n-1}\}, Q_q = \{q_0,\dots,q_{n-1}\} \cup \{q_0'\}, Q_v = \{i,i_0',i_1',i_0,i_1 \setDel i \in V\}$.
	We say that states $p_i,q_i,i,i_0',i_1',i_0,i_1$ for an $i \in V$ are on layer $i$.
	We set $p_0$ as the initial state. Further, we define $F = \{n-1\}$.
	We do not introduce any additional letters and use the alphabet $\Sigma = \{0,1\}$ for $\mathcal{A}$.
	
	Finally, we define the transition function $\delta$ in the following way:
	\begin{itemize}
		\item For every $i \in V$ and $\sigma \in \Sigma$ define:
		\begin{align*}
			\delta(p_i,\sigma) = \begin{cases}
				p_{i+1} &\text{ if $i < n-1$ and $\sigma = 0$}\\
				p_{n-1} &\text{ if $i = n-1$ and $\sigma = 0$}\\
				i &\text{ else, thus if $\sigma = 1$}
			\end{cases}.
		\end{align*}
	
		\item For every $q \in \{q_0,\dots,q_{n-1}\} \cup \{q_0'\}$ and $\sigma \in \Sigma$ define:
		\begin{align*}
			\delta(q,\sigma) = \begin{cases}
				q_{0}'	&\text{ if $q=q_0$ and $\sigma = 0$}\\
				0		&\text{ if $q=q_0$ and $\sigma = 1$}\\
				q_{1}	&\text{ if $q=q_0'$ and $\sigma = 0$}\\
				q_{i+1} &\text{ if $q=q_i$ with $i \in \{1,\dots,n-2\}$ and $\sigma = 0$}\\
				q_0		&\text{ if $q=q_{n-1}$ and $\sigma = 0$}\\
				q		&\text{ else, thus if $q \neq q_0$ and $\sigma = 1$}
			\end{cases}.
		\end{align*}
	
		\item For every $i \in V$ and $\sigma \in \Sigma$ define: $\delta(i,\sigma) = i_\sigma'$.
		
		\item For every $i \in V$ and $j,\sigma \in \Sigma$ define:
		\begin{align*}
			\delta(i_j',\sigma) &= \begin{cases}
				i_j'	&\text{ if $\sigma = 0$}\\
				i_j		&\text{ else, thus if $\sigma = 1$}
			\end{cases}\\
			\delta(i_j,\sigma) &= \begin{cases}
				\delta'(i,\sigma)	&\text{ if $\sigma = j$}\\
				q_i					&\text{ else, thus if $\sigma \neq j$}
			\end{cases}.
		\end{align*} 
	\end{itemize}
	\cref{fig:2STCONto2DFAMinimalReduction} displays the DFA $\mathcal{A}$.
	
	We need to show that $\mathcal{A}$ is minimal iff $n-1$ is reachable from $0$ in $G$.
	
	First, we consider the case that $n-1$ is reachable from $0$ in $G$. We show that in this case the DFA $\mathcal{A}$ is minimal.
	
	If $n-1$ is reachable from $0$ in $G$ there obviously exists a word $w = \sigma_1\dots\sigma_k \in \Sigma^k$ such that $\delta'(0,w) = n-1$, where $k \in \natNumGeq{1}$. Using $w$ we now construct a word $w_+$ with $\delta(0,w) = n-1$, proving that $n-1$ is reachable from $0$ in $\mathcal{A}$. To do this, for each $l \in \{1,\dots,k\}$ we need to add the word $1\sigma_l$ after each letter $\sigma_l$. This word bridges the states $i_\sigma,i_\sigma'$. Thus, we have $w_+ = \sigma_1(1\sigma_1)\dots\sigma_k(1\sigma_k)$. It can be easily verified that $\delta(0,w_+) = n-1$.
	
	The word $w_+$ thus witnesses that $n-1$ is reachable from $0$ in $\mathcal{A}$. From here on, let $w_+$ be a word of minimal length with $\delta(0,w_+)$. The minimal length of $w_+$ clearly implies that in the run of $\mathcal{A}$ on $w_+$ starting in $0$ there are only $v$-states.
	
	We start our minimality proof of $\mathcal{A}$ by arguing that each state of $\mathcal{A}$ is reachable. This is obvious, since from $p_0$ every state $p_i$ with $i \in V$ can be reached by reading $0^i$. From $p_i$, every other state located on layer $i$ is reachable. Therefore, every state of $\mathcal{A}$ is reachable.
	
	We continue our minimality proof of $\mathcal{A}$ by showing that $\mathcal{A}$ possesses no equivalent states. Let $a,b \in Q$ with $a \neq b$. We have to prove $\lang{\mathcal{A}^a} \neq \lang{\mathcal{A}^b}$. To do this, we use a lengthy case distinction. Note that in this case distinction we will often trace cases back to different cases. But we will only trace back cases to already handled cases, thereby avoiding any circular reasoning.
	\begin{description}
		\item[Case 1: \normalfont{$a,b \in Q_q$.}] By reading the letter $0$ the DFA $\mathcal{A}$ can circle through the $q$-states. Therefore, there exists a $k \in \{0,\dots,n\}$ such that $\delta(a,0^k) = q_0$ and $b' = \delta(b,0^k) \neq q_0$, where $b' \in Q_q$.
		Then we clearly have $\delta(a,0^k1w_+) = \delta(0,w_+) = n-1$ and $\delta(b,0^k1w_+) = \delta(b',w_+) \neq n-1$, since the $q$-states can only be left into state $0$ and $w_+$ is a word of minimal length with $\delta(0,w_+) = n-1$.
		
		The word $0^k1w_+$ witnesses the inequivalence of $a$ and $b$. We are done with Case 1.
		
		\item[Case 2: \normalfont{$a,b \in Q_p$.}] Let $i,j \in V, i \neq j$ with $a = p_i, b = p_j$. Then we have $\delta(p_i,1(011)) = q_i$ and $\delta(p_j,1(011)) = q_j$.
		
		Case 2 can be traced back to Case 1.
		
		\item[Case 3: \normalfont{$a \in Q_p \Leftrightarrow b \notin Q_p$.}] W.l.o.g. let $a \in Q_p$. Let $i \in V$ with $a = p_i$. Then $\delta(a,0^{(n-1)-i}) = p_{n-1}$ holds. Let $b' = \delta(b,0^{(n-1)-i})$. If $\delta(b',1) \neq n-1$ we are done with Case 1 with witness $0^{(n-1)-i}1$. Therefore, we assume $\delta(b',1) = n-1$.
		
		Since $b \notin Q_p$, we have $b' \notin Q_p$. Since $0 = s \neq t = n-1$, we additionally have $\delta(q_0,1) = 0 \neq n-1$. Therefore, we have $b' \neq q_0$ and thus $b' \notin Q_q$. Thus, we have $b' \in Q_v$. With $\delta(b',1) = n-1$ this implies $b' = j_1$ for a $j \in V$ with $\delta'(j,1) = n-1$. But then $\delta(j_1,01) = \delta(q_j,1) \in \{q_j,0\}$ holds, while at the same time we have $\delta(p_{n-1},01) = \delta(p_{n-1},1) = n-1$. Then we are done with Case 1 with witness $0^{(n-1)-i}01$.
		
		We have shown the inequivalence of $a$ and $b$. We are done with Case 3.
		
		\item[Case 4: \normalfont{$(a \in Q_q \wedge b \in Q_v) \vee (a \in Q_v \wedge b \in Q_q)$.}] W.l.o.g. let $a \in Q_q$.
		\begin{description}
			\item[Case 4.1: \normalfont{$b \notin \{i_1,i_1' \setDel i \in V\}$.}] Then there exists an $r \in \{0,1,2\}$ with $\delta(b,0^r) = j_0'$ for a $j \in V$. Let $a' = \delta(a,0^r)$. Note that $a' \in Q_q$.
			
			Then there exists an $s \in \{0,\dots,n\}$ such that $\delta(a',0^s) = q_0'$. Additionally, we have $\delta(j_0',0^s) = j_0'$. Then we have $\delta(q_0',1) = q_0'$ and $\delta(j_0',1) = q_j \neq q_0'$.
			
			In summary, we have $\delta(a,0^r0^s1) = q_0'$ and $\delta(b,0^r0^s1) = q_j$ for a $j \in V$.
			
			Case 4.1 can be traced back to Case 1.
			
			\item[Case 4.2: \normalfont{$b \in \{i_1 \setDel i \in V\} \wedge a = q_0$.}] Then we have $\delta(a,0) = \delta(q_0,0) = q_0'$ and $\delta(b,0) = \delta(i_1,0) = q_i$ for an $i \in V$.
			
			Case 4.2 can be traced back to Case 1.
			
			\item[Case 4.3: \normalfont{$b \in \{i_1 \setDel i \in V\} \wedge a \neq q_0$.}] Then we have $\delta(a,1) = a$ and $\delta(b,1) = j$ for a $j \in V$. 
			
			Case 4.3 can be traced back to Case 4.1.
			
			\item[Case 4.4: \normalfont{$b \in \{i_1' \setDel i \in V\} \wedge a \neq q_0$.}] Then we have $\delta(a,1) = a$ and $\delta(b,1) = i_1$ for an $i \in V$.
			
			Case 4.4 can be traced back to Case 4.3.
			
			\item[Case 4.5: \normalfont{$b \in \{i_1' \setDel i \in V\} \wedge a = q_0$.}] Then we have $\delta(a,0) = q_0'$ and $\delta(b,0) = b$.
			
			Case 4.5 can be traced back to Case 4.4.
		\end{description}
		
		\item[Case 5: \normalfont{$a,b \in Q_v$.}] Let $i,j \in V$ with $a \in \{i,i_0',i_0,i_1',i_1\}$ and $b \in \{j,j_0',j_0,j_1',j_1\}$.
		
		For any $x \in Q$ we define the $q$-distance $\qdistOp{x}$ of $x$ as the length of the shortest word $w \in \Sigma^*$ such that $\delta(x,w) \in Q_q$. For a $k \in V$ the state $k$ then has a $q$-distance of three, states $k_0'$ and $k_1'$ have a $q$-distance of two, and states $k_0$ and $k_1$ have a $q$-distance of one.
		\begin{description}
			\item[Case 5.1: \normalfont{$\qdistOp{a} \neq \qdistOp{b}$.}] W.l.o.g. let $\qdistOp{a} < \qdistOp{b}$. Then there exists a word $w \in \Sigma^*$ such that $\delta(a,w) \in Q_q$ and $\delta(b,w) \in Q_v$.
			
			Case 5.1 can be traced back to Case 4.
			
			\item[Case 5.2: \normalfont{$\qdistOp{a} = \qdistOp{b} \wedge a \in \{i_0',i_0\} \wedge b \in \{j_1',j_1\}$.}] Let $r = \qdistOp{a}$. Clearly, we have $r \in \{1,2\}$. Additionally, we have $\delta(a,1^{r-1}) = i_0, \delta(a,1^r) = q_i$ and $\delta(b,1^{r-1}) = j_1, \delta(b,1^r) = \delta'(j,1) \in V \subseteq Q_v$.
			
			Case 5.2 can be traced back to Case 4.
			
			\item[Case 5.3: \normalfont{$\qdistOp{a} = \qdistOp{b} \wedge a \in \{i_1',i_1\} \wedge b \in \{j_0',j_0\}$.}] Case 5.3 is analogous to Case 5.2.
			
			\item[Case 5.4: \normalfont{$\qdistOp{a} = \qdistOp{b} \wedge a \in \{i_0',i_0\} \wedge b \in \{j_0',j_0\}$.}] Clearly, this implies $i \neq j$. Let $r = \qdistOp{a}$. Then $\delta(a,1^r) = q_i, \delta(b,1^r) = q_j$ holds.
			
			Case 5.4 can be traced back to Case 1.
			
			\item[Case 5.5: \normalfont{$\qdistOp{a} = \qdistOp{b} \wedge a \in \{i_1',i_1\} \wedge b \in \{j_1',j_1\}$.}] Case 5.5 is analogous to Case 5.4.
			
			\item[Case 5.6: \normalfont{$\qdistOp{a} = \qdistOp{b} \wedge a = i \wedge b = j$.}] Clearly, this implies $i \neq j$. Then $\delta(a,011) = q_i, \delta(b,011) = q_j$ holds.
			
			Case 5.6 can be traced back to Case 1.
		\end{description}
	\end{description}
	With Cases 1-5 the states $a$ and $b$ are not equivalent.
	
	In conclusion, every state of $\mathcal{A}$ is reachable and $\mathcal{A}$ does not possess two states that are not identical but equivalent. Therefore, the DFA $\mathcal{A}$ is minimal.
	
	Thus, we have shown that $\mathcal{A}$ is minimal if $n-1$ is reachable from $0$ in $G$.
	
	Second, we consider the case that $n-1$ is not reachable from $0$ in $G$. We show that in this case the DFA $\mathcal{A}$ is not minimal.
	
	If $n-1$ is not reachable from $0$ in $G$ then clearly the state $n-1$ is not reachable from the state $0$ in $\mathcal{A}$, since the $p$-states cannot be reentered once they were left and the $q$-states can only be left to state $0$. Additionally, this means that $n-1$ is unreachable from every $q$-state. Since $n-1$ is the only accepting state of $\mathcal{A}$, this implies that at least the state $0$ and the $q$-states can be replaced by a rejecting sink. Since there are a strictly positive number of $q$-states, the DFA constructed in this manner is strictly smaller than $\mathcal{A}$ but recognizes the same language as $\mathcal{A}$. Therefore, $\mathcal{A}$ is not minimal.
	
	We have shown that $\mathcal{A}$ is not minimal if $n-1$ is unreachable from $0$ in $G$.
	
	In conclusion, we have shown that $\mathcal{A}$ is minimal iff $n-1$ is reachable from $0$ in $G$. Since $\mathcal{A}$ can clearly be constructed in logarithmic space, we have specified an \complexityClassFont{L}-reduction of \problemFont{2STCON} to \minimalDFA{2}. Thus, the problem \minimalDFA{2} is \complexityClassFont{NL}-hard.
	
	Since \minimalDFA{2} is in \complexityClassFont{NL} and is \complexityClassFont{NL}-hard, it is \complexityClassFont{NL}-complete. We are done.
\end{proof}

Next, we prove:
\theSPrimeDFAComplexity*
\begin{proof}[Proof of \cref{the:SPrimeDFAComplexity}]
	Before we turn to the proof, we introduce some notation. Let $\mid \subseteq \mathbb{Z} \times \mathbb{Z}$ be the usual divisibility relation. That is, for $a,b \in \mathbb{Z}$ it holds that $a \mid b$ iff there exists a $k \in \mathbb{Z}$ with $ka = b$.
	
	We begin by arguing that \sPrimeDFA{} is in \complexityClassFont{ExpSpace}. This follows directly from the proof of \cite[Theorem 2.4]{DBLP:journals/iandc/KupfermanM15}, which states that \primeDFA{} is in \complexityClassFont{ExpSpace}. Note that we can adapt this proof with the minor modification that we now need to consider every DFA with less states than the given DFA instead of every DFA with less states than the index of the given DFA. This modification is necessary, since the notion of S-primality uses the size instead of the index of the given DFA. The remainder of the proof can remain unaltered. Thus, it follows trivially from \cite[Theorem 2.4]{DBLP:journals/iandc/KupfermanM15} that \sPrimeDFA{} is in \complexityClassFont{ExpSpace}. This implies that the restrictions \sPrimeDFA{k} for $k \in \natNumGeq{2}$ are in \complexityClassFont{ExpSpace} as well.
	
	Now we consider the lower complexity boundary. We will show that \sPrimeDFA{2} is \complexityClassFont{NL}-hard. This immediately implies the \complexityClassFont{NL}-hardness of \sPrimeDFA{} and \sPrimeDFA{k} for $k \in \natNumGeq{2}$.
	
	To establish the \complexityClassFont{NL}-hardness of \sPrimeDFA{2}, we will \complexityClassFont{L}-reduce the \complexityClassFont{NL}-complete problem \problemFont{2STCON} to \sPrimeDFA{2}. To do this, we will adapt the construction used in \cref{the:2DFAMinimalNLComplete}.
	
	Let $(G,s,t)$ with $G = (V,E)$ be an input for \problemFont{2STCON}. We construct a DFA $\hat{\mathcal{A}} = (\hat{Q},\Sigma,p_0,\hat{\delta},\hat{F})$ with $\Sigma = \{0,1\}$ that is S-prime iff $t$ is reachable from $s$ in $G$.
	
	If $s = t$ then $t$ is trivially reachable from $s$ in $G$ and we can construct an arbitrary S-prime DFA with $\Sigma = \{0,1\}$, for example the minimal DFA recognizing the empty language. This case is done.
	
	From here on, we assume $s \neq t$. Analogous to the proof of \cref{the:2DFAMinimalNLComplete} we further assume w.l.o.g. $V = \{0,\dots,n-1\}$ and $s=0, t=n-1$. Additionally, let $\mathcal{A} = (Q,\Sigma,p_0,\delta,F)$ be the DFA constructed out of $(G,s,t)$ in the proof of \cref{the:2DFAMinimalNLComplete}, which is displayed in \cref{fig:2STCONto2DFAMinimalReduction}. We construct $\hat{\mathcal{A}}$ by modifying $\mathcal{A}$.
	
	We expand the set of states. Let $\myUnderbar{Q} = \{\myUnderbar{x} \setDel x \in Q \setminus \{p_0\}\}$. Then we define $\hat{Q} = Q \cup \myUnderbar{Q} \cup \{z_+\}$. Further, set $p_0$ as the initial state and define $\hat{F} = \{z_+\}$. We keep $\Sigma$ as the alphabet. Finally, for every $x \in \hat{Q}$ and $\sigma \in \Sigma$ we define:
	\begin{align*}
		\hat{\delta}(x,\sigma) = \begin{cases}
			z_+			&\text{ if $x = z_+$}\\
			z_+			&\text{ if $x = \myUnderbar{n-1}$ and $\sigma = 1$}\\
			p_0			&\text{ if $x \in \myUnderbar{Q} \setminus \{\myUnderbar{n-1}\}$ and $\sigma = 1$}\\
			y			&\text{ if $x \in \myUnderbar{Q}$ and $\sigma = 0$, where $y \in Q$ with $\myUnderbar{y} = x$}\\
			\myUnderbar{\delta(x,\sigma)} &\text{ else, thus if $x \in Q$}
		\end{cases}.
	\end{align*}

	From here on, we use the following notation: For a word $w = \sigma_1\dots\sigma_n \in \Sigma^n$ with $n \in \natNumGeq{2}$, it is $f(w) = \sigma_10\sigma_20\dots\sigma_{n-1}0\sigma_n$. For $\sigma \in \Sigma$, it is $f(\sigma) = \sigma$. For the empty word $\varepsilon$, it is $f(\varepsilon) = \varepsilon$.
	
	We make a couple of observations about $\hat{\mathcal{A}}$.
	
	First, note that every state of $\hat{\mathcal{A}}$ is reachable. 
	We have $\hat{\delta}(p_0,f(0^i)) = \myUnderbar{p_i}$ for every $i \in V$. From $\myUnderbar{p_i}$ every other state of layer $i$ is reachable. Since in this way the state $\myUnderbar{n-1}$ is reachable, the state $z_+$ is reachable as well. Therefore, every state of $\hat{\mathcal{A}}$ is reachable.
	
	Second, note that $\hat{\mathcal{A}}$ is a co-safety DFA, since $z_+$ is the only accepting state and $z_+$ is a sink.
	
	Third, note that every two states in $\hat{Q} \setminus \{z_+\}$ are reachable from one another. We have $\hat{\delta}(x,1) = p_0$ for every $x \in \myUnderbar{Q} \setminus \{\myUnderbar{n-1}\}$. We further have $\hat{\delta}(\myUnderbar{n-1},001) = p_0$. And finally, we have $\hat{\delta}(x,0) = \myUnderbar{\delta(x,0)}$ for every $x \in Q$. Therefore, the state $p_0$ is reachable from every state in $\hat{Q} \setminus \{z_+\}$. Since, as just argued, every state of $\hat{\mathcal{A}}$ is reachable, this immediately implies that every two states in $\hat{Q} \setminus \{z_+\}$ are reachable from one another.
	
	The second and third point imply that, in the terminology of \cite{DBLP:journals/iandc/KupfermanM15}, the DFA $\hat{\mathcal{A}}$ is a \defHighlight{simple} co-safety DFA. That is, a co-safety DFA that consists of the accepting sink and a second component in which every two states are reachable from one another. With \cite[Theorem 5.5]{DBLP:journals/iandc/KupfermanM15} this implies that $\hat{\mathcal{A}}$ is prime. Note that here we refer to the original notion of primality, not S-primality.
	
	Here, it is important to note that a simple co-safety DFA is S-prime iff it is minimal. This is easy to see. If a DFA, be it a simple co-safety DFA or not, is not minimal, then it is not S-prime. Therefore, a non-minimal simple co-safety DFA is not S-prime.
	Further, if a DFA is minimal, then it is S-prime iff it is prime. Since every simple co-safety DFA is prime, this implies that every minimal simple co-safety DFA is S-prime.
	Therefore, a simple co-safety DFA is S-prime iff it is minimal.
	
	Finally, we point out that until the accepting sink is reached, the DFA $\hat{\mathcal{A}}$ alternates between states in $Q$ and $\myUnderbar{Q}$. More precisely, we have:
	\begin{align*}
		\forall x \in Q. \forall \sigma \in \Sigma. \hat{\delta}(x,\sigma) = \myUnderbar{\delta(x,\sigma)},
	\end{align*}
	and:
	\begin{align*}
		\forall \myUnderbar{x} \in \myUnderbar{Q}. \hat{\delta}(\myUnderbar{x},0) = x \wedge (\myUnderbar{x} \neq \myUnderbar{n-1} \Rightarrow \hat{\delta}(\myUnderbar{x},1) = p_0) \wedge (\myUnderbar{x} = \myUnderbar{n-1} \Rightarrow \hat{\delta}(\myUnderbar{x},1) = z_+).  
	\end{align*}

	This immediately implies:
	\begin{align*}
		\forall w \in \Sigma^+.\forall x \in Q. \hat{\delta}(x,f(w)) = \myUnderbar{\delta(x,w)} \wedge \hat{\delta}(x,f(w)0) = \delta(x,w),
	\end{align*}
	and:
	\begin{align*}
		\forall w \in \Sigma^+. \forall \myUnderbar{x} \in \myUnderbar{Q}. \hat{\delta}(\myUnderbar{x}, 0f(w)) = \myUnderbar{\delta(x,w)} \wedge \hat{\delta}(\myUnderbar{x}, 0f(w)0) = \delta(x,w).
	\end{align*}
	Thus, by adding the letter $0$ after every letter of a word the DFA $\hat{\mathcal{A}}$ can simulate the behavior of the DFA $\mathcal{A}$.
	
	Additionally, this implies:
	\begin{align*}
		&\forall n \in \natNumGeq{1}. \forall w = \sigma_1\dots\sigma_n \in \Sigma^n. \forall x \in Q.\\
		&\hspace{2cm} \hat{\delta}(x,\sigma_1\dots\sigma_{n-1}) \neq z_+ \\
		&\hspace{3.5cm} \Rightarrow\\
		&\hspace{2cm} ( \hspace{0.3cm}
		(\hat{\delta}(x,\sigma_1\dots\sigma_n) \in Q \Leftrightarrow (2 \mid n) \wedge (\sigma_n = 0 \vee \hat{\delta}(x,\sigma_1\dots\sigma_{n-1}) \neq \myUnderbar{n-1})) \\
		&\hspace{2cm} \wedge (\hat{\delta}(x,\sigma_1\dots\sigma_n) \in \myUnderbar{Q} \Leftrightarrow \neg(2 \mid n)) \\
		&\hspace{2cm} \wedge (\hat{\delta}(x,\sigma_1\dots\sigma_n) = z_+ \Leftrightarrow (2 \mid n) \wedge \sigma_n = 1 \wedge \hat{\delta}(x,\sigma_1\dots\sigma_{n-1}) = \myUnderbar{n-1})
		).
	\end{align*}
	Thus, if the run of $\hat{\mathcal{A}}$ on a word $w$ beginning in a state in $Q$ does not end in $z_+$ then it ends in a state in $Q$ if $|w|$ is even. Otherwise, that is, if $|w|$ is odd, such a run ends in a state in $\myUnderbar{Q}$. Further and in particular, if the run of $\hat{\mathcal{A}}$ on a word $w$ beginning in a state in $Q$ ends in $z_+$ then $z_+$ is entered for the first time after reading an even number of letters.
	
	Similar observations can be made for runs beginning in a state in $\myUnderbar{Q}$. In particular, if the run of $\hat{\mathcal{A}}$ on a word $w$ beginning in a state in $\myUnderbar{Q}$ ends in $z_+$ then $z_+$ is entered for the first time after reading an odd number of letters.
	
	Now we turn to the actual proof. Again, we have to show that $\hat{\mathcal{A}}$ is S-prime iff $t$ is reachable from $s$ in $G$. We use our observation that $\hat{\mathcal{A}}$ is S-prime iff it is minimal. 
	
	We begin by assuming that $t$ is reachable from $s$ in $G$. We show that $\hat{\mathcal{A}}$ is minimal, which implies its S-primality. We already argued that every state of $\hat{\mathcal{A}}$ is reachable. Therefore, we only have to prove that $\hat{\mathcal{A}}$ does not possess two states different from each other that are equivalent.
	
	Let $a,b \in \hat{Q}$ with $a \neq b$. If one of them is the accepting sink $z_+$ then the two states are trivially inequivalent. Therefore, we assume $a,b \neq z_+$.
	\begin{description}
		\item[Case 1: \normalfont{$a,b \in Q$.}] Since the state $n-1$ needs to be handled separately, we use a second case distinction.
		\begin{description}
			\item[Case 1.1: \normalfont{$a \neq n-1 \wedge b \neq n-1$.}] Since $\mathcal{A}$ is minimal, the states $a$ and $b$, which are in $Q$ and are therefore states of $\mathcal{A}$ as well, are inequivalent in $\mathcal{A}$. W.l.o.g. let $w \in \Sigma^*$ with $\delta(a,w) = n-1$ and $\delta(b,w) \neq n-1$. Since $a \neq n-1$, we have $w \neq \varepsilon$. Then with the above observation we have $\hat{\delta}(a,f(w)) = \myUnderbar{\delta(a,w)} = \myUnderbar{n-1}$ and $\hat{\delta}(b,f(w)) = \myUnderbar{\delta(b,w)} \neq  \myUnderbar{n-1}$.
			
			Thus, the word $f(w)1$ witnesses the inequivalence of $a$ and $b$. We are done with Case 1.1.
			
			\item[Case 1.2: \normalfont{$a = n-1 \vee b = n-1$.}] W.l.o.g. let $a = n-1$. Then $b \in Q$ with $b \neq n-1$. We have $\delta(n-1,1) = (n-1)_1'$. Let $b' = \delta(b,1)$. Obviously, we have $b' \neq (n-1)_1'$. 
			
			If $b' \neq n-1$ then with $\hat{\delta}(n-1,f(1)0) = (n-1)_1'$ and $\hat{\delta}(b,f(1)0) = b' \notin \{n-1,(n-1)_1'\}$ the Case 1.2 can be traced back to Case 1.1.
			
			If $b' = n-1$ then we have $\delta(n-1,11) = \delta((n-1)_1',1) = (n-1)_1$ and $\delta(b,11) = \delta(n-1,1) = (n-1)_1'$ and thus $\hat{\delta}(n-1,f(11)0) = (n-1)_1$ and $\hat{\delta}(b,f(11)0) = (n-1)_1'$. Then Case 1.2 can be traced back to Case 1.1.
			
			In summary, Case 1.2 can be traced back to Case 1.1.
		\end{description}
	
		\item[Case 2: \normalfont{$a,b \in \myUnderbar{Q}$.}] Then there exist $c,d \in Q$ with $\myUnderbar{c} = a$ and $\myUnderbar{d} = b$. Since $a \neq b$, we have $c \neq d$. Additionally, we have $\hat{\delta}(a,0) = c$ and $\hat{\delta}(b,0) = d$.
		
		Case 2 can be traced back to Case 1.
		
		\item[Case 3: \normalfont{$(a \in Q \wedge b \in \myUnderbar{Q}) \vee (a \in \myUnderbar{Q} \wedge b \in Q)$.}] W.l.o.g. let $a \in Q$ and $b \in \myUnderbar{Q}$. Note that the accepting sink $z_+$ is reachable from both $a$ and $b$, since, as explained above, the initial state $p_0$ is reachable from both $a$ and $b$ and the accepting sink is reachable from $p_0$. Let $w_a \in \Sigma^*$ be a word of minimal length with $\hat{\delta}(a,w_a) = z_+$. Let $w_b \in \Sigma^*$ be a word of minimal length with $\hat{\delta}(b,w_b) = z_+$. Note that, as explained above as well, the word $w_a$ is of even length, while the word $w_b$ is of odd length. This implies $|w_a| \neq |w_b|$ and therefore $|w_a| < |w_b|$ or $|w_a| > |w_b|$.
		
		The shorter one of the words $w_a,w_b$ witnesses the inequivalence of $a$ and $b$. We are done with Case 3.
	\end{description}
	With Cases 1-3 the two states $a$ and $b$ are not equivalent. Therefore, every state of $\hat{\mathcal{A}}$ is reachable and $\hat{\mathcal{A}}$ does not possess two states that are not identical but equivalent. Therefore, the DFA $\hat{\mathcal{A}}$ is minimal.
	
	As explained above, since $\hat{\mathcal{A}}$ is a simple co-safety DFA, the minimality of $\hat{\mathcal{A}}$ implies its S-primality. Therefore, we have shown that $\hat{\mathcal{A}}$ is S-prime if $t$ is reachable from $s$ in $G$.
	
	Next, we assume that $t$ is not reachable from $s$ in $G$. We have to show that $\hat{\mathcal{A}}$ is not S-prime. To do this, we show that $\hat{\mathcal{A}}$ is not minimal. Since every non-minimal DFA is trivially S-composite, showing this is sufficient to prove that $\hat{\mathcal{A}}$ is not S-prime.
	
	To prove that $\hat{\mathcal{A}}$ is not minimal, we prove that the states $0,q_0 \in Q \subseteq \hat{Q}$ are equivalent. That is, we prove: $\forall w \in \Sigma^*. \hat{\delta}(0,w) = z_+ \Leftrightarrow \hat{\delta}(q_0,w) = z_+$.
	
	Let $w \in \Sigma^*$.
	
	First, assume that $\hat{\delta}(0,w) = z_+$. We have to prove $\hat{\delta}(q_0,w) = z_+$. Since $0 \neq z_+$, we have $w \neq \varepsilon$. Additionally, since $t=n-1$ is unreachable from $s=0$ in $G$, the state $n-1$ is unreachable from the state $0$ in $\mathcal{A}$. 
	
	We prove that this implies the existence of $u,v \in \Sigma^*$ where $\neg(2 \mid |u|)$ and $w = u1v$, that is, that there is a $1$ in $w$ at an even position. Assume by contradiction that there is a $0$ in $w$ at every even position. Then there exists a $w' \in \Sigma^*$ with $f(w') = w$ if $|w|$ is odd and $f(w')0 = w$ if $|w|$ is even. 
	
	For $f(w') = w$, we arrive at a contradiction with: $z_+ = \hat{\delta}(0,w0) = \hat{\delta}(0,f(w')0) = \delta(0,w') \neq z_+$.
	For $f(w')0 = w$, we arrive at a contradiction with: $z_+ = \hat{\delta}(0,w) = \hat{\delta}(0,f(w')0) = \delta(0,w') \neq z_+$.
	
	Thus, we have shown by contradiction the existence of $u,v \in \Sigma^*$ where $\neg(2 \mid |u|)$ and $w = u1v$. Let $u$ be the shortest prefix of $w$ for which this holds. Then we clearly have $\hat{\delta}(0,u) \in \myUnderbar{Q}\setminus\{\myUnderbar{n-1}\}$ and $\hat{\delta}(q_0,u) \in \myUnderbar{Q}\setminus\{\myUnderbar{n-1}\}$. 
	Otherwise, we would have $\hat{\delta}(0,u) = \myUnderbar{n-1}$ or $\hat{\delta}(q_0,u) = \myUnderbar{n-1}$ with $u$ having letter $0$ at every even position. With the same argument as above this leads to a contradiction, since for $u' \in \Sigma^*$ with $f(u') = u$ we would have $\delta(0,u') = n-1$ or $\delta(q_0,u') = n-1$.
	
	With $\hat{\delta}(0,u) \in \myUnderbar{Q}\setminus\{\myUnderbar{n-1}\}$ and $\hat{\delta}(q_0,u) \in \myUnderbar{Q}\setminus\{\myUnderbar{n-1}\}$ we have $\hat{\delta}(0,u1) = p_0 = \hat{\delta}(q_0,u1)$ and thus $\hat{\delta}(q_0,w) = \hat{\delta}(q_0,u1v) = \hat{\delta}(0,u1v) = \hat{\delta}(0,w) = z_+$.
	
	We have shown that $\hat{\delta}(0,w) = z_+$ implies $\hat{\delta}(q_0,w) = z_+$.
	
	Second, we assume $\hat{\delta}(q_0,w) = z_+$. We have to prove $\hat{\delta}(0,w) = z_+$. Note that this case is symmetrical to the first cast, since above we did not use any specifics of state $0$ that do not hold for state $q_0$. In particular, we have $q_0 \neq z_+$ and thus $w \neq \varepsilon$, and the state $n-1$ is unreachable from $q_0$ in $\mathcal{A}$. Since this case is symmetrical to the first case, it follows that $\hat{\delta}(q_0,w) = z_+$ implies $\hat{\delta}(0,w) = z_+$.
	
	Then $\hat{\delta}(0,w) = z_+$ holds iff $\hat{\delta}(q_0,w) = z_+$. Therefore, the states $0$ and $q_0$ are equivalent. This implies that $\hat{\mathcal{A}}$ is not minimal.
	
	As outlined above, the non-minimality of $\hat{\mathcal{A}}$ implies the S-compositionality of $\hat{\mathcal{A}}$. Therefore, we have shown that $\hat{\mathcal{A}}$ is not S-prime if $t$ is not reachable from $s$ in $G$.
	
	In total, we have shown that $\hat{\mathcal{A}}$ is S-prime iff $t$ is reachable from $s$ in $G$. Since $\hat{\mathcal{A}}$ can clearly be constructed in logarithmic space, we have found an \complexityClassFont{L}-reduction from \problemFont{2STCON} to \sPrimeDFA{2}, thus proving the \complexityClassFont{NL}-hardness of \sPrimeDFA{2}. This trivially implies the \complexityClassFont{NL}-hardness of \sPrimeDFA{} and \sPrimeDFA{k} for $k \in \natNumGeq{2}$.
	
	We have shown that \sPrimeDFA{} and its restrictions \sPrimeDFA{k} for $k \in \natNumGeq{2}$ are in \complexityClassFont{ExpSpace} and are \complexityClassFont{NL}-hard. The proof of \cref{the:SPrimeDFAComplexity} is complete.
\end{proof}

Finally, we prove:
\thePrimeDFAComplexity*
\begin{proof}[Proof of \cref{the:PrimeDFAComplexity}]
	The problem \primeDFA{} is in \complexityClassFont{ExpSpace} with \cite[Theorem 2.4]{DBLP:journals/iandc/KupfermanM15}. This implies that the restrictions \primeDFA{k} for $k \in \natNumGeq{2}$ are in \complexityClassFont{ExpSpace} as well. Therefore, we only have to concern ourselves with the lower complexity boundary.
	
	We begin by introducing another problem, which we will use in the \complexityClassFont{L}-reduction to establish the lower boundary. 
	In the proof of \cref{lem:fl_PrimeDFAFinNLHard} we introduced the emptiness problem for DFAs, denoted with \emptyDFA{}, which is known to be \complexityClassFont{NL}-complete \cite{DBLP:journals/jcss/Jones75}. Now we introduce $\emptyDFAq{2}$, which denotes the restriction of \emptyDFA{} to DFAs with at most two letters which have exactly one accepting state, which is an accepting sink.
	
	First, note that, since \emptyDFA{} is in \complexityClassFont{NL}, the restriction $\emptyDFAq{2}$ is in \complexityClassFont{NL} as well. Further, the standard \complexityClassFont{L}-reduction of \problemFont{2STCON} to \emptyDFA{} employs a DFA with at most two letters and exactly one accepting state. This state can be made into an accepting sink while preserving the validity of the reduction. Therefore, the problem \problemFont{2STCON} can be \complexityClassFont{L}-reduced to $\emptyDFAq{2}$, which implies the \complexityClassFont{NL}-hardness of $\emptyDFAq{2}$. Thus, the restriction $\emptyDFAq{2}$ is \complexityClassFont{NL}-complete.
	
	Now we return to the problem \primeDFA{} and its restrictions \primeDFA{k} for $k \in \natNumGeq{2}$.
	
	The problem \primeDFA{} is \complexityClassFont{NL}-hard with \cite[Theorem 2.5]{DBLP:journals/iandc/KupfermanM15}. But since the DFA construction used for the \complexityClassFont{L}-reduction of \emptyDFA{} to \primeDFA{} introduces an additional letter, we cannot use it to prove the \complexityClassFont{NL}-hardness of \primeDFA{2}. Instead, we will give an \complexityClassFont{L}-reduction of $\emptyDFAq{2}$ to \primeDFA{2}.
	
	Let $\mathcal{A} = (Q,\Sigma,q_I,\delta,F)$ be an input for $\emptyDFAq{2}$. W.l.o.g. let $\Sigma = \{0,1\}$ and let $F = \{q_+\}$. That is, the state $q_+$ is the accepting sink of $\mathcal{A}$, which is the only accepting state of $\mathcal{A}$.
	
	We will construct a DFA $\hat{\mathcal{A}} = (\hat{Q},\Sigma,q_I,\hat{\delta},\hat{F})$ that is prime iff $\mathcal{A}$ recognizes the empty language.
	
	To do this, let $\mathcal{A}_6' = (Q_6',\Sigma,q_0',\delta_6',F_6')$ be the minimal DFA with $\lang{\mathcal{A}_6'} = \{w \in \Sigma^* \setDel |w|_1 \equiv 0 \mod 6\}$. Clearly, \cref{fig:A_6'} pictures the DFA $\mathcal{A}_6'$ correctly. It is equally clear that $\mathcal{A}_6'$ is composite, since $\lang{\mathcal{A}_6'} = \lang{\mathcal{A}_2'} \cap \lang{\mathcal{A}_3'}$, where $\mathcal{A}_2',\mathcal{A}_3'$ are analogous to $\mathcal{A}_6'$ but instead of modulo six they use modulo two and three, respectively.
	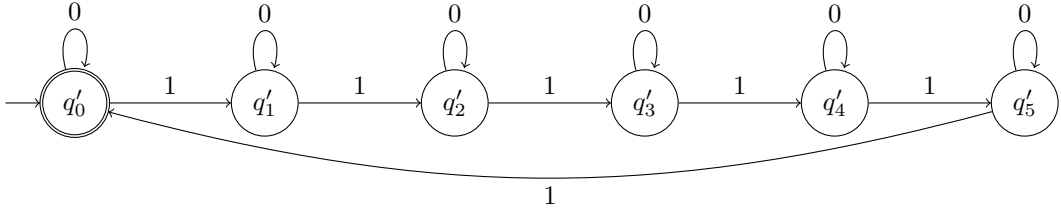
\begin{figure}[t]
		\centering
		\begin{tikzpicture}[node distance=2.5cm]
		\node[state, initial, accepting] 					(q0') 	{$q_0'$};
		\node[state, right of=q0'] 							(q1') 	{$q_1'$};
		\node[state, right of=q1'] 							(q2') 	{$q_2'$};
		\node[state, right of=q2'] 							(q3') 	{$q_3'$};
		\node[state, right of=q3']			 				(q4') 	{$q_4'$};
		\node[state, right of=q4']			 				(q5') 	{$q_5'$};
		
		\draw	(q0')	edge[loop above]						node{$0$}		(q0');
		\draw	(q0')	edge[above]								node{$1$}		(q1');
		
		\draw	(q1')	edge[loop above]						node{$0$}		(q1');
		\draw	(q1')	edge[above]								node{$1$}		(q2');
		
		\draw	(q2')	edge[loop above]						node{$0$}		(q2');
		\draw	(q2')	edge[above]								node{$1$}		(q3');
		
		\draw	(q3')	edge[loop above]						node{$0$}		(q3');
		\draw	(q3')	edge[above]								node{$1$}		(q4');
		
		\draw	(q4')	edge[loop above]						node{$0$}		(q4');
		\draw	(q4')	edge[above]								node{$1$}		(q5');
		
		\draw	(q5')	edge[loop above]						node{$0$}		(q5');
		\draw	(q5')	edge[bend left=15, below]					node{$1$}		(q0');
		\end{tikzpicture}
		\caption{DFA $\mathcal{A}_6'$ with $\lang{\mathcal{A}_6'} = \{w \in \Sigma^* \setDel |w|_1 \equiv 0 \mod 6\}$.}
		\label{fig:A_6'}
	\end{figure}

	Now we construct the DFA $\hat{\mathcal{A}}$ out of $\mathcal{A}$ using $\mathcal{A}_6'$. Define $\hat{Q} = Q \cup Q_6'$. We retain $q_I$ as the initial state and set $q_0'$, the only accepting state of $\mathcal{A}_6'$, as the only accepting state of $\hat{\mathcal{A}}$. That is, $\hat{F} = \{q_0'\}$. We keep the alphabet $\Sigma = \{0,1\}$ unaltered. Finally, for each $q \in \hat{Q}$ and $\sigma \in \Sigma$ we define:
	\begin{align*}
		\hat{\delta}(q,\sigma)=\begin{cases}
			q_0'	&\text{ if $q = q_+$ and $\sigma=0$} \\
			q_+		&\text{ if $q = q_+$ and $\sigma=1$} \\
			\delta(q,\sigma) &\text{ if $q \in Q \setminus \{q_+\}$}\\
			\delta'(q,\sigma) &\text{ else, thus if $q \in Q'$}
		\end{cases}.
	\end{align*}

	We need to show that $\hat{\mathcal{A}}$ is prime iff $\mathcal{A}$ recognizes the empty language. 
	
	If $\mathcal{A}$ recognizes the empty language then the accepting sink $q_+$ is unreachable in $\mathcal{A}$. Obviously, this implies that $q_+$ is unreachable in $\hat{\mathcal{A}}$ as well. Since $q_0'$, the only accepting state in $\hat{\mathcal{A}}$, is only reachable from the initial state $q_I$ via $q_+$, this implies that $\hat{\mathcal{A}}$ recognizes the empty language. Thus, the DFA $\hat{\mathcal{A}}$ is prime. We have shown that the DFA $\hat{\mathcal{A}}$ is prime if $\mathcal{A}$ recognizes the empty language. This part of the equivalence is done.
	
	Next, we assume that $\mathcal{A}$ does not recognize the empty language. We have to show that $\hat{\mathcal{A}}$ is not prime.
	
	We begin by constructing the minimal DFA $\hat{\mathcal{A}}^!$ recognizing $\lang{\hat{\mathcal{A}}}$. Then we prove that $\hat{\mathcal{A}}^!$ is composite, which implies that $\hat{\mathcal{A}}$ is composite as well.
	
	Let $\mathcal{A}^! = (Q^!,\Sigma,q_I^!,\delta^!,F^!)$ be the minimal DFA recognizing $\lang{\mathcal{A}}$. Obviously, the DFA $\mathcal{A}^!$ has exactly one accepting state, which is an accepting sink. Let $q_+^!$ be this accepting sink.
	Then construct $\hat{\mathcal{A}}^! = (\hat{Q}^!,\Sigma,q_I^!,\hat{\delta}^!,\hat{F}^!)$ out of $\mathcal{A}^!$ using $\mathcal{A}_6'$ analogously to the construction of $\hat{\mathcal{A}}$ out of $\mathcal{A}$ using $\mathcal{A}_6'$. That is, redirect the $0$-self-loop of state $q_+^!$ of $\mathcal{A}^!$ to state $q_0'$ of $\mathcal{A}_6'$, make $q_I^!$ the initial state, and make $q_0'$ the sole accepting state. Clearly, we have $\lang{\hat{\mathcal{A}}^!} = \lang{\hat{\mathcal{A}}}$ and every state in $\hat{\mathcal{A}}^!$ is reachable.
	
	Now we prove that $\hat{\mathcal{A}}^!$ is the minimal DFA recognizing $\lang{\hat{\mathcal{A}}}$. To do this, we only have to prove the minimality of $\hat{\mathcal{A}}^!$. Let $\mathcal{B} = (S,\Sigma,s_I,\eta,G)$ be a DFA with $\lang{\hat{\mathcal{A}}^!} \subseteq \lang{\mathcal{B}}$ and $\size{\mathcal{B}} < \size{\hat{\mathcal{A}}^!}$. We prove $\lang{\hat{\mathcal{A}}^!} \subset \lang{\mathcal{B}}$, which implies the minimality of $\hat{\mathcal{A}}^!$.
	
	Since every state in $\hat{\mathcal{A}}^!$ is reachable, there exist $w,w' \in \Sigma^*$ with $\hat{\delta}^!(q_I^!,w) \neq \hat{\delta}^!(q_I^!,w')$ and $\eta(s_I,w) = \eta(s_I,w')$. Let $q = \hat{\delta}^!(q_I^!,w), q' = \hat{\delta}^!(q_I^!,w')$.
	\begin{description}
		\item[Case 1: \normalfont{$q,q' \in Q^!$.}] Since $\mathcal{A}^!$ is minimal, we have $\lang{{\mathcal{A}^!}^{q}} \neq \lang{{\mathcal{A}^!}^{q'}}$. W.l.o.g. let $\lang{{\mathcal{A}^!}^{q}} \not\subseteq \lang{{\mathcal{A}^!}^{q'}}$ and let $u \in \lang{{\mathcal{A}^!}^{q}} \setminus \lang{{\mathcal{A}^!}^{q'}}$. Then we have $wu \in \lang{\mathcal{A}^!}, w'u \notin \lang{\mathcal{A}^!}$ and therefore $wu0 \in \lang{\hat{\mathcal{A}}^!},w'u0 \notin \lang{\hat{\mathcal{A}}^!}$. But we also have $\eta(s_I,w'u0) = \eta(\eta(s_I,w'),u0) = \eta(\eta(s_I,w),u0) = \eta(s_I,wu0)$. With $wu0 \in \lang{\hat{\mathcal{A}}^!} \subseteq \lang{\mathcal{B}}$ this implies $w'u0 \in \lang{\mathcal{B}}$.
		
		With witness $w'u0$ we are done with Case 1.
		
		\item[Case 2: \normalfont{$q,q' \in Q_6'$.}] Let $i,j \in \{0,\dots,5\}, i \neq j$ with $q = q_i', q' = q_j'$. Then we have $\hat{\delta}^!(q_i',1^{6-i}) = q_0'$ and $\hat{\delta}^!(q_j',1^{6-i}) \neq q_0'$. Therefore, we have $w1^{6-i} \in \lang{\hat{\mathcal{A}}^!}, w'1^{6-i} \notin \lang{\hat{\mathcal{A}}^!}$. But we also have $\eta(s_I,w'1^{6-i}) = \eta(\eta(s_I,w'),1^{6-i}) = \eta(\eta(s_I,w),1^{6-i}) = \eta(s_I,w1^{6-i})$. With $w1^{6-i} \in \lang{\hat{\mathcal{A}}^!} \subseteq \lang{\mathcal{B}}$ this implies $w'1^{6-i} \in \lang{\mathcal{B}}$.
		
		With witness $w'1^{6-i}$ we are done with Case 2.
		
		\item[Case 3: \normalfont{$q \in Q^! \Leftrightarrow q' \in Q_6'$.}] W.l.o.g. let $q \in Q_6', q' \in Q^!$. Let $i \in \{0,\dots,5\}$ with $q = q_i'$. Then $\hat{\delta}^!(q_i',1^{6-i}) = q_0'$ holds. Since the only transition connecting the states in $Q^!$ with the states in $Q_6'$ is the $0$-Transition from $q_+^!$ to $q_0'$, we have $\hat{\delta}^!(q',1^{6-i}) \in Q^!$. Therefore, we have $w1^{6-i} \in \lang{\hat{\mathcal{A}}^!}, w'1^{6-i} \notin \lang{\hat{\mathcal{A}}^!}$. But analogously to Case 2 we also have $w'1^{6-i} \in \lang{\mathcal{B}}$.
		
		With witness $w'1^{6-i}$ we are done with Case 3.
	\end{description}
	With Cases 1-3 we have $\lang{\hat{\mathcal{A}}^!} \subset \lang{\mathcal{B}}$, which implies the minimality of $\hat{\mathcal{A}}^!$.
	
	We have shown that $\hat{\mathcal{A}}^!$ is the minimal DFA recognizing $\lang{\hat{\mathcal{A}}}$. Now we prove the compositionality of $\hat{\mathcal{A}}^!$, which implies the compositionality of $\hat{\mathcal{A}}$ as well.
	
	The compositionality of $\hat{\mathcal{A}}^!$ is easy to see. As outlined above, the DFA $\mathcal{A}_6'$ is composite with $\lang{\mathcal{A}_6'} = \lang{\mathcal{A}_2'} \cap \lang{\mathcal{A}_3'}$. We can then construct the DFAs $\hat{\mathcal{A}}_2^!$ and $\hat{\mathcal{A}}_3^!$ out of $\mathcal{A}^!$ using $\mathcal{A}_2'$ and $\mathcal{A}_3'$ respectively in the same way we constructed $\hat{\mathcal{A}}^!$ out of $\mathcal{A}^!$ using $\mathcal{A}_6'$. We obviously have $\lang{\hat{\mathcal{A}}^!} = \lang{\hat{\mathcal{A}}_2^!} \cap \lang{\hat{\mathcal{A}}_3^!}$ and $\size{\hat{\mathcal{A}}_2^!} = \size{\mathcal{A}^!} + \size{\mathcal{A}_2'} = \size{\mathcal{A}^!} + 2 < \size{\mathcal{A}^!} + 6 = \size{\mathcal{A}^!} + \size{\mathcal{A}_6'} = \ind{\hat{\mathcal{A}}^!}$ as well as $\size{\hat{\mathcal{A}}_3^!} = \size{\mathcal{A}^!} + \size{\mathcal{A}_3'} = \size{\mathcal{A}^!} + 3 < \size{\mathcal{A}^!} + 6 = \size{\mathcal{A}^!} + \size{\mathcal{A}_6'} = \ind{\hat{\mathcal{A}}^!}$. Thus, the DFA $\hat{\mathcal{A}}^!$ is composite, which implies the compositionality of $\hat{\mathcal{A}}$.
	
	We have shown that $\hat{\mathcal{A}}$ is composite if $\mathcal{A}$ does not recognize the empty language.
	
	In conclusion, we have shown that $\hat{\mathcal{A}}$ is prime iff $\mathcal{A}$ recognizes the empty language. Since $\hat{\mathcal{A}}$ can obviously be constructed in logarithmic space, we have found an \complexityClassFont{L}-reduction of $\emptyDFAq{2}$ to \primeDFA{2}. Thus, we have shown the \complexityClassFont{NL}-hardness of \primeDFA{2}. This obviously implies the \complexityClassFont{NL}-hardness of \primeDFA{} and its restrictions \primeDFA{k} for $k \in \natNumGeq{2}$.
	
	We have shown that \primeDFA{} and its restrictions \primeDFA{k} for $k \in \natNumGeq{2}$ are in \complexityClassFont{ExpSpace} and are \complexityClassFont{NL}-hard. The proof of \cref{the:PrimeDFAComplexity} is complete.
\end{proof}

This ends our discussion of the proofs for \cref{sec:2DFAMinimalAndSPrimeDFA}.
\end{document}